\documentclass[11pt, a4paper]{article} 
\usepackage[english]{babel}
\usepackage[utf8]{inputenc}
\usepackage[a4paper]{geometry}

\usepackage{color}
\usepackage{graphicx}

\usepackage{mathrsfs}

\usepackage{float}

\usepackage{bbm}
\usepackage{enumerate}
\usepackage{tikz}
\usetikzlibrary{matrix}
\usepackage{pgfplots}
\usetikzlibrary{arrows.meta}

\usepackage[colorlinks=true,linkcolor=black,citecolor=black]{hyperref}

\usepackage[symbol]{footmisc}

\usepackage{amsmath}
\usepackage{amssymb}
\usepackage{amsthm}
\usepackage{stmaryrd}

\DeclareMathOperator*{\esssup}{ess\,sup}
\DeclareMathOperator*{\essinf}{ess\,inf}

\newcommand{\F}{\mathcal{F}}
\newcommand{\PM}{\mathbb{P}}

\newcommand{\E}{\mathbb{E}}
\newcommand{\bP}{\mathbf{b}\mathcal{P}}

\newsavebox\dotbox
\sbox{\dotbox}{\(\displaystyle\bigodot\)}
\newcommand{\bigcdot}{\stackrel{\mbox{\tiny$\bullet$}}{}}
\newcommand{\Vl}{V^{\mathrm{liq}}}

\newcommand{\beao}{\begin{eqnarray*}}
\newcommand{\eeao}{\end{eqnarray*}\noindent}

\newcommand{\beam}{\begin{eqnarray}}
\newcommand{\eeam}{\end{eqnarray}\noindent}

\def\bbr{{\Bbb R}}   
\def\bbn{{\Bbb N}}

\def\bbq{{\Bbb Q}}

\newcommand{\eps}{{\varepsilon}}
\newcommand{\vp}{{\varphi}}

\newcommand{\Var}{{\rm Var}}

\newcommand{\ov}{\overline}
\newcommand{\un}{\underline}
\newcommand{\wh}{\widehat}
\newcommand{\wt}{\widetilde}

\newcommand{\ph}{\varphi}

\newcommand{\mal}{\stackrel{\mbox{\tiny$\bullet$}}{}}
\newcommand{\auf}{[\![}
\newcommand{\zu}{]\!]}

\DeclareMathOperator*{\uplim}{up-lim}

\newtheorem{theorem}{Theorem}[section]
\newtheorem{lemma}[theorem]{Lemma}
\newtheorem{corollary}[theorem]{Corollary}
\newtheorem{proposition}[theorem]{Proposition}

\newtheorem{note}[theorem]{Note}
\newtheorem{remark}[theorem]{Remark}
\newtheorem{example}[theorem]{Example}
\newtheorem{assumption}[theorem]{Assumption}
\theoremstyle{definition}
\newtheorem{definition}[theorem]{Definition}

\numberwithin{equation}{section}
\begin{document}
\title{Semimartingale price systems in models with transaction costs
beyond efficient friction\thanks{We would like to thank Christoph Czichowsky, Miryana Grigorova, Yuri Kabanov, and Martin Schweizer for valuable comments and fruitful discussions. 
Valuable comments and suggestions of two anonymous referees are greatly appreciated.}}
\author{Christoph Kühn\thanks{Institut für Mathematik, Goethe-Universität Frankfurt, D-60054 Frankfurt a.M., Germany, e-mail: \{ckuehn,
			molitor\}@math.uni-frankfurt.de}\and Alexander Molitor\footnotemark[2]}
\date{}

\maketitle

\begin{abstract}
	A standing assumption in the literature on proportional transaction costs is
	efficient  friction. Together with robust no free lunch with vanishing 
	risk, it rules out
	strategies of infinite variation, as they usually appear in
	frictionless markets. In this paper, we show how the models with and
	without transaction costs can be unified.
	
	The bid and the ask price of a risky asset are given by
	c\`adl\`ag processes which are locally bounded from below and may coincide at some points. In a first step, we show that if the bid-ask model satisfies ``no unbounded profit
	with bounded risk'' for simple strategies, then there exists a
	semimartingale lying between the bid and the ask price process.
	
	In a second step, under the additional assumption that the
	zeros of the bid-ask spread are either starting points of an excursion away from zero or inner points from the right, we show that for every bounded predictable strategy specifying the amount of risky assets, the semimartingale can be used to construct the corresponding self-financing risk-free position in a consistent way. Finally, the set of most general strategies is introduced, which also provides a new view on the frictionless case.
\end{abstract}

\begin{tabbing}
	semimartingales, strategies of infinite variation, stochastic integration\\ 
	
\end{tabbing}
\section{Introduction}
In frictionless markets, asset price processes have to be semimartingales unless they allow for an ``unbounded profit with bounded risk''~(UPBR) with simple strategies
(see Delbaen and Schachermayer~\cite{delbaen1994general}). With semimartingale price processes, the powerful tools of stochastic calculus can be used to construct the gains from dynamic trading. A trading strategy specifying the amounts of shares an investor holds in her portfolio is a predictable process that is integrable w.r.t. the 
vector-valued price process. 
Strategies can be of infinite variation since in the underlying limiting procedure, one directly considers the (book) profits made rather than the portfolio rebalancings.

On the other hand, under arbitrary small transaction costs also 
non-semimartingales can lead to markets without ``approximate arbitrage opportunities''. Guasoni~\cite{guasoni.2006}  and Guasoni, R\'asonyi, and Schachermayer~\cite{guasoni.r.s.2008} derive the sufficient condition of ``conditional full support'' of the mid-price process, that is satisfied, e.g., by a fractional Brownian motion, and arbitrary small constant proportional costs. Guasoni, R\'asonyi, and Schachermayer~\cite{guasoni.r.s.2010} derive a fundamental theorem of asset pricing for a family of transaction costs models.  

Under the assumptions of efficient friction, i.e., nonvanishing bid-ask spreads, and the existence of a strictly consistent price system, Kabanov and Stricker~\cite{kabanov.stricker.2002} and Campi and Schachermayer~\cite{campi.schachermayer} show for continuous and c\`adl\`ag processes, respectively, that a finite credit line implies that the variation of the trading strategies is bounded in probability. A similar assertion is shown in Guasoni, L{\'e}pinette, and R{\'a}sonyi~\cite{guasoni2012fundamental} under the condition of ``robust no free lunch with vanishing risk''. An important consequence for hedging and portfolio optimization is that the set of portfolios that are attainable with strategies of finite variation is Fatou-closed. For a detailed discussion, we refer to the monograph of Kabanov and Safarian~\cite{kabanov.safarian.2009}.

In this paper, we consider c\`adl\`ag bid and ask price processes that are not necessarily different. The ask price is bigger or equal to the bid price. The spread, which models the transaction costs, can vary in time and can even vanish.  
The contribution of this paper is twofold. First, we show that if the bid-ask model 
satisfies ``no unbounded profit with bounded risk'' (NUPBR) for simple long-only strategies,  then there exists a semimartingale lying between the bid and the ask price process.
This generalizes Theorem~7.2 of Delbaen and Schachermayer~\cite{delbaen1994general}
for the frictionless case. The proof in \cite{delbaen1994general} is very intuitive. Roughly speaking, it first shows that an explosion of the quadratic increments of the price process along stopping times would lead to an (UPBR). Then, it considers a discrete time Doob decomposition of the asset price process and shows that an explosion of the drift part as the mesh of the grid tends to zero would lead to an (UPBR). 
This already yields that under (NUPBR), the asset price process has to be a good integrator and thus a semimartingale by the Bichteler-Dellacherie theorem.
More recently, Beiglb{\"o}ck, Schachermayer, and Veliyev \cite{beiglbock2011direct}
provide an alternative proof of the Bichteler-Dellacherie theorem combining these no-arbitrage arguments with Koml\'os type arguments. Kardaras and Platen~\cite{kardaras.platen} follow a quite different approach that only requires long investments. They construct supermartingale deflators as dual variables in suitable utility maximization problems under a variation of (NUPBR) for simple long-only strategies. B\'alint and Schweizer~\cite{balint.schweizer} assume that asset prices are expressed in a possibly nontradable accounting unit. In their setting there need not exist an asset with a strictly positive price process that can be used as a numéraire. They show that if there exists a portfolio with strictly positive value process then, under a discounting invariant form of absence of arbitrage, which generalizes the condition used in \cite{kardaras.platen}, the asset prices discounted by the portfolio value are semimartingales. Since in transaction costs models it is natural to start with the relative prices of the tradable assets, there is no obvious analogy of discounting by a portfolio value. In our model, we implicitly assume the existence of an asset with strictly positive price process that serves as a reference asset.  

In the bid-ask model, we consider a Dynkin zero-sum stopping game in which the lower payoff process is the bid price and the upper payoff process the ask price. 
The Doob decomposition of the dynamic value of the discrete time game along 
arbitrarily fine grids is used to identify smart investment opportunities. 
The crucial point is that the drift of the Dynkin value can be earned by trading in the bid-ask market. This we combine with the
brilliant idea in Lemma~4.7 of \cite{delbaen1994general} to control the martingale part. We complete the proof by showing that under the assumptions above, the continuous time Dynkin value has to be a local quasimartingale.

In the second part of the paper, we show how a semimartingale between the bid and the ask process can be used to define the self-financing condition of the model beyond efficient friction.
Without efficient friction, strategies of infinite variation 
can make sense since they do not produce infinite trading costs. This of course means that we cannot use them as integrators without major hesitation.
In the first step, we only consider bounded amounts of risky assets. Thus, the trading gains charged in the semimartingale are finite. Then, we add the costs caused by the fact that the trades are carried out at the less favorable bid-ask prices. Roughly speaking, if the spread is away from zero the costs are a Riemann-Stieltjes integral similar to Guasoni, L{\'e}pinette, and R{\'a}sonyi~\cite{guasoni2012fundamental}. Then, we exhaust the costs when the spread is away from zero. The crucial point is that these costs are always nonnegative, and the semimartingale gains are finite. Especially, infinite costs cannot be compensated and lead to ruin.
Under a rather mild additional assumption on the behavior of the spread at zero (see Assumption~\ref{8.12.2019.3}),
that goes at least far beyond the frictionless case and the case of efficient friction, this approach leads to a well-founded self-financing condition.
Especially, the self-financing risk-less position does not depend on the choice of the semimartingale we use in the construction (see Corollary~\ref{cor:Independence}).

A self-financing condition for general strategies has to be justified by {\em suitable} approximations with simple strategies. With 
transaction costs, this is a delicate issue. Namely, under pointwise convergence of the strategies alone, one should not expect that portfolio processes converge. By the strict Fatou-type inequality (see Theorem~A.9(iv) of \cite{guasoni2012fundamental}), some variation/costs can disappear in the limit. Thus, roughly speaking, we postulate the following: first, the limit strategy is better than all (almost) pointwise converging simple strategies and second, for each strategy there exists a special sequence of approximating simple strategies s.t. the wealth processes converge (see Theorem~\ref{8.12.2019.1}).  

In the second step, we extend the self-financing condition from the bounded strategies to the maximal set of strategies for which it can be defined in a ``reasonable'' way. In the special case
of a frictionless market, this maximal set coincides with the set of predictable processes which are integrable w.r.t. the semimartingale price process in the classic sense (see, e.g., \cite{jacod.shiryaev}). 
Thus, we also provide a further characterization of this ubiquitous set. 

In the no-arbitrage theory, the need for general strategies is already proven in the special case of frictionless markets. Indeed, Delbaen and Schachermayer~\cite[Lemma 7.9 and Lemma 7.10]{delbaen1994general} provide an example with a bounded asset price process showing that no free lunch with vanishing risk (NFLVR) for simple strategies does not imply the existence of an equivalent martingale measure (EMM). Consequently, under transaction costs general strategies can become an important tool to guarantee the existence of a consistent price system (CPS), which plays a similar role as an EMM in the frictionless theory, under an appropriate no-arbitrage condition. On the other hand, in general a CPS does not exist even though (NFLVR) for multivariate portfolio processes is satisfied. This can already be seen in discrete time (see Schachermayer~\cite[Example 3.1]{schachermayer2004fundamental}) with the observation that general strategies as described in Definition~\ref{def:DefinitionL}
coincide with simple strategies if the time is discrete.

In a nutshell, we provide a well-founded self-financing condition for models beyond efficient friction by relating the original trading gains under transaction costs with the gains 
in a fictitious frictionless market defined by a semimartingale and subtracting the appropriate costs. The idea to relate markets under transaction costs with fictitious frictionless 
markets is not new. It is already widely used in the theory of portfolio optimization. Here, shadow price processes, i.e., fictitious frictionless pricing systems that lead to the 
same \emph{optimal} decisions and trading gains as under transaction costs, are utilized to determine optimal trading strategies. The existence of shadow prices and their relationship 
with a suitable dual problem goes back to Cvitani{\'c} and Karatzas~\cite{CvitanicKaratzas1996}. In discrete time, Kallsen and Muhle-Karbe~\cite{KallsenMuhleKarbe2011} show that 
on finite probability spaces shadow price processes always exist as long as the original problem has a solution, and Czichowsky et al.~\cite{CzichowskyMuhleKarbeSchachermayer2014} 
provide counterexamples on infinite probability spaces. Conditions for the existence of a shadow price process in a semimartingale model are established by 
Czichowsky et al.~\cite{CzichowskySchachermayerYang2017} and starting with Kallsen and Muhle-Karbe~\cite{KallsenMuhleKarbe2010} various explicit constructions of shadow prices 
processes have been given in Black-Scholes type models. Even in non-semimartingale models this \emph{dual} approach is successfully applied 
(see, e.g., \cite{CzichowskyShachermayer2016,CzichowskyShachermayer2017,CzichowskyPeyreSchachermayerYang2018}) under efficient friction. 
In the proof of Theorem~\ref{15.7.2021.1}, 
we provide a direct connection between our work and shadow price processes for particular optimization problems.

The paper is organized as follows. In Section~\ref{21.12.2019.01}, we show the existence of a semimartingale price system (Theorem~\ref{theo:ResultSemimartingale}). 
In Section~\ref{21.12.2019.02}, we construct the cost process which allows us to introduce the self-financing condition for bounded strategies, which is justified by 
Theorem~\ref{8.12.2019.1} and Corollary~\ref{cor:Independence}. In Section~\ref{15.7.2021.2},
the extension to unbounded strategies is established (Proposition~\ref{prop:Welldefined}). In addition, the special case of a 
frictionless market is considered (Proposition~\ref{prop:FrictionlessL}) and the separate convergence of trading gains and cost terms 
of the approximating bounded strategies is discussed (Theorem~\ref{15.7.2021.1}). 
Technical proofs are postponed to Section~\ref{4.1.2020.1} and Appendix~\ref{app:TechnicalCosts}.
\section{Existence of a semimartingale price system}\label{21.12.2019.01}
Throughout the paper, we fix a terminal time~$T\in\bbr_+$ and a filtered 
probability\linebreak space~$(\Omega,\F, (\F_t)_{t\in[0,T]},\PM)$ satisfying the usual conditions. 
The predictable $\sigma$-algebra on $\Omega\times[0,T]$ is denoted by $\mathcal{P}$, 
the set of bounded predictable processes starting at zero by $\bP$.
To simplify the notation, a stopping time~$\tau$ is allowed to take the value $\infty$, 
but $\auf\tau\zu:=\{(\omega,t)\in\Omega\times[0,T]: t=\tau(\omega)\}$. Especially,
we use the notation $\tau_A$, $A\in\mathcal{F}_\tau$, for the stopping time that
coincides with $\tau$ on $A$ and is infinite otherwise.
$\Var_a^b(X)$ denotes the pathwise variation of a process~$X$ on the interval $[a,b]$.
A process~$X$ is called l\`agl\`ad iff all paths possess finite left and right limits (but they can have double jumps).
We set $\Delta^+ X:=X_+-X$ and $\Delta X:=\Delta^- X:=X-X_-$, where $X_{t+}:=\lim_{s\downarrow t} X_s$ and $X_{t-}:=\lim_{s\uparrow t} X_s$. For a random variable~$Y$, we set $Y^+:=\max(Y,0)$ and $Y^-:=\max(-Y,0)$.\\

The financial market consists of one risk-free bond with price~$1$ and one risky asset with bid price $\underline{S}$ and ask price $\overline{S}$. 
Throughout the paper, we make the following assumption.
\begin{assumption}\label{ass:ProcessProperties}
	$(\underline{S}_t)_{t\in[0,T]}$ and $(\overline{S}_t)_{t\in[0,T]}$ are adapted processes with c\`{a}dl\`{a}g paths. In addition, $\underline{S}_t\leq \overline{S}_t$ for all $t\in[0,T]$ and $\underline{S}$ is locally bounded from below. 
\end{assumption}

In this section, we only consider simple trading strategies in the following sense.
\begin{definition}
	A \emph{simple trading strategy}  is a  stochastic process $(\varphi_t)_{t\in[0,T]}$ of the form
	\begin{align}\label{eq:varphi1}
		\varphi=\sum_{i=1}^{n}Z_{i-1}\mathbbm{1}_{\rrbracket T_{i-1}, T_i\rrbracket}, 
	\end{align} 
	where $n\in\mathbb{N}$ is a finite number, $0=T_0\leq T_1\leq \dots \leq T_n=T$ is an increasing sequence of stopping times and $Z_i$ is $\F_{T_i}$-measurable
	for all $i=0,\dots,n-1$.
\end{definition}

The strategy $\vp$ specifies the amount of risky assets in the portfolio.
The next definition corresponds to the self-financing condition of the model. It specifies the holdings in the risk-free bond given a simple trading strategy.
\begin{definition}
	Let $(\varphi_t)_{t\in[0,T]}$ be a simple trading strategy. The \emph{corresponding position in the risk-free bond} $(\varphi^0_t)_{t\in[0,T]}$ is given by
	\begin{align} \begin{aligned}
			\label{eq:varphi0}
			\varphi^0_t:=\sum_{0\leq s<t}\left(\underline{S}_s(\Delta^+\varphi_s)^--\overline{S}_s(\Delta^+\varphi_s)^+\right), \quad t\in[0,T].
		\end{aligned}
	\end{align}
\end{definition}
\begin{definition}
	Let $(\varphi_t)_{t\in[0,T]}$ be a simple trading strategy. The \emph{liquidation value process $(\Vl_t(\varphi))_{t\in[0,T]}$} is given by
	\begin{align}
		\label{eq:DefLiquidationValue}
		\Vl_t(\varphi):=\varphi^0_t+(\varphi_t)^+\underline{S}_t-(\varphi_t)^-\overline{S}_t, \quad t\in[0,T].
	\end{align}
	If it is clear from the context, we write $(\Vl_t)_{t\in[0,T]}$ instead of $(\Vl_t(\varphi))_{t\in[0,T]}$.  
\end{definition}

We adapt the notion of an unbounded profit with bounded risk (UPBR) from Bayraktar and Yu~\cite{bayraktar2018market} to the present setting of simple long-only trading strategies.
\begin{definition}\label{def:UPBR}
	We say that $(\underline{S}_t,\overline{S}_t)_{t\in[0,T]}$ admits an \emph{unbounded profit with bounded risk (UPBR) for simple long-only strategies} if there exists a sequence of simple  trading strategies $(\varphi^{n})_{n\in\mathbb{N}}$ with $\varphi^{n}\ge 0$ s.t.  
	\begin{enumerate}[(i)]
		\item $\Vl_t(\varphi^{n})\geq -1$ for all $t\in[0,T]$ and $n\in\bbn$,
		\item The sequence $(\Vl_T(\varphi^{n}))_{n\in\mathbb{N}}$ is unbounded in probability, i.e., 
		\begin{align}
			\label{eq:DefUPBR}
			\lim\limits_{m\to\infty}\sup_{n\in\mathbb{N}}\PM\left(\Vl_T(\varphi^{n})\geq m\right)>0.
		\end{align}
	\end{enumerate}
	If no such sequence exists, we say that the bid-ask process $(\underline{S},\overline{S})$ satisfies the \emph{no unbounded profit with bounded risk (NUPBR) condition for simple long-only strategies}. 
\end{definition}
\begin{remark}
	The admissibility condition~(i) is rather restrictive, e.g., compared to \cite{guasoni2012fundamental}, see Definition~4.4. therein, which means that the present version of (NUPBR) is a weak condition. But, for the following first main result of the paper, it is already sufficient. 
\end{remark}
\begin{theorem}\label{theo:ResultSemimartingale}
	Let $(\underline{S}_t,\overline{S}_t)_{t\in[0,T]}$ satisfy 
	Assumption~\ref{ass:ProcessProperties} and	
	the (NUPBR) condition for simple long-only strategies. Then, there exists a semimartingale $S=(S_t)_{t\in[0,T]}$ s.t. 
	\begin{align}\label{eq:InBetween}
		\underline{S}_t\leq S_t\leq \overline{S}_t\quad\text{for all}\ t\in[0,T]. 
	\end{align}
\end{theorem}

A semimartingale~$S$ satisfying (\ref{eq:InBetween}) we call a \emph{semimartingale price system}.	The remaining part of the section is devoted to the proof of Theorem~\ref{theo:ResultSemimartingale}. As a first step, we will show that it is actually sufficient to prove 
the following seemingly weaker version of the theorem. 
\begin{theorem}\label{theo:ExistenceSimple}
	Suppose that $0\leq \underline{S}\leq \overline{S}\le 1$, and that (NUPBR) for simple long-only strategies holds. Then there exists a semimartingale $S=(S_t)_{t\in[0,T]}$ s.t.  
	\begin{align*}
		\underline{S}_t\leq S_t\leq \overline{S}_t \quad\text{for all}\ t\in[0,T]. 
	\end{align*}
\end{theorem}
\begin{proposition}\label{prop: ReductionStep} Theorem~\ref{theo:ExistenceSimple} implies Theorem~\ref{theo:ResultSemimartingale}.
\end{proposition}
\begin{proof}
	We assume that Theorem~\ref{theo:ExistenceSimple} holds true. 
	
	\emph{Step 1:} Let $\underline{S}$ be locally bounded from below, $\ov{S}\le 1$, and $(\un{S},\ov{S})$ satisfies (NUPBR). Thus, there is an increasing sequence $(\sigma^n)_{n\in\mathbb{N}}$ of stopping times with $\mathbb{P}(\sigma^n=\infty)\to 1$ s.t. $\underline{S}\geq -n$ on $\llbracket 0,\sigma^n\rrbracket$ for all $n\in\mathbb{N}$. With $(\un{S},\ov{S})$, a fortiori  $((\un{S}^{\sigma^n}+n)/(n+1),(\ov{S}^{\sigma^n}+n)/(n+1))$ satisfies (NUPBR). By 
	Theorem~\ref{theo:ExistenceSimple}, there is a semimartingale $S^n$ for each $n\in\mathbb{N}$ s.t. $(\underline{S}^{\sigma^n}+n)/(n+1)\le S^n\le (\overline{S}^{\sigma^n}+n)/(n+1)$. Therefore, the process $S:=\sum_{n=1}^\infty \mathbbm{1}_{\llbracket {\sigma}^{n-1},{\sigma}^{n}\llbracket} ((n+1)S^{n}-n)$, where $\sigma^0:=0$, lies between $\underline{S}$ and $\overline{S}$. $S$ is a local semimartingale and, thus, a semimartingale. Consequently, Theorem~\ref{theo:ExistenceSimple} holds true under the milder condition that $\un{S}$ is only locally bounded from below instead of nonnegative.
	
	\textit{Step 2:} Let $\underline{S}$ be locally bounded from below and $(\un{S},\ov{S})$ satisfies (NUPBR) for simple long-only strategies.
	Consider the stopping times $\tau^n:=\inf\{t\geq0: \overline{S}_t > n\}$, $n\in\bbn$.
	One has that  $\mathbb{P}(\tau^n=\infty)= \mathbb{P}(\ov{S}_t\le n\ \forall t\in[0,T]) \to1$ as $n\to\infty$. 
	With short-selling constraints, liquidation value processes that are attainable in the market $((\underline{S}^{\tau^n}/n)\wedge 1,(\overline{S}^{\tau^n}/n)\wedge 1)$ can be dominated by those in $(\un{S},\ov{S})$. Indeed, for $t<\tau^n$, one has $(\ov{S}^{\tau^n}_t/n)\wedge 1 =\ov{S}_t/n$, and a purchase at time~$\tau^n$ cannot generate a profit in the market~$((\underline{S}^{\tau^n}/n)\wedge 1,(\overline{S}^{\tau^n}/n)\wedge 1)$.  
	Thus, $((\underline{S}^{\tau^n}/n)\wedge 1,(\overline{S}^{\tau^n}/n)\wedge 1)$ satisfies (NUPBR) with simple 
	{\em long-only} strategies and by Step~1 there exist 
	semimartingales~$S^n$ with $(\underline{S}^{\tau^n}/n)\wedge 1\le S^n\le (\overline{S}^{\tau^n}/n)\wedge 1$ for all $n\in\bbn$.
	Then, $S:=\sum_{n=1}^\infty \mathbbm{1}_{\llbracket {\tau}^{n-1},{\tau}^{n}\llbracket} n S^{n}$, where $\tau^0:=0$, shows the assertion. 
\end{proof}

For the remainder of the section, we work under the assumptions of Theorem~\ref{theo:ExistenceSimple}. More specifically we assume the following. 
\begin{assumption}\label{ass:bounded}
	We assume $0\leq \underline{S}\leq \overline{S}\leq 1$ and that $(\underline{S},\overline{S})$ satisfies (NUPBR) for simple long-only strategies for the remainder of the section.
\end{assumption}
In addition, we set w.l.o.g. $T=1$. We now proceed with the proof of Theorem~\ref{theo:ExistenceSimple}.
The candidate for the semimartingale will be the value process of a Dynkin zero-sum stopping game played on the bid and ask price, i.e., let $(S_t)_{t\in[0,1]}$ be the right-continuous version of
\begin{align}
	\begin{aligned}\label{eq:ContDynkin}
		S_t&:=\esssup_{\tau\in\mathcal{T}_{t,1}}\essinf_{\sigma\in\mathcal{T}_{t,1}}\E\left[\underline{S}_\tau\mathbbm{1}_{\{\tau\leq \sigma\}}+\overline{S}_\sigma \mathbbm{1}_{\{\tau >\sigma\}}\mid \F_t \right]\\&=\essinf_{\sigma\in\mathcal{T}_{t,1}}\esssup_{\tau\in\mathcal{T}_{t,1}}\E\left[\underline{S}_\tau\mathbbm{1}_{\{\tau\leq \sigma\}}+\overline{S}_\sigma \mathbbm{1}_{\{\tau >\sigma\}}\mid \F_t \right],
	\end{aligned}
\end{align}
where $\mathcal{T}_{t,1}$ is the set of $[t,1]$-valued stopping times for $t\in[0,1]$. The existence of such a process and the non-trivial equality in \eqref{eq:ContDynkin} is 
guaranteed by \emph{Th\'eor\`eme 7 \& 9 and Corollaire 12} in \cite{lepeltier1984jeu}. Obviously, $S=(S_t)_{t\in[0,1]}$ satisfies $\underline{S}\le S\le \overline{S}$. 
Thus, we only have to show that (NUPBR) for simple long-only trading strategies implies that $S$ is a semimartingale.
We note that all arguments remain valid for a different terminal value of the game between $\un{S}_1$ and $\ov{S}_1$. 

The arguments below also provide a financial interpretation of the value process~$S$ of this Dynkin game. In the special case that the terminal bid- and ask price coincide,
a discrete time approximation of $S$ can be interpreted as a shadow price for a utility maximization problem with a risk-neutral investor and the constraint that her dynamic stock 
position has to take values in $[-1,1]$. Put differently, in the bid-ask market, an investor can earn the same expected profit as via an optimal strategy 
in the frictionless market with price process~$S$ (besides a finite deviation caused by different liquidation values). 

Next, we recall the notion of a quasimartingale and Rao's Theorem (see, e.g., Theorem 17 in \cite[Chapter 3]{protter2005stochastic} or Theorem 3.1 in \cite{beiglbock2014riemann}). 
\begin{definition} Let $X=(X_t)_{t\in[0,1]}$ be an adapted process s.t.  
	$\E(|X_t|)<\infty$ for all $t\in[0,1]$.
	Given a deterministic partition $\pi=\{0=t_0<t_1<\dots<t_n=1\}$ of $[0,1]$ the \emph{mean-variation of $X$ along $\pi$} is defined as
	\begin{align*}
		MV(X,\pi):=\E\left[\sum_{t_i\in\pi}\left\vert \E\left[X_{t_i}-X_{t_{i+1}}\mid \F_{t_i}\right]\right\vert\right]
	\end{align*}
	and the \emph{mean variation of $X$} is defined as
	\begin{align*}
		MV(X):=\sup_\pi MV(X,\pi).
	\end{align*}
	Finally, $X$ is called a \emph{quasimartingale} if $MV(X)<\infty$.
\end{definition}
\begin{theorem}[Rao]\label{theo:rao}
	Let $X$ be an adapted right-continuous process. Then, $X$ is a quasimartingale if and only if $X$ has a decomposition $X=Y-Z$ where $Y$ and $Z$ are each positive right-continuous supermartingales. In this case, the paths of $X$ are a.s. c\`adl\`ag.
\end{theorem}
\begin{remark}\label{remark:right-continuity}
	Usually, Rao's theorem is formulated for an adapted \emph{c\`adl\`ag} process $X$. However, to show that $X$ can be written as the difference of two right-continuous supermartingales, the existence of the finite left limits of $X$ is not needed (see the proofs of Theorem~8.13 in \cite{he.wang.yan.1992} or
	Theorem 14 in \cite[Chapter 3]{protter2005stochastic}). On the other hand, right-continuous supermartingales possess a.s. finite left limits (see Theorem~VI.3 in \cite{dellacherie.meyer.1982}). This means that the theorem can be formulated for an a priori only right-continuous quasimartingale that turns out to be c\`adl\`ag.
\end{remark}

If we can show that the right-continuous process~$S$ is a local quasimartingale, Rao's theorem (in the version of Theorem~\ref{theo:rao}) yields that $S$ can locally be written as the difference of two supermartingales, and it admits a c\`{a}dl\`{a}g modification. Thus, $S$ is a semimartingale by the Doob-Meyer-Theorem (Case without Class D) \cite[Chapter 3, Theorem 16]{protter2005stochastic}. Hence, we now want to show that $S$ is a local quasimartingale.

For this, we consider a discrete time approximation $S^n=(S^n_t)_{t\in D_n}$ of $S$ on the set $D_n:=\{0,1/2^n,\dots (2^n-1)/2^n, 1\}$ of dyadic numbers defined by $S^n_1=\underline{S}_1$ and 
\begin{align}\label{eq:RecursiveDefinition}
	S^n_t:=\min\left(\overline{S}_t,\max\left( \underline{S}_t, \mathbb{E}\left[S^n_{t+1/2^n}\mid \mathcal{F}_t\right]\right)\right), \quad
	t\in D_n, \ t<1.
\end{align} 
Indeed, it is well-known (see, e.g., \cite[Proposition VI-6-9]{neveu1975discrete}) that 
\begin{align}
	\begin{aligned}\label{eq:DisDynkin}
		S^n_t&=\esssup_{\tau\in\mathcal{T}^n_{t,1}}\essinf_{\sigma\in\mathcal{T}^n_{t,1}}\E\left[\underline{S}_\tau\mathbbm{1}_{\{\tau\leq \sigma\}}+\overline{S}_\sigma \mathbbm{1}_{\{\tau >\sigma\}}\mid \F_t \right]\\&=\essinf_{\sigma\in\mathcal{T}^n_{t,1}}\esssup_{\tau\in\mathcal{T}^n_{t,1}}\E\left[\underline{S}_\tau\mathbbm{1}_{\{\tau\leq \sigma\}}+\overline{S}_\sigma \mathbbm{1}_{\{\tau >\sigma\}}\mid \F_t \right],\quad t\in D_n,
	\end{aligned}
\end{align}
where $\mathcal{T}^n_{t,1}$ denotes the set of all $\{t,t+1/2^n,\dots,1\}$-valued stopping times. The following proposition generalizes Kifer~\cite[Proposition 3.2]{kifer2000game} from continuous processes to right-continuous processes. 
\begin{proposition} \label{prop:Approximation}
	Let $m\in \mathbb{N}$ and $t\in D_m$. 
	Then, we have \[\lim\limits_{\substack{n\to\infty\\n\geq m}}S^n_t=S_t\quad \PM\mathrm{-a.s.}\]
\end{proposition}
\begin{proof}
	Let $n\in\mathbb{N}$ with $n\geq m$ and $t\in D_m$. The pair of $\{t,t+1/2^n,\dots,1\}$-valued stopping times 
	\begin{align*}
		\tau^{n}_t&:=\inf\{s\geq t: s\in D_n, S^n_s = \underline{S}_s\},\\
		\sigma^n_t&:=\inf\{s\geq t: s\in D_n, S^n_s = \overline{S}_s\}
	\end{align*}
	is a Nash equilibrium of the discrete time game started at time~$t$, i.e.,
	\begin{align}\label{eq:DisEpsOpti}
		\E\left[R(\tau,\sigma^n_t)\mid \F_t \right]\leq S^n_t\leq \E\left[R(\tau^n_t,\sigma)\mid \F_t \right]\quad\text{for all}\ \tau,\sigma\in \mathcal{T}^n_{t,T},
	\end{align}
	where $R(\tau,\sigma):=\underline{S}_\tau\mathbbm{1}_{\{\tau\leq \sigma\}}+\overline{S}_\sigma \mathbbm{1}_{\{\tau >\sigma\}}$. This follows from \cite[Proposition VI-6-9]{neveu1975discrete} and its proof with the observation that in {\em finite} discrete time the assertion also holds for $\eps=0$ by dominated convergence.
	For any $\tau\in \mathcal{T}_{t,T}$, we let $D_n(\tau):=\inf\{t\geq \tau: t\in D_n\}$ and  
	\[\eta_n(\tau)(\omega):=\sup_{\substack{s\in(\tau(\omega),\tau(\omega)+1/2^n)}}\max\left(\left\vert \underline{S}_s(\omega)-\underline{S}_\tau(\omega)\right\vert, \left\vert \overline{S}_s(\omega)-\overline{S}_\tau(\omega)\right\vert\right),\quad \omega\in\Omega.\]
	This yields the estimates
	\begin{align}\label{eq:DynkingDiscreteContSwitch}
		R(\tau,D_n(\sigma))-\eta_n(\tau)\leq R(D_n(\tau),D_n(\sigma))\leq R(D_n(\tau),\sigma)+\eta_n(\sigma)
	\end{align}
	for all $\tau,\sigma\in \mathcal{T}_{0,T}$. Let $\eps>0$. For the continuous time game, the pair of stopping times 
	\begin{align*}
		\tau^{*}_t&:=\inf\{s\geq t: S_s\leq \underline{S}_s+\varepsilon\},\\
		\sigma^*_t&:=\inf\{s\geq t: S_s\geq \overline{S}_s-\varepsilon\}
	\end{align*}
	is an $\eps$-Nash equilibrium, i.e., 
	\begin{align}\label{eq:ContEpsOpti}
		\E\left[R(\tau,\sigma^*_t)\mid \F_t \right]-\varepsilon\leq S_t\leq \E\left[R(\tau^*_t,\sigma)\mid \F_t \right]+\varepsilon,\quad\text{for all}\ \tau,\sigma\in \mathcal{T}_{t,T}.
	\end{align}
	This is shown in Corollaire 12 and its proof in \cite{lepeltier1984jeu}.  Combining the first inequality in \eqref{eq:DisEpsOpti} with $\tau=D_n(\tau^*_t)$, the first inequality in \eqref{eq:DynkingDiscreteContSwitch} and the second inequality in \eqref{eq:ContEpsOpti} yields 
	\begin{align*}
		S^n_t&\geq \E\left[R\left(D_n(\tau^*_t), \sigma^n_t\right)\mid\mathcal{F}_t\right]\\
		&\geq \E\left[ R(\tau^*_t, \sigma^n_t)\mid\F_t\right]-\E\left[\eta_n(\tau^*_t)\mid \F_t\right]\\
		&\geq S_t -\varepsilon-\E\left[\eta_n(\tau^*_t)\mid \F_t\right].
	\end{align*}
	Similar, applying the second inequality \eqref{eq:DisEpsOpti} with $\sigma=D_n(\sigma^*_t)$, the second inequality in \eqref{eq:DynkingDiscreteContSwitch} and the first inequality in \eqref{eq:ContEpsOpti}, yields the corresponding upper estimate on $S^n_t$. Putting together, we get
	\begin{align*}
		S_t+\varepsilon+\E\left[\eta_n(\sigma^*_t)\mid\F_t\right]\geq S^n_t\geq S_t -\varepsilon-\E\left[\eta_n(\tau^*_t)\mid \F_t\right].
	\end{align*}
	Finally, as $\eta_n(\tau^*_t)\to 0$ and $\eta_n(\sigma^*_t)\to0$ a.s. by the right-continuity of $\overline{S}$ and $\underline{S}$, the dominated convergence theorem for conditional expectations implies
	\begin{align*}
		S_t+\varepsilon\geq \limsup\limits_{\substack{n\to\infty\\ n\geq m}}S^n_t\geq \liminf\limits_{\substack{n\to\infty\\ n\geq m}}S^n_t\geq S_t -\varepsilon \quad \PM\mathrm{-a.s.},
	\end{align*}
	which is the assertion as $\varepsilon>0$ is arbitrary.
\end{proof}
In the following, we will consider the discrete-time Doob-decomposition of the processes $(S^n)_{n\in\mathbb{N}}$, i.e.,
we write $S^n_t=S^n_0+M^n_t+A^n_t$ with
\begin{align}
	\label{eq:DMA}
	A^n_t&:=\sum_{t_i\in D_n, 0<t_i\leq t}\E\left[S^n_{t_i}-S^n_{t_{i-1}}\mid \F_{t_{i-1}}\right],\\
	\label{eq:DMM}M^n_t&:=\sum_{t_i\in D_n, 0<t_i\leq t}\left(S^n_{t_i}-S^n_{t_{i-1}}-\E\left[S^n_{t_i}-S^n_{t_{i-1}}\mid \F_{t_{i-1}}\right]\right)
\end{align} 
for $t\in D_n$. In particular, we have (with a slight abuse of notation) 
\begin{align}\label{eq:discreteMV}
	MV(S^n,D_n):=\E\left[\sum_{t_i\in D_n}\left\vert \E\left[S^n_{t_{i+1}}-S^n_{t_i}\mid \F_{t_i}\right]\right\vert\right]=\E\left[\sum_{t_i\in D_n}\left\vert A^n_{t_{i+1}}-A^n_{t_i}\right\vert\right].
\end{align}
The following observation is at the core of why our approach works.
\begin{lemma}\label{lemma:MakingItWork}
	Let $n\in\mathbb{N}$ and $t=0,1/2^n,\dots, (2^n-1)/2^n$. Then, we have 
	\begin{align*}
		\{A^n_{t+1/2^n}-A^n_{t}>0\}&\subseteq\{S^n_t=\overline{S}_t\},\\
		\{A^n_{t+1/2^n}-A^n_{t}<0\}&\subseteq\{S^n_t=\underline{S}_t\}.
	\end{align*}
\end{lemma}
\begin{proof}
	From definition~\eqref{eq:DMA} we get $\E\left[S^n_{t+1/2^n}\mid\F_t\right]-S^n_t=A^n_{t+1/2^n}-A^n_t,$
	which together with  $S^n_t=\min(\overline{S}_t,\max( \underline{S}_t, \mathbb{E}[S^n_{t+1/2^n}\mid \mathcal{F}_t]))$ yields the assertion. 
\end{proof}

We now start to establish a uniform bound on \eqref{eq:discreteMV} (after some stopping). 
\begin{lemma}\label{Lemma:L0-Bound-Martingale}
	Let Assumption~\ref{ass:bounded} hold. Then, the set $$\left\{\sup_{t\in D_n}\vert M^n_t\vert: n\in\mathbb{N} \right\}$$ is bounded in probability.
\end{lemma} 
\begin{proof}
	Before we begin, we roughly sketch the idea of the proof. If $\{\sup_{t\in D_n}\vert M^n_t\vert: n\in\mathbb{N}\}$ failed to be bounded in probability, the same would hold in some sense for the sequence $(A^n)_{n\in\mathbb{N}}$. Indeed, this is a consequence of $S^n=S^n_0 + M^n+A^n$ and the fact that $|S^n|\le 1$. Keeping Lemma~\ref{lemma:MakingItWork} in mind, we show that by suitable long-only investments in the bid-ask market, one can earn the increasing parts of $A^n$ without suffering from the decreasing parts. In doing so, 
	we would achieve an (UPBR) since
	the gains from $A^n$ are of a higher order than the potential losses from the martingale part~$M^n$.   
	The proof of the latter relies on the brilliant ideas of Delbaen and Schachermayer~\cite[Lemma 4.7]{delbaen1994general}, which we adapt to the present setting. The present setting is easier than in \cite[Lemma 4.7]{delbaen1994general} since the jumps of $S^n$ are uniformly bounded. 
	
	\textit{Step 1:}  Assume that the claim does not hold true, i.e., there is a subsequence $(\sup_{t\in D_{m_n}}\vert M^{m_n}_t\vert)_{n\in\mathbb{N}}$ and $\alpha\in (0,1/10)$ s.t. \begin{align*}
		\mathbb{P}(\sup_{t\in D_{m_n}}\vert M^{m_n}_t\vert \geq n^3)>10\alpha, \quad n\in\mathbb{N}.
	\end{align*}
	In the following, we write $(\sup_{t\in D_{n}}\vert M^{n}_t\vert)_{n\in\mathbb{N}}$ instead of $(\sup_{t\in D_{m_n}}\vert M^{m_n}_t\vert)_{n\in\mathbb{N}}$ in order to simplify the notation. For this, it is important to note that from now on, we do not use properties of $M^n$ that do not hold for $M^{m_n}$.
	Let $T_n:=\inf\{t\in D_n: \vert M^n_t\vert \geq n^3\}$ and define the process $(\wt{S}^n_t)_{t\in D_n}$ by $\wt{S}^n_t:=\frac{1}{n^2} S^n_{t\wedge T_n}$. Note that the (discrete-time) Doob decomposition of $\wt{S}^n$ is given by 
	\[\widetilde{S}^n_t=\wt{S}^n_0+\wt{M}^n_t+\wt{A}^n_t=\frac{1}{n^2}S_0^n+\frac{1}{n^2}M^n_{t\wedge T_n}+\frac{1}{n^2}A^n_{t\wedge T_n},\quad t\in D_n,\]
	where $(\wt{M}^n_t)_{t\in D_n}=(\frac{1}{n^2}M^n_{t\wedge T_n})_{t\in D_n}$ is the martingale part and $(\wt{A}^n_t)_{t\in D_n}=(\frac{1}{n^2}A^n_{t\wedge T_n})_{t\in D_n}$ the predictable part. In addition, we have 
	\begin{align}
		\label{eq:PropertyM1}
		\PM(\sup_{t\in D_n}\vert \wt{M}^n_t\vert \geq n)>10\alpha,\quad 
		\vert \wt{S}^n_t-\wt{S}^n_{t-1/2^n}\vert\leq \frac{1}{n^2},\ t\in D_n.
	\end{align} 
	Next, we define $T_{n,0}:=0$ and, recursively, \[T_{n,i}:=\inf\{t\geq T_{n,i-1}: t\in D_n,\ \vert \wt{M}^n_{t}-\wt{M}^n_{T_{n,i-1}}\vert\geq 1 \},\quad i\in\bbn.\] 
	Since $|A^n_{t}-A^n_{t-1/2^n}|\le 1$ and thus
	\beam\label{25.12.2019.1}
	|M^n_t-M^n_{t-1/2^n}|\le |S^n_{t}-S^n_{t-1/2^n}|	+ |A^n_t-A^n_{t-1/2^n}|\le 2
	\eeam
	for all $t\in D_n\setminus\{0\}$, we get  
	\begin{align}\label{eq:MartingalIncrement}
		|\wt{M}^n_{T_{n,i}\wedge 1}-\wt{M}^n_{T_{n,i-1}\wedge 1}|\leq 1+ |\wt{M}^n_{T_{n,i}\wedge 1}-\wt{M}^n_{(T_{n,i}-1/2^n)\wedge 1}|\leq 1+2/n^2\leq 3,
	\end{align}
	for all $n,i\in\bbn$. \eqref{eq:MartingalIncrement} implies
	\begin{align}\label{eq:finitenessTni1}
		\PM(T_{n,i}<\infty)>10\alpha\quad\text{for}\ n\in\bbn\quad\mbox{and}\quad i=0,\dots, k_n,
	\end{align}
	where $k_n:=\lfloor (n-1)/3\rfloor$ denotes the integer part of $(n-1)/3$. 
	
	Next, we establish a lower bound in $L^0(\PM)$ on $(\wt{M}^n_{T_{n,i}\wedge 1}-\wt{M}^n_{T_{n,i-1}\wedge 1})^-$ for $i=1,\dots, k_n$. The martingale property of $\wt{M}^n$  together with \eqref{eq:finitenessTni1} implies
	\begin{align*}
		\E\left[(\wt{M}^n_{T_{n,i}\wedge 1}-\wt{M}^n_{T_{n,i-1}\wedge 1})^-\right]=\frac{1}{2}\E\left[\vert\wt{M}^n_{T_{n,i}\wedge 1}-\wt{M}^n_{T_{n,i-1}\wedge 1}\vert\right]\geq \frac{1}{2}\mathbb{P}(T_{n,i}<\infty)>5\alpha.
	\end{align*}
	For $B_{n,i}:=\{(\wt{M}^n_{T_{n,i}\wedge 1}-\wt{M}^n_{T_{n,i-1}\wedge 1})^-\geq 2\alpha\}$, we get
	\beao
	\E\left[(\wt{M}^n_{T_{n,i}\wedge 1}-\wt{M}^n_{T_{n,i-1}\wedge 1})^-\mathbbm{1}_{B_{n,i}}\right]
	\ge \E\left[(\wt{M}^n_{T_{n,i}\wedge 1}-\wt{M}^n_{T_{n,i-1}\wedge 1})^-\right] - 2\alpha
	>3\alpha
	\eeao
	and thus by (\ref{eq:MartingalIncrement})
	\begin{align}\label{eq:MartingalpartLowerEstimate}
		\PM\left(B_{n,i}\right)> \alpha \quad\text{for}\ n\in\bbn\quad\mbox{and}\quad i=0,\dots, k_n.
	\end{align} 
	
	We now turn our attention to the increments $(\wt{A}^n_{T_{n,i}\wedge 1}-\wt{A}^n_{T_{n,i-1}\wedge 1})_{i=1,\dots, k_n}$ for $n\in\bbn$. 
	Since $|\wt{S}^n_{T_{n,i}\wedge 1}-\wt{S}^n_{T_{n,i-1}\wedge 1}|\le 1/n^2$, \eqref{eq:MartingalpartLowerEstimate} implies
	\begin{align*}
		\PM\left(\wt{A}^{n}_{T_{n,i}\wedge 1}-\wt{A}^n_{T_{n,i-1}\wedge 1}\geq \alpha \right)\geq 
		\PM\left(\wt{A}^{n}_{T_{n,i}\wedge 1}-\wt{A}^n_{T_{n,i-1}\wedge 1}\geq 2\alpha -\frac1{n^2}\right)\geq \PM\left(B_{n,i}\right)> \alpha
	\end{align*}
	for all $n\ge \sqrt{\alpha}$ and $i=1,\dots, k_n$.  In particular, if we define $(\wt{A}^{n,\uparrow}_t)_{t\in D_n}$ by 
	\[\wt{A}^{n,\uparrow}_t:=\sum_{t_i\in D_n, 0<t_i\leq t}
	(\wt{A}^{n}_{t_i}-\wt{A}^{n}_{t_{i-1}})^+,\quad t\in D_n, \]
	we also get 
	\begin{align}\label{eq:EstimateOnAProfit}
		\PM\left(\wt{A}^{n,\uparrow}_{T_{n,i}\wedge 1}-\wt{A}^{n,\uparrow}_{T_{n,i-1}\wedge 1}\geq \alpha\right)> \alpha
	\end{align}
	for all  all $n\ge 1/\sqrt{\alpha}$ and $i=1,\dots, k_n$.
	
	\textit{Step 2:} In the second part of the proof, we construct an (UPBR) by placing smart bets on the process $(\wt{A}^{n,\uparrow}_t)_{t\in D_n}$. This is similar to the second part of \cite[Lemma 4.7]{delbaen1994general} with the major difference that 
	we cannot invest directly into $S^n$. We define two sequences of $D_n\cup\{\infty\}$-valued stopping times $(\sigma^n_k)_{k=1}^{2^n}$ and $(\tau^n_k)_{k=1}^{2^n}$ by
	\begin{align*}
		\sigma^n_1:=\inf\{t\in D_n\mid A^n_{t+1/2^n}-A^n_t>0\}, \quad \tau^n_1:=\inf\{t>\sigma^n_1\mid t\in D_n,\ A^n_{t+1/2^n}-A^n_{t}<0\},
	\end{align*}  
	and, recursively,
	\begin{align*}
		\sigma^n_k&:=\inf\{t>\tau^n_{k-1}\mid t\in D_n,\ A^n_{t+1/2^n}-A^n_t>0\}, \\ \tau^n_k&:=\inf\{t>\sigma^n_k\mid t\in D_n,\ A^n_{t+1/2^n}-A^n_{t}<0\}
	\end{align*}
	for $k=2,3,\dots, 2^n$. Next, define a sequence of simple trading strategies $(\varphi^{n})_{n\in\mathbb{N}}$ by
	\begin{align*}
		\varphi^{n}:=\left(\sum_{k=1}^{2^n}\frac{1}{n^2}\mathbbm{1}_{\rrbracket\sigma^n_k,\tau^n_k\rrbracket}\right)\mathbbm{1}_{\rrbracket 0, T_{n,k_n}\rrbracket}.
	\end{align*} 
	By Lemma~\ref{lemma:MakingItWork}, the strategies~$\vp^n$ only buy if $S^n_t=\ov{S}_t$ and sell if $S^n_t=\un{S}_t$, despite of a possible liquidation at  $T_{n,k_n}$
	Together with $S^n_{t_i}-\un{S}_t\leq 1$ for all $t_i\in D_n$, $t\in[0,1]$, this implies that $\Vl(\varphi^n)$ can be bounded from below by 
	\begin{align}\nonumber
		\Vl_t(\varphi^{n})&\geq \sum_{t_i\in D_n, 0<t_i\leq t}\varphi^{n}_{t_i}(S^n_{t_{i}}-S^n_{t_{i-1}})-\frac{1}{n^2}\\\nonumber
		&=\wt{A}^{n,\uparrow}_{\lfloor 2^nt\rfloor /2^n\wedge T_{n,k_n}}+\sum_{t_i\in D_n, 0<t_i\leq t}\varphi^{n}_{t_i}(M^n_{t_{i}}-M^n_{t_{i-1}})-\frac{1}{n^2}\\ \nonumber
		&\geq \sum_{t_i\in D_n, 0<t_i\leq t}\varphi^{n}_{t_i}(M^n_{t_{i}}-M^n_{t_{i-1}})-\frac{1}{n^2}\\ \label{eq:profits1}
		&=\sum_{t_i\in D_n, 0<t_i\leq t}(n^2\varphi^{n}_{t_i})(\wt{M}^n_{t_{i}}-\wt{M}^n_{t_{i-1}})-\frac{1}{n^2},\quad t\in[0,1].
	\end{align} 
	This means that the strategy allows us to invest in $\wt{A}^{n,\uparrow}$, but we still do not know if it actually allows for an (UPBR) as we need to get some control on the martingale part in \eqref{eq:profits1}. Therefore notice that 
	\begin{align}\label{eq:AdjustedStrategyBound}
		& &	\left\Vert \sum_{t_i\in D_n, 0<t_i\leq T_{n,k_n}}(n^2\varphi^{n}_{t_i})(\wt{M}^n_{t_{i}}-\wt{M}^n_{t_{i-1}})\right\Vert_{L^2(\PM)}\leq \left\Vert \widetilde{M}_{T_{n,k_n}\wedge 1}\right\Vert_{L^2(\PM)}\nonumber\\
		& & \le	
		\sqrt{\sum_{i=1}^{k_n}\left\Vert\wt{M}^n_{T_{n,i}\wedge 1}-\wt{M}^n_{T_{n,i-1}\wedge 1}\right\Vert_{L^2(\PM)}^2}\leq  3\sqrt{k_n}.
	\end{align}
	Thus, Doob's maximal inequality yields
	\begin{align}\label{eq:MartingalePartZero}
		\left\Vert \sup_{t\in D_n,\ t\le T_{n,k_n}}\left\vert\sum_{t_i\in D_n, 0<t_i\leq t}(n^2\varphi^{n}_{t_i})(\wt{M}^n_{t_{i}}-\wt{M}^n_{t_{i-1}})\right\vert\right\Vert_{L^2(\PM)}\leq 6\sqrt{k_n}.
	\end{align}
	Consequently, we get the estimate
	\begin{align}\nonumber
		&\PM\left(\inf_{t\in[0,T_{n,k_n}\wedge 1]}\Vl_t\left(\varphi^{n}\right)\leq -k_n^{3/4}n^{-1/8}-n^{-2}\right)\\ \label{eq:lowerestimate}&\leq\PM\left(\sup_{t\in D_n,\ t\le T_{n,k_n}}\left\vert\sum_{t_i\in D_n, 0<t_i\leq t}(n^2\varphi^n_{t_i})(\wt{M}^n_{t_{i}}-\wt{M}^n_{t_{i-1}})\right\vert\geq k_n^{3/4}n^{-1/8}\right)\leq \frac{36n^{1/4}}{\sqrt{k_n}}
	\end{align}
	by Tschebyscheff's inequality. Thus, let us define the stopping times \[U_n:=\inf\{t\geq 0\ :\ \Vl_t(\varphi^{n})\leq -k_n^{3/4}n^{-1/8} -n^{-2} \}\wedge T_{n,k_n},\]
	which satisfy $\PM\left(U_n<T_{n,k_n}\right)\leq 36n^{1/4}/\sqrt{k_n}$. We now pass to the strategy \[\widetilde{\varphi}^{n}:=(k_n)^{-3/4}\varphi^{n}\mathbbm{1}_{\rrbracket0,U_n\rrbracket}.\] The left and right jumps of $\Vl(\wt{\varphi}^{n})$ are bounded from below by $-k_n^{-3/4}n^{-2}$, which is a direct consequence of $
	0\le \underline{S}\le \ov{S}\le 1$. We obtain \begin{align}\label{eq:UniformBoundBelow}
		\inf_{t\in[0,T_{n,k_n}\wedge 1]}\Vl_t(\wt{\varphi}^{n})\geq 
		-n^{-1/8}-2k_n^{-3/4}n^{-2}\to 0,\quad \text{for}\ n\to\infty.
	\end{align}
	It remains to show (\ref{eq:DefUPBR}). First notice that using \eqref{eq:EstimateOnAProfit} in conjunction with \cite[Corollary A1.3]{delbaen1994general}, yields
	\begin{align*}
		\PM\left(\wt{A}^{n,\uparrow}_{T_{n,k_n}\wedge 1}\geq\frac{\alpha^2}2\right)>\frac{\alpha}2.
	\end{align*}
	It follows that 
	\begin{align}\label{24.12.2019.1}
		\begin{aligned}
			\PM\left((k_n)^{-3/4}\wt{A}^{n,\uparrow}_{T_{n,k_n}\wedge \frac{\lfloor 2^nU_n\rfloor}{2^n}\wedge 1}\geq k_n^{1/4}\frac{\alpha^2}2\right)&>\frac{\alpha}2
			-\PM\left(U_n<T_{n,k_n}\right)
			\\ & \geq \frac{\alpha}2-\frac{36n^{1/4}}{\sqrt{k_n}}.
		\end{aligned}
	\end{align}
	Putting (\ref{eq:profits1}), (\ref{eq:lowerestimate}), (\ref{eq:UniformBoundBelow}), and (\ref{24.12.2019.1}) together yields that $(\wt{\vp}^n)_{n\in\bbn}$ is an (UPBR).
\end{proof}
\begin{lemma}\label{lemma:Stopping Times}
	Let Assumption~\ref{ass:bounded} hold. 
	For each $\varepsilon>0$, there exists a constant $C>0$ and a sequence of $D_n\cup\{\infty\}$-valued stopping times $(\tau_n)_{n\in\mathbb{N}}$ s.t. $\PM(\tau_n<\infty)<\varepsilon$ and the stopped processes $S^{n,\tau_n}=(S^n_{t\wedge \tau_n})_{t\in D_n}$, $A^{n,\tau_n}=(A_{t\wedge \tau_n})_{t\in D_n}$ satisfy
	\begin{align}
		\sum_{t_i\in D_n}\left\vert A^{n,\tau_n}_{t_{i+1}}-A^{n,\tau_n}_{t_i}\right\vert&\leq C\\
		\text{and, consequently,}\quad
		MV(S^{n,\tau_n},D_n)=\E\left[\sum_{t_i\in D_n}\left\vert A^{n,\tau_n}_{t_{i+1}}-A^{n,\tau_n}_{t_i}\right\vert\right]&\leq C.
	\end{align}
\end{lemma}
\begin{proof}
	The idea of the proof is akin to the proofs of Proposition 3.1 and Lemma 3.4 in Beiglb{\"o}ck et al.~\cite{beiglbock2011direct}. Thus, we only give a sketch of the proof 
	and leave the details to the reader. We first claim that \begin{align}\label{eq:newnumber2}
		\left\{\sum_{t_i\in D_n}\left( A^{n}_{t_{i+1}}-A^{n}_{t_i}\right)^+: n\in\mathbb{N}\right\}
	\end{align} is bounded in probability.  
	We proceed by contraposition, i.e., we suppose otherwise and want to show that this leads to an (UPBR). Using Lemma~\ref{lemma:MakingItWork}, we can analogously to the previous proof construct a sequence of simple trading strategies $(\varphi^{n})_{n\in\mathbb{N}}$ with $0\leq \vp^n\leq 1$ s.t. $\vp^n$ invests in $\sum_{t_i\in D_n}\left( A^n_{t_{i+1}}-A^n_{t_i}\right)^+$ while only making potential losses in the martingale part $M^n$ and at liquidation. Indeed, similar as in step 2 of the proof of Lemma~\ref{Lemma:L0-Bound-Martingale}, it can be shown that the associated liquidation values can be bounded from below by
	\begin{align}\label{eq:newnumber1}
		\Vl_t(\varphi^{n})\geq \sum_{t_i\in D_n, 0<t_i\leq t}\left( A^n_{t_{i+1}}-A^n_{t_i}\right)^++ \sum_{t_i\in D_n, 0<t_i\leq t}\varphi^{n}_{t_i}\left(M^n_{t_{i}}-M^n_{t_{i-1}}\right) -1.
	\end{align}   
	By the previous Lemma~\ref{Lemma:L0-Bound-Martingale} and some stopping, there is no loss of generality by assuming that $(M^n)_{n\in\bbn}$ is uniformly bounded. Hence, 
	by Doob's maximal inequality, the pathwise maxima of the martingale parts in \eqref{eq:newnumber1} are bounded in $L^2$. Thus, by further stopping (cf. the arguments used in 
	Beiglb{\"o}ck et al.~\cite[page 2433, lines 11-15]{beiglbock2011direct}), we may assume that \eqref{eq:newnumber1} is uniformly bounded from below. 
	On the other hand, by assumption, the RHS of 
	\eqref{eq:newnumber1} is unbounded in probability from above. Thus, the (adjusted) strategies yield an (UPBR) with long-only strategies (after rescaling), and we arrive at a 
	contradiction. Consequently, \eqref{eq:newnumber2} has to be bounded in probability. Since the martingale parts are also bounded in probability by Lemma~\ref{Lemma:L0-Bound-Martingale}, the 
	same holds for $\left\{\sum_{t_i\in D_n}\left( A^{n}_{t_{i+1}}-A^{n}_{t_i}\right)^-: n\in\mathbb{N}\right\}$, and we are done.
\end{proof}

In order to finish the proof of Theorem~\ref{theo:ExistenceSimple}  we still need a couple of auxiliary results, which give us some more information about $MV(S^n,D_n)$ in comparison to $MV(S^m,D_m)$. 
Given a partition $\pi=\{0=t_0<t_1<\dots<t_n=1\}$ of $[0,1]$ and a stopping time $\tau$, we have the following notation $\pi(\tau):=\inf\{t\in \pi: t\geq \tau\}$.	Recall the following useful result from \cite{beiglbock2014riemann}.
\begin{lemma}[Lemma 3.2 of \cite{beiglbock2014riemann}] \label{lemma:Beigl1}
	Let Assumption~\ref{ass:bounded} hold. Then 
	\begin{align*}
		\mathrm{MV}(S^{\pi(\tau)},\pi)=\E\left[\sum_{t_i\in\pi}\mathbbm{1}_{\{t_i<\tau\}}\left\vert\E\left[S_{t_{i+1}}-S_{t_i}\mid \F_{t_i}\right]\right\vert\right]
	\end{align*}
	and $\left\vert \mathrm{MV}(S^{\pi(\tau)},\pi)-\mathrm{MV}(S^\tau,\pi)\right\vert\leq 1$. 
\end{lemma}

Compared to the frictionless case with $S^n=\un{S}=\ov{S}$, the analysis is complicated by the fact that in general $S^m_t\not=S^n_t$ for $t\in D_n$. We have nevertheless the following monotonicity result. 
\begin{lemma}\label{lemma:monotonie}
	Let Assumption~\ref{ass:bounded} hold. In addition, let $n,m\in\mathbb{N}$ with $m>n$ and let $\tau_m$ be a $D_m\cup\{\infty\}$-valued stopping time. For any $s\in D_n$, we have 
	\begin{align*}
		&\E\left[\sum_{t_i\in D_n, t_i\geq s}\mathbbm{1}_{\{t_i<\tau_m\}}\vert\E\left[S^{n}_{t_{i+1}}-S^{n}_{t_{i}}\mid \F_{t_i}\right]\vert \mid \F_s\right]\\& \leq\ \E\left[\sum_{t_i\in D_m, t_i\geq s}\mathbbm{1}_{\{t_i<\tau_m\}}\vert\E\left[S^{m}_{t_{i+1}}-S^{m}_{t_{i}}\mid \F_{t_i}\right]\vert \mid \F_s\right]+\left(2-\vert S^{n}_s-S^{m}_s\vert\right)\mathbbm{1}_{\{s<\tau_m\}} .
	\end{align*}
	In particular, for $s=0$ this yields
	\begin{align*}
		MV(S^{n,D_n(\tau_m)},D_n)\leq MV(S^{m,\tau_m},D_m)+2.
	\end{align*}
	In addition, we have
	\beam\label{remark:ReplacingDnDm}
	\mathrm{MV}(S^{m,D_n(\tau)}, D_n)\leq \mathrm{MV}(S^{m,D_m(\tau)},D_m)+1. 
	\eeam
	for all $[0,1]\cup\{\infty\}$-valued stopping times~$\tau$.
\end{lemma}
\begin{proof}
	\textit{Step 1:} In a first step, we keep the grid $D_n$ but replace $S^n$ with $S^m$. Thus, we want to show
	\begin{align}\nonumber
		&\E\left[\sum_{t_i\in D_n, t_i\geq s}\mathbbm{1}_{\{t_i<\tau_m\}}\vert\E\left[S^{n}_{t_{i+1}}-S^{n}_{t_{i}}\mid \F_{t_i}\right]\vert \mid \F_s\right]\\ \label{eq:Step1Stopped} &\leq\E\left[\sum_{t_i\in D_n, t_i\geq s}\mathbbm{1}_{\{t_i<\tau_m\}}\vert\E\left[S^{m}_{t_{i+1}}-S^{m}_{t_{i}}\mid \F_{t_i}\right]\vert \mid \F_s\right]+\left(1-\vert S^{n}_s-S^{m}_s\vert\right)\mathbbm{1}_{\{s<\tau_m\}}.
	\end{align}
	We start by showing the one-step estimate 
	\beam\label{21.12.2019.3}
	& &  \left\vert\E\left[S^n_{s+1/2^n}-S^n_s\mid \F_s\right]\right\vert\nonumber\\
	& & =\left\vert\E\left[S^n_{s+1/2^n}-S^m_s\mid \F_s\right]\right\vert-\vert S^n_{s}-S^m_s\vert\nonumber\\
	& & \leq \left\vert\E\left[S^m_{s+1/2^n}-S^m_s\mid \F_s\right]\right\vert
	+\E\left[\left\vert S^m_{s+1/2^n}-S^n_{s+1/2^n}\right\vert\mid \F_s\right] -\vert S^n_{s}-S^m_s\vert
	\eeam
	for all $s=1-1/2^n,1-2/2^n,\ldots,0$. The equality in (\ref{21.12.2019.3}) can be checked separately on the $\mathcal{F}_s$-measurable sets $B_1:=\{\E\left[S^n_{s+2^{-n}}\mid \F_s\right]>\overline{S}_{s}\}$,  $B_2:=\{\E\left[S^n_{s+2^{-n}}\mid \F_{s}\right]<\underline{S}_{s}\}$, and $B_3:=\{\underline{S}_{s}\leq \E\left[S^n_{s+2^{-n}}\mid \F_{s}\right]\leq\overline{S}_s\}$.
	By the definition of $S^n$, $B_1\subseteq \{S^n_{s}=\overline{S}_{s}\}$. On the other hand, $S^m_{s}\leq \overline{S}_{s}$, which implies the equality on $B_1$. 
	On the set $B_2\subseteq\{S^n_{s}=\underline{S}_{s}\}$, the situation is completely symmetric. Finally, on $B_3
	=\{S^n_s = \E\left[S^n_{s+2^{-n}}\mid \F_{s}\right]\}$, the equality is obvious. The inequality in (\ref{21.12.2019.3}) follows from Jensen's inequality for conditional expectations and the triangle inequality.
	
	Now, we show \eqref{eq:Step1Stopped} by a backward-induction on $s=1-1/2^n,1-2/2^n,\ldots,0$. For the initial step $s=1-1/2^n$, we 
	only have to multiply (\ref{21.12.2019.3}) for $s=1-1/2^n$ by $\mathbbm{1}_{\{1-2^{-n}<\tau_m\}}$ and use that  $|S^m_1-S^n_1|\le 1$. 
	
	Induction step $s+1/2^n\leadsto s$: By the induction hypothesis,
	one has 
	\begin{align}\nonumber
		\E\left[\sum_{t_i\in D_t, t_i\geq s+1/2^n}\mathbbm{1}_{\{t_i<\tau^m\}}\left\vert\E\left[S^n_{t_{i+1}}-S^n_{t_{i}}\mid \F_{t_i}\right]\right\vert \bigg\vert \F_{s}\right]\\ \nonumber
		\leq\E\left[\sum_{t_i\in D_n, t_i\geq s+1/2^n}\mathbbm{1}_{\{t_i<\tau_m\}}\left\vert\E\left[S^{m}_{t_{i+1}}-S^{m}_{t_{i}}\mid\F_{t_i}\right]\right\vert \bigg\vert \F_{s}\right]\\ \label{eq:IVMonotonieStopped}+\mathbbm{1}_{\{s<\tau_m\}}\E\left[1-\vert S^{n}_{s+1/2^n}-S^{m}_{s+1/2^n}\vert\vert \F_s\right], 
	\end{align} 
	where we take on both sides of \eqref{eq:Step1Stopped} for $s+1/2^n$    the conditional expectation under $\F_s$
	and use that $\{s+1/2^n<\tau_m\}\subseteq\{s<\tau_m\}$.
	Multiplying (\ref{21.12.2019.3}) by $\mathbbm{1}_{\{s<\tau_m\}}$ and adding (\ref{eq:IVMonotonieStopped}) yields \eqref{eq:Step1Stopped}.
	
	\textit{Step 2:} We still need to pass from $D_n$ to $D_m$ for the process~$S^m$, i.e., we now want to show that
	\begin{align} \nonumber
		&\E\left[\sum_{t_i\in D_n, t_i\geq s}\mathbbm{1}_{\{t_i<\tau_m\}}\vert\E\left[S^{m}_{t_{i+1}}-S^{m}_{t_{i}}\mid \F_{t_i}\right]\vert \mid \F_s\right]\\&\leq \E\left[\sum_{t_i\in D_m, t_i\geq s}\mathbbm{1}_{\{t_i<\tau_m\}}\vert\E\left[S^{m}_{t_{i+1}}-S^{m}_{t_{i}}\mid\F_{t_i}\right]\vert \mid \F_s\right]+\mathbbm{1}_{\{s<\tau_m\}}. \label{eq:MontonySecondAss}
	\end{align}
	This is less tricky: for $\tau_m=1$, it directly follows from the triangle inequality together with Jensen's inequality for conditional expectations and the second summand on the RHS is not needed. However, in the general case there is the problem that $\tau_m$ can stop in $D_m\setminus D_n$. Thus, for every $i\in\{s2^n,s2^n+1,\ldots,2^n-1\}$, we have to make the following calculations
	\begin{align}\nonumber
		&\mathbbm{1}_{\{i/2^n<\tau_m\}}\left\vert\E\left[ S^m_{(i+1)/2^n}-S^m_{i/2^n}\mid \F_{i/2^n}\right]\right\vert\\ \nonumber
		=&\ \mathbbm{1}_{\{i/2^n<\tau_m\}}\left\vert \E\left[\sum_{j=i2^{m-n}}^{(i+1)2^{m-n}-1}\left(S^m_{(j+1)/2^m}-S^m_{j/2^m}\right)\mid \F_{i/2^n}\right]\right\vert\\\nonumber
		\le &\ \E\left[\sum_{j=i2^{m-n}}^{(i+1)2^{m-n}-1}\mathbbm{1}_{\{j/2^m<\tau_m\}}\left\vert\E\left[S^m_{(j+1)/2^m}-S^m_{j/2^m}\mid \F_{j/2^m}\right]\right\vert
		\mid \F_{i/2^n}\right] \\\label{eq:MonotonEstimate}& 
		+\left\vert\E\left[\mathbbm{1}_{\{i/2^n<\tau_m\}}
		\sum_{j=i2^{m-n}}^{(i+1)2^{m-n}-1}
		\mathbbm{1}_{\{j/2^m\geq\tau_m\}}\left(S^m_{(j+1)/2^m}-S^m_{j/2^m}\right)\mid \F_{i/2^n}\right]\right\vert.
	\end{align}
	For the second summand, we can use the estimate
	\begin{align}
		\nonumber
		&\left\vert\mathbbm{1}_{\{i/2^n<\tau_m\}}
		\sum_{j=i2^{m-n}}^{(i+1)2^{m-n}-1}\mathbbm{1}_{\{j/2^m\geq\tau_m\}}\left(S^m_{(j+1)/2^m}-S^m_{j/2^m}\right)\right\vert\\
		\nonumber
		&=\left\vert\sum_{j=i2^{m-n}+1}^{(i+1)2^{m-n}-1}\mathbbm{1}_{\{(j-1)/2^m< \tau_m\le j/2^m\}}\left(S^m_{(i+1)/2^n}-S^m_{j/2^m}\right)\right\vert\\  \label{eq:MonotonEstimate2}
		&\leq \sum_{j=i2^{m-n}+1}^{(i+1)2^{m-n}-1}\mathbbm{1}_{\{(j-1)/2^m< \tau_m\le j/2^m\}} \le \mathbbm{1}_{\{i/2^n< \tau_m\le (i+1)/2^n\}},
	\end{align}
	where we use $0\le S^m_{t_i}\leq 1$ for all $t_i\in D_m$. Putting \eqref{eq:MonotonEstimate} and \eqref{eq:MonotonEstimate2} together and summing up over all $i$, we  arrive at \eqref{eq:MontonySecondAss}. Together with \eqref{eq:Step1Stopped}, this yields the main assertion. (\ref{remark:ReplacingDnDm}) is just (\ref{eq:MontonySecondAss}). 
\end{proof}

For the convenience of the reader, we recall the following result from \cite{beiglbock2014riemann}.
\begin{lemma}[Lemma 4.2 in \cite{beiglbock2014riemann}] \label{lemma:Beigl2}
	Assume that $(\tau_n)_{n\in\mathbb{N}}$ is a sequence of $[0,1]\cup \{\infty\}$-valued stopping times s.t. $\PM(\tau_n=\infty)\geq 1-\varepsilon$ for some $\varepsilon>0$ and all $n\in\bbn$. Then, there exists a stopping time $\tau$ and for each $n\in\bbn$ convex weights $\mu_n^n,\dots,\mu_{N_n}^n$, i.e., $\mu^n_k\ge 0$, $k=n,\ldots,N_n$ and $\sum_{k=n}^{N_n}\mu_k^n=1$, s.t. $\PM(\tau=\infty)\geq 1-3\varepsilon$ and 
	\begin{align}
		\mathbbm{1}_{\llbracket 0,\tau\rrbracket}\leq 2\sum_{k=n}^{N_n}\mu_k^n\mathbbm{1}_{\llbracket 0,\tau_k\rrbracket},\quad n\in\bbn.
	\end{align}
\end{lemma}	

We are now in the position to prove Theorem~\ref{theo:ExistenceSimple}. 
\begin{proof}[Proof of Theorem \ref{theo:ExistenceSimple}] Let Assumption~\ref{ass:bounded} hold. Let $\varepsilon>0$, $(\tau_n)_{n\in\mathbb{N}}$ and $C>0$ as in Lemma~\ref{lemma:Stopping Times}. In addition, let $\tau$ as in Lemma~\ref{lemma:Beigl2}. We have
	\begin{align}\nonumber
		MV(S^{n,D_n(\tau)},D_n)&=\E\left[\sum_{t_i\in D_n}\mathbbm{1}_{\{t_i<\tau\}}\left\vert\E\left[S^n_{t_{i+1}}-S^n_{t_i}\mid \F_{t_i}\right]\right\vert\right]\\ \nonumber
		&\leq 2\E\left[\sum_{t_i\in D_n}\sum_{k=n}^{N_n}\mu_k^n\mathbbm{1}_{\{t_i<\tau_k\}}\left\vert\E\left[S^n_{t_{i+1}}-S^n_{t_i}\mid \F_{t_i}\right]\right\vert\right]\\ \nonumber
		&=2\sum_{k=n}^{N_n}\mu_k^n MV(S^{n,D_n(\tau_k)},D_n)\\ \label{eq:EstimateProof1}
		&\leq 2\sum_{k=n}^{N_n}\mu_k^n (MV(S^{k,\tau_k},D_k)+2)\leq 2C+4,\quad n\in\mathbb{N}.
	\end{align}
	Indeed, both equalities hold by Lemma~\ref{lemma:Beigl1}. The first inequality is due to Lemma~\ref{lemma:Beigl2}
	and the second inequality follows from Lemma~\ref{lemma:monotonie}.
	The third inequality holds by Lemma~\ref{lemma:Stopping Times}. 
	Next, let us show that
	\begin{align*}
		\mathrm{MV}(S^{D_n(\tau)},D_n)=\lim\limits_{\substack{m\to\infty\\ m\geq n}}\mathrm{MV}(S^{m,D_n(\tau)},D_n)\leq\limsup\limits_{\substack{m\to\infty\\ m\geq n}}\mathrm{MV}(S^{m,D_m(\tau)},D_m)+1\leq 2C+5,
	\end{align*}
	$n\in\bbn$, where $S$ is the value process of the continuous time game.
	Indeed, the equality follows from Proposition~\ref{prop:Approximation} and the dominated convergence theorem. 
	The first inequality is (\ref{remark:ReplacingDnDm}) and the second follows from \eqref{eq:EstimateProof1}. Together with Lemma~\ref{lemma:Beigl1}, we arrive at
	\begin{align}\label{eq:ProofFinal}
		\mathrm{MV}(S^\tau,D_n)\leq 2C+6,\quad n\in\bbn.
	\end{align}
	Finally, by the right-continuity of $S^\tau$ and \eqref{eq:ProofFinal}, we get
	\begin{align*}
		\mathrm{MV}(S^\tau)=\lim\limits_{n\to\infty}\mathrm{MV}(S^\tau,D_n)\leq 2C+6. 
	\end{align*}
	Together with $\PM(\tau<\infty)\leq 3\varepsilon$, this establishes that the right-continuous process~$S$ is a  local quasimartingale and,
	thus, a semimartingale by Rao's theorem (in the version of Theorem~\ref{theo:rao}) and the Doob-Meyer-Decomposition~\cite[Chapter 3, Theorem 16]{protter2005stochastic}. \end{proof}
\begin{proof}[Proof of Theorem~\ref{theo:ResultSemimartingale}]
	Having shown that Theorem~\ref{theo:ExistenceSimple} holds the assertion follows directly by Proposition~\ref{prop: ReductionStep}.
\end{proof}
\begin{remark}
	The arguments presented here rely heavily on the two-dimensional setting. However, Theorem~2.7 can be directly applied to a model with 
	a bank account and finitely 
	many risky assets since in this case it is sufficient to have a semimartingale price system for each risky asset separately (cf. also \cite[Theorem~7.2]{delbaen1994general}). 
	On the other hand, 
	it seems that the approach cannot be adapted to the general Kabanov model (cf. Kabanov and 
	Safarian~\cite[Section 3.6]{kabanov.safarian.2009}), in which there need not exist a bank account that is involved in every transaction. 
\end{remark}
\section{The self-financing condition}\label{21.12.2019.02}

As already discussed in the introduction, we use the semimartingale 
to define the self-financing condition in the bid-ask model for general strategies. A self-financing condition can be identified with an operator $\vp\mapsto \Pi(\vp)$ that maps each amount of risky assets to the corresponding
position in the risk-less bank account (if the later exists). Here, we assume that the initial position and the risk-less interest are zero. 
In addition, {\bf for the rest of the paper, we assume that there exists 
	a semimartingale price system $S$}, i.e., $S$ is a semimartingale s.t. $\underline{S}\leq S\leq \overline{S}$ (cf. Theorem~\ref{theo:ResultSemimartingale}). The aim is to define $\Pi(\vp)$ as $\vp\mal S - \vp S - {\rm ``costs"}$, where the process~$\vp\mal S$ denotes the stochastic integral. At this stage, the process~$\vp$ is bounded (see Proposition~\ref{prop:Welldefined} for the extension to general strategies).
The costs are caused by the fact that the trades are carried out at the less favorable bid-ask prices.
Since the gains in the semimartingale are finite, they cannot compensate infinite costs and the latter lead to ruin.

\subsection{Construction of the cost term}

We construct the cost associated to a strategy $\varphi\in\bP$ path-by-path, i.e., in the following, $\omega\in\Omega$ is fixed and $\vp,\un{S},\ov{S},S$ are identified with functions in time. 

We follow a two-step procedure. First, we calculate the costs on intervals in which the left limit of the spread is bounded away from zero by means of a modified Riemann-Stieltjes integral. The integral turns out to always exist (but it can take the value $\infty$).  In the second step, we exhaust the set of points 
with positive spread by finite unions of such intervals and define the total costs as the supremum of the costs along these unions.
One may see a vague analogy between the second step and the way a Lebesgue integral is constructed. 

This approach leads to a well-founded self-financing condition under the additional Assumption~\ref{8.12.2019.3} on the behavior of the spread at zero. Very roughly speaking, there should not occur costs 
if the investor builds up positions  at times the spread is zero 
and the positions are already closed {\em before} the spread reaches any positive value (cf. Example~\ref{16.12.2019.1} for a counterexample). Since for the construction  of our cost process itself, the assumption is not needed, we introduce it later on. 

In order to introduce the integral, we need the following notation. 
\begin{definition}\label{def:Partition}
	Let $I=[a,b]\subseteq [0,T]$ with $a<b$.
	\begin{enumerate}[(i)]
		\item A collection $P=\{t_0,\dots t_n\}$ of points $t_i\in[a,b]$ for $n\in\bbn$ and $i=0,\dots,n$ with $a=t_0< t_1<\dots< t_n=b$ is called a \emph{partition of $I$}.
		\item A partition $P'=\{t'_0,\dots,t'_m\}$ with  $P'\supseteq P$ is called a \emph{refinement of $P$}.
		\item If $P,P'$ are two partitions of $I$, the \emph{common refinement} $P\cup P'$ is the partition obtained by ordering the points of $\{t_0,\dots t_n\}\cup\{t'_0,\dots,t'_m\}$ in increasing order.
		\item Given a partition $P=\{t_0,\dots,t_n\}$ of $I$ a collection $\lambda=\{s_1,\dots,s_n\}$ with $s_i\in [t_{i-1},t_i)$ for $i=1,\dots,n$ is called a \emph{modified intermediate subdivision of $P$}.
		\item Let $\varphi\in\bP$, $P=\{t_0,\dots,t_n\}$ be a partition of $I$ and $\lambda=\{s_1,\dots,s_n\}$ be an modified intermediate subdivision of $P$, the \emph{modified Riemann-Stieltjes sum} is defined by
		\begin{align*}
			R(\varphi,P,\lambda):=\sum_{i=1}^{n} (\overline{S}_{s_i}-S_{s_i})(\varphi_{t_i}-\varphi_{t_{i-1}})^++\sum_{i=1}^{n} (S_{s_i}-\underline{S}_{s_i})(\varphi_{t_i}-\varphi_{t_{i-1}})^-.
		\end{align*}
	\end{enumerate} 
\end{definition}
\begin{definition}\label{def:CostTermIntervall} Let $\varphi\in\bP$ and $I=[a,b]\subseteq [0,T]$ with $a<b$. The \emph{cost term of $\varphi$ on $I$} exists and equals $C(\varphi,I)\in\mathbb{R}_+\cup\{\infty\}$ if for all $\varepsilon>0$ there is a partition $P_\varepsilon$ of $I$ s.t. for all refinements $P$ of $P_\varepsilon$ and all modified intermediate subdivisions $\lambda$ of $P$ the following is satisfied:	
	\begin{enumerate}[(i)]
		\item In the case of $C(\varphi,I)<\infty$, we  have $\vert C(\varphi,I)-R(\varphi,P,\lambda)\vert<\varepsilon,$
		\item In the case of $C(\varphi,I)=\infty$, we have $\vert R(\varphi,P,\lambda)\vert>\frac1{\varepsilon}.$
	\end{enumerate}
	In addition, we set $C(\varphi,\{a\}):=0$ for all $a\in[0,T]$ and $C(\varphi,\emptyset):=0$.   
\end{definition}

The next proposition establishes the existence of the cost term on an interval $I$ where the spread is bounded away from zero. 
\begin{proposition}\label{prop:ExistenceCostTerm}
	Let $\varphi\in\bP$ and $I=[a,b]\subseteq [0,T]$ with $a<b$ s.t. $\inf_{t\in[a,b)}(\overline{S}_t-\underline{S}_t)>0$. Then, the cost term $C(\varphi,I)$ in Definition~\ref{def:CostTermIntervall} exists and is unique. In addition, we have
	\begin{align*}
		\begin{cases}
			C(\varphi,I)<\infty, & \text{if}\ \mathrm{Var}_a^b(\varphi)<\infty\\
			C(\varphi,I)=\infty, & \text{if}\ \mathrm{Var}_a^b(\varphi)=\infty,
		\end{cases}
	\end{align*}
	where $ \mathrm{Var}_a^b(\varphi)$ denotes the pathwise variation of $\vp$ on the interval $[a,b]$.	
\end{proposition}
We postpone the technical proof of Proposition~\ref{prop:ExistenceCostTerm} to Appendix~\ref{app:TechnicalCosts}.

\begin{remark}\label{remark:StieltjesFiniteVar}	
	First note that a priori, $\vp$ need not be of finite variation. Thus, we cannot decompose it into its
	increasing part~$\varphi^\uparrow$ and decreasing part~$\varphi^\downarrow$ to define 
	$\int_{a}^{b}(\overline{S}_s-S_s)d\varphi^\uparrow_s+\int_{a}^{b}(S_s-\underline{S}_s)d\varphi^\downarrow_s:=C(\varphi^\uparrow,[a,b])+C(\varphi^\downarrow,[a,b])$. Instead, we consider the increasing and decreasing parts of $\vp$ along grids and weight them with the corresponding prices before passing to the limit.   
	
	However, if $\mathrm{Var}_a^b(\varphi)<\infty$, it can be shown that   $C(\varphi^\uparrow,[a,b])+C(\varphi^\downarrow,[a,b])=C(\varphi,[a,b])$. This can be seen by an inspection of the proof of
	Proposition~\ref{prop:ExistenceCostTerm}, in which the condition
	$\inf_{t\in [a,b)}(\overline{S}_t-\underline{S}_t)>0$ can be dropped
	if $\mathrm{Var}_a^b(\varphi)<\infty$.
\end{remark}
\begin{remark}
	Definition~\ref{def:CostTermIntervall}(i) only requires that the cost term exists in the Moore-Pollard-Stieltjes-sense (see, e.g., Hildebrandt~\cite[Section 4]{hildebrandt1938definitions} and Mikosch and  Norvai{\v{s}}a~\cite[Section 2.3]{mikosch2000stochastic}), i.e., as the limit of the net $R(\varphi,\cdot,\cdot)$ indexed by the directed set of tuples $(P,\lambda)$ with the partial order $(P,\lambda)\geq (P',\lambda')$ iff $P$ is a refinement of $P'$.
	This is weaker than the existence in the norm-sense, i.e., as the limit of the net $R(\varphi,\cdot,\cdot)$ indexed by the tuples $(P,\lambda)$ with the partial order $(P,\lambda)\geq (P',\lambda')$ iff $\max_{i=1,\dots,n}(t_i-t_{i-1})\leq \max_{i=1,\dots,m}(t'_{i}-t'_{i-1})$, that is 
	guaranteed for the usual Riemann-Stieltjes integral with a continuous integrator of finite variation. A straightforward adaptation of the existence in the norm-sense of the usual 
	Riemann-Stieltjes integral to the present context would read: 
	
	The cost term is said to exist and equal to $C(\varphi,I)\in\mathbb{R}_+$ if for each $\varepsilon>0$ there is $\delta>0$ s.t. $\vert C(\varphi,I)-R(\varphi,P,\lambda)\vert<\varepsilon$ for all partitions $P=\{t_0,\dots,t_n\}$ with $\max_{i=1,\dots,n}(t_i-t_{i-1})<\delta$ and all intermediates subdivision $\lambda=\{s_1,\dots,s_n\}$ with $s_i\in[t_{i-1},t_i)$. 
	
	But, the following example, similar to Guasoni et al.~\cite[Example A.3]{guasoni2012fundamental} shows that $C(\vp,I)$ does in general not exist in the norm sense: let $T=2$, $\overline{S}-S=\mathbbm{1}_{[1,2]}$ and $\varphi=\mathbbm{1}_{(1,2]}$. Namely, if $t_i=1$ is not included in the partition~$P$, $R(\varphi,P,\lambda)$ can oscillate
	between $0$ and $1$.
	
	The example shows that the points of common discontinuities of integrator and integrand are critical to calculate the costs. Thus, they have to be included 
	in the partition, which is guaranteed by the Moore-Pollard-Stieltjes approach.
\end{remark}
\begin{remark}\label{remark:intermediate}
	The restriction that the point~$s_i$ of the intermediate 
	subdivision~$\lambda$ has to lie in the interval $[t_{i-1},t_{i})$, and not only in $[t_{i-1},t_{i}]$, has a clear financial interpretation.
	
	If an investor buys $\vp_s-\vp_{s-}$ shares at time $s$, she
	pays $(\vp_s-\vp_{s-})\ov{S}_{s-}$ monetary units. Consequently, 
	if she updates her position between $t_{i-1}$ and $t_i$, only the stock prices on the time interval~$[t_{i-1},t_{i})$ have to be considered. In the limit, the choice of the price in $[t_{i-1},t_{i})$ does not matter.
	Indeed, a well-known way to guarantee the existence of Riemann-Stieltjes integrals in the case of simultaneous jump discontinuities on the same side of integrator and integrand is to exclude the boundary points (see Hildebrandt~\cite[Section 6]{hildebrandt1938definitions}).
	
	Finally, we mention that in the case of $\mathrm{Var}_a^b(\varphi)<\infty$, the integrals are the same as in Guasoni et al.~\cite[Section A.2]{guasoni2012fundamental}. But, besides considering different processes, we introduce the integrals in a different way.
\end{remark}

The next proposition states that the cost term is additive with regard to the underlying interval. Its proof is obvious. 
\begin{proposition}\label{prop:SubIntervallProperty}
	Let $\varphi\in\bP$, $I=[a,b]\subseteq[0,T]$ s.t. $\inf_{t\in[a,b)}(\overline{S}_t-\underline{S}_t)>0$ and $c\in[a,b]$. Then, we have 
	\begin{align*}
		C(\varphi, [a,b])=C(\varphi, [a,c])+ C(\varphi, [c,b]).
	\end{align*}
\end{proposition}
Having defined the costs for all subintervals $I=[a,b]\subseteq [0,T]$ with $\inf_{t\in[a,b)}(\overline{S}_t-\underline{S}_t)>0$, we now proceed to define the accumulated costs as a process. Therefore, we let 
\begin{align}\label{21.2.2020.1}
	\mathcal{I}:=\left\{\cup_{i=1}^n[a_i,b_i]: \begin{array}{l}
		n\in\bbn,\ 0\leq a_1\le b_1\leq a_2\le \dots\leq a_n\le b_n\leq T,\\ \inf_{t\in [a_i,b_i)}(\overline{S}_t-\underline{S}_t)>0,\ i=1,\dots,n\end{array}\right\}\cup\{\emptyset\}.
\end{align}
We now extend the cost term to $\mathcal{I}$. Given $\varphi\in\bP$ and $J=\cup_{i=1}^n[a_i,b_i]$ with $\inf_{t\in [a_i,b_i)}(\overline{S}_t-\underline{S}_t)>0$ for all $i=1,\dots,n$, we define \emph{the costs along $J$} by 
\beam\label{31.12.2019.1}
C(\varphi, J):=\sum_{i=1}^nC(\varphi,[a_i,b_i]),
\eeam
where the cost terms $C(\varphi,[a_i,b_i])$ for $i=1,\dots,n$ are defined in Definition~\ref{def:CostTermIntervall}. By Proposition~\ref{prop:SubIntervallProperty}, the RHS of (\ref{31.12.2019.1}) does not depend on the representation of $J$. Thus, the cost term $C(\varphi,J)$ is well-defined for all $J\in\mathcal{I}$.  
\begin{definition}\label{def:PathwiseCost}(Cost process)
	Let $\varphi\in\bP$. Then, the \emph{cost process} $(C_t(\varphi))_{t\in[0,T]}$ is defined by
	\begin{align*}
		C_t(\varphi):=\sup_{J\in\mathcal{I}}C(\varphi,J\cap [0,t])\in[0,\infty], \quad t\in[0,T]
	\end{align*}
	(Note that $\{0\}\in \mathcal{I}$ with $C(\vp,\{0\})=0$ and thus the supremum is nonnegative). If it is clear from the context, we also write $(C_t)_{t\in[0,T]}$ for the cost process associated to $\varphi$. 
\end{definition}

\begin{proposition}\label{prop:PathwiseProperties}
	Let $\varphi\in\bP$. The cost process $(C_t(\varphi))_{t\in[0,T]}$ is $[0,\infty]$-valued, increasing and, consequently, l\`agl\`ad (if finite). In addition, the following assertions hold.
	\begin{enumerate}[(i)]
		\item For any $0\leq s\leq t\leq T$, we have $C_t(\varphi) =C_s(\varphi)+\sup_{J\in \mathcal{I}}C(\varphi, J\cap [s,t])$,
		\item For any $0\leq s\leq t\leq T$ with $\inf_{\tau\in [s,t)}(\overline{S}_\tau-\underline{S}_\tau)>0$, we have $C_t(\varphi)= C_s(\varphi)+C(\varphi, [s,t])$,
		\item For any $0\leq s\leq t\leq T$, we have $C_t(\varphi)\leq C_s(\varphi)+ \sup_{\tau\in[s,t)}(\overline{S}_\tau-\underline{S}_\tau)\mathrm{Var}_s^t(\varphi)$.
	\end{enumerate} 
\end{proposition}
The assertions above follow directly from Definitions~\ref{def:CostTermIntervall} and \ref{def:PathwiseCost}. Thus, we leave the easy proof to the reader.

The next proposition determines sequences of partitions whose corresponding Riemann-Stieltjes sums converge to the cost term on an interval where the spread is bounded away from zero. This will be crucial to show that the cost term is predictable. For this purpose, recall that the oscillation $\mathrm{osc}(f,I)$ of a function $f:[0,T]\to\mathbb{R}$ on an interval $I\subseteq[0,T]$ is defined by $\mathrm{osc}(f,I):=\sup\{\vert f(t)-f(s)\vert: s,t\in I\}$. 
\begin{proposition}\label{prop:approximationcost}
	Let $\varphi\in\bP$ and $I=[a,b]\subseteq [0,T]$ with $a<b$ and $\inf_{t\in[a,b)}(\overline{S}_t-\underline{S}_t)>0.$
	In addition, let $(P_n)_{n\in\mathbb{N}}$ be a refining sequence of partitions of $I$, i.e., $P_n=\{t_0^n,\dots,t_{m_n}^n\}$ with $a=t_0^n< t_1^n< \dots < t_{m_n}^n=b$ and $P_{n+1}\supseteq P_n$, s.t. \begin{enumerate}[(i)]
		\item $\lim\limits_{n\to\infty}\max(\sup\limits_{i=1,\dots, m_n}\mathrm{osc}(\overline{S}-S,[t^n_{i-1},t^n_i)), \sup\limits_{i=1,\dots, m_n}\mathrm{osc}(S-\underline{S},[t^n_{i-1},t^n_i)))=0$
		\item $\lim\limits_{n\to\infty}\sum_{i=1}^{m_n}\vert \varphi_{t^n_i}-\varphi_{t^n_{i-1}}\vert =\mathrm{Var}_a^b(\varphi)$.
	\end{enumerate}
	Then, for any sequence $\lambda_n=\{s_1^n,\dots,s_{m_n}^n\}$ of modified intermediate subdivision, we have \begin{align*}
		R(\varphi,P_n,\lambda_n) \to C(\varphi,[a,b])\quad \text{as}\ n\to\infty.
	\end{align*}
	In addition, such a sequence $(P_n)_{n\in\bbn}$ always exists. 
\end{proposition}
The proof of Proposition~\ref{prop:approximationcost} is closely related to the proof of Proposition~\ref{prop:ExistenceCostTerm}. Thus, we also postpone it to Appendix~\ref{app:TechnicalCosts}. We now conclude the subsection with a first approximation result. 
\begin{proposition}\label{prop:liminf}
	Let $\varphi,\varphi^n\in\bP$, $n\in\bbn$, $t\in[0,T]$ and $J\in\mathcal{I}$. Then, we have the implication
	\begin{align}\label{eq:liminfclaim1}
		&\varphi^n\to\varphi\quad \text{pointwise}\quad \Rightarrow\quad  \liminf_{n\to\infty}C(\varphi^n,J\cap [0,t])\geq C(\varphi,J\cap[0,t]).
	\end{align}
\end{proposition}
\begin{proof} 
	Let $\varphi^n\to\varphi$ pointwise and $t\in [0,T]$. We start by noting that the claim is trivial if $J=\{a\}$ for some $a\in[0,T]$ or if $J=\emptyset$. 
	
	\emph{Step 1.} We now treat the special case $J=[a,b]\in\mathcal{I}$ with $a<b$. In this case, we have $C(\varphi, J\cap [0,t])=C(\varphi,[a,b\wedge t])$ and 
	$C(\varphi^n, J\cap [0,t])=C(\varphi^n,[a,b\wedge t])$ for all $n\in\bbn$, where we use the convention $[c,d]=\emptyset$ if $d<c$. In addition, by the preceding observation, 
	we may assume $t>a$. 
	
	We only consider the case $C(\varphi,[a,b\wedge t])<\infty$ since the opposite case $C(\varphi,[a,b\wedge t])=\infty$ is analogous.  
	Let $\varepsilon>0$. There is a partition $P_\varepsilon=\{t_0,\dots,t_m\}$ of $[a,b\wedge t]$ s.t.
	\begin{align}\label{eq:liminf0}
		\begin{aligned}
			&\sum_{i=1}^{m} \inf_{s\in[t_{i-1},t_i)}(\overline{S}_{s}-S_{s})(\varphi_{t_i}-\varphi_{t_{i-1}})^++\sum_{i=1}^{m} \inf_{s\in[t_{i-1},t_i)}(S_{s}-\underline{S}_{s})(\varphi_{t_i}-\varphi_{t_{i-1}})^-\\\geq&\ C(\varphi,[a,b\wedge t])-\varepsilon.
		\end{aligned}
	\end{align}
	Using the pointwise convergence of $(\varphi^n)_{n\in\mathbb{N}}$, we can find $N\in\bbn$ s.t. for all $n\geq N$, we have
	\begin{align}\label{eq:liminf1}
		\begin{aligned}
			&\sum_{i=1}^{m} \inf_{s\in[t_{i-1},t_i)}(\overline{S}_{s}-S_{s})(\varphi^n_{t_i}-\varphi^n_{t_{i-1}})^++\sum_{i=1}^{m} \inf_{s\in[t_{i-1},t_i)}(S_{s}-\underline{S}_{s})(\varphi^n_{t_i}-\varphi^n_{t_{i-1}})^-\\\geq&\ C(\varphi,[a,b\wedge t])-2\varepsilon.
		\end{aligned}
	\end{align}
	
	Keeping this in mind, for each $n$, we choose a partition $\overline{P}_n$ s.t. for all refinements $P$ of $\overline{P}_n$ and intermediate subdivisions $\lambda$ of $P$, we have $C(\varphi^n,[a,b\wedge t])\geq R(\varphi^n,P,\lambda)-\varepsilon$. Now, we let $P_n:=P_\varepsilon\cup \overline{P}_n$ and write $P_n=\{t_0^n,\dots, t^n_{m_n}\}$. Denoting by $t_{i-1}=t^n_{i_1}< t^n_{i_{2}}<\dots < t^n_{i_j}=t_i$ the points of $P_n$ in between $t_{i-1}$ and $t_i$, we have 
	\begin{align*}
		\sum_{k=2}^{j}(\varphi^n_{t^n_{i_k}}-\varphi^n_{t^n_{i_{k-1}}})^+\geq (\varphi^n_{t_i}-\varphi^n_{t_{i-1}})^+\quad\text{and}\quad
		\sum_{k=2}^{j}(\varphi^n_{t^n_{i_k}}-\varphi^n_{t^n_{i_{k-1}}})^-\geq (\varphi^n_{t_i}-\varphi^n_{t_{i-1}})^-.
	\end{align*}
	Together with \eqref{eq:liminf1} this yields
	\begin{align*}
		C(\varphi^n,[a,b\wedge t])\geq R(\varphi^n,P_n,\lambda_n)-\varepsilon\geq C(\varphi,[a,b\wedge t])-3\varepsilon
	\end{align*}
	for all $n\geq N$ and intermediate subdivision $\lambda_n$ of $P_n$. Hence, we have
	\begin{align*}
		\liminf_{n\to\infty}C(\varphi^n,[a,b\wedge t])\geq C(\varphi,[a,b\wedge t])-3\varepsilon,
	\end{align*}
	which tantamount to the claim as $\varepsilon\downarrow0$. 
	
	\emph{Step 2.} Finally, let $J=\cup_{i=1}^m[a_i,b_i]\in\mathcal{I}$. Then, using the non-negativity of the sequences $(C(\varphi^n,[a_i,b_i\wedge t]))_{n\in\bbn}$ for $i=1,\dots,m$, we have 
	\begin{align*}
		\liminf_{n\to\infty} C(\varphi^n, J\cap [0,t])=\liminf_{n\to\infty}\sum_{i=1}^mC(\varphi^n, [a_i,b_i\wedge t])\geq \sum_{i=1}^m \liminf_{n\to\infty}C(\varphi^n, [a_i,b_i\wedge t]).
	\end{align*} 
	Thus, \eqref{eq:liminfclaim1} follows directly from step 1 and the observation at the start of the proof.
\end{proof}

\subsection{The cost term as a stochastic process}
Until now we kept $\omega\in\Omega$ fixed, i.e., the construction is path-by-path. To show some measurability properties of the cost term, we now consider it as a stochastic process. 
\begin{proposition}\label{prop:costprocessproperties}
	Let $\varphi\in\bP$. The cost process $C(\varphi)=(C_t(\varphi))_{t\in[0,T]}$ coincides with a predictable process up to evanescence.
\end{proposition}

In order to prove Proposition~\ref{prop:costprocessproperties}, we need the following lemma, whose proof relies on some deep results of Doob~\cite{doob1975stochastic} and thus is postponed to Appendix~\ref{app:TechnicalCosts}.
\begin{lemma}\label{lemma:PredictableIntervall}
	Let $\varphi\in\bP$ and $\sigma\le \tau$ two stopping times s.t. $\inf_{\sigma(\omega)\leq t<\tau(\omega)}(\overline{S}_t(\omega)-\underline{S}_t(\omega)>0$ for all $\omega\in\Omega$. Then, the process $C(\varphi, [\sigma\wedge \cdot,\tau\wedge \cdot])$ coincides with a predictable process up to evanescence. 
\end{lemma}

In order to establish Proposition~\ref{prop:costprocessproperties}, we still need to approximate the supremum in Definition~\ref{def:PathwiseCost} in a measurable way. Therefore, we define for each $n\in \mathbb{N}$ a sequence of stopping times by $\tau^n_0:=0$ and \begin{align} \label{eq:StoppingTimes}
	\tau^n_k:=\begin{cases}
		\inf\{t\geq \tau^n_{k-1}:\overline{S}_t-\underline{S}_t\leq 2^{-(n+1)}\},& k\ \mathrm{odd}\\
		\inf\{t> \tau^n_{k-1}: \overline{S}_t-\underline{S}_t\geq 2^{-n}\},& k\ \mathrm{even}
	\end{cases},\quad \mathrm{for}\ k\in\bbn.
\end{align}
Note that only a finite number of $\{\tau_k^n(\omega)\}_{k\in\mathbb{N}}$ is less than infinity as the process $\overline{S}-\underline{S}$ has c\`adl\`ag sample paths, $\tau^n_{2k}<\tau^n_{2k+1}$ on $\{\tau_{2k}<\infty\}$, and  
\begin{align*}
	\inf_{\tau^n_{2k}(\omega)\leq t<\tau^n_{2k+1}(\omega)} \left(\overline{S}_t(\omega)-\underline{S}_t(\omega)\right)\geq 2^{-(n+1)}\quad\ \text{for}\quad k\in\bbn_0
\end{align*} 
for all $\omega\in\Omega$. In particular, this means that the process $C^n(\varphi)=(C^n(\varphi)_t)_{t\in[0,T]}$  
\begin{align}\label{eq:ApproximationCost1}
	C^n_t(\varphi):=\sum_{k=0}^{\infty} C(\varphi, [\tau^n_{2k}\wedge t, \tau^n_{2k+1}\wedge t])
\end{align}
is well-defined and coincides with a predictable process up to evanescence for each $n\in\mathbb{N}$ by Lemma~\ref{lemma:PredictableIntervall}.
\begin{lemma}\label{lemma:approximationlemma}
	Let $\varphi\in\bP$ and $(C^n(\varphi))_{n\in\mathbb{N}}$ as above. Then, $C^n(\varphi)\to C(\varphi)$ pointwise. 
\end{lemma}
\begin{proof} We write $C^n$ instead of $C^n(\varphi)$ to not overburden the notation. Let $(\omega,t)\in\Omega\times[0,T]$. For $C_t(\omega)<\infty$,  we claim: for each $\varepsilon>0$ there is $N=N(\omega)\in\mathbb{N}$ s.t. \begin{align}\label{eq:ApproximationCost2}
		C_t(\omega)-\varepsilon\leq C^n_t(\omega)\leq C_t(\omega)\quad \text{for all}\quad n\geq N.
	\end{align} 
	Thus, let us prove \eqref{eq:ApproximationCost2}. It is obvious from Definitions~\ref{def:PathwiseCost} that we have $C_t(\omega)\geq C^n_t(\omega)$ for all $n\in\mathbb{N}$. To prove the other inequality, let $\varepsilon>0$ and choose $0\leq a_1<b_1\leq a_2<\dots\leq a_n<b_n\leq t$ s.t. $\inf_{t\in[a_i,b_i)}(\overline{S}_t(\omega)-\underline{S}_t(\omega))>0$ for $i=1,\dots, n$ and
	\begin{align}\label{eq:ApproximationCost3}
		C_t(\omega)-\varepsilon\leq \sum_{i=1}^{n}C(\varphi(\omega),[a_i,b_i])
	\end{align}   
	Let $\delta:=\min_{i=1,\dots,n}\inf_{t\in[a_i,b_i)}(\overline{S}_t(\omega)-\underline{S}_t(\omega))>0$ and choose $N\in\bbn$ s.t. $2^{-N}<\delta$. Then, it follows from the definition of the stopping times~\eqref{eq:StoppingTimes} that
	\begin{align*}
		\bigcup_{i=1}^n[a_i,b_i]\subseteq\bigcup_{k=0}^\infty [\tau^n_{2k}(\omega)\wedge t, \tau^n_{2k+1}(\omega)\wedge t]\quad \text{for all}\quad n\geq N. 
	\end{align*}
	Combining \eqref{eq:ApproximationCost3} this with Proposition~\ref{prop:SubIntervallProperty}, we find $C_t(\omega)-\varepsilon\leq C^n_t(\omega)$ for all $n\geq N,$ which proves (\ref{eq:ApproximationCost2}). Of course, for $C_t(\omega)=\infty$ the arguments are completely analogous. 
\end{proof}
\begin{proof}[Proof of Proposition~\ref{prop:costprocessproperties}]
	Applying Lemma~\ref{lemma:PredictableIntervall}, we find that $C^n$ 
	coincides with a predictable process up to evanescence. Together with Lemma~\ref{lemma:approximationlemma} this yields that $C$ does the same. 
\end{proof}

Next, we want to calculate the cost of an ``almost simple'' trading strategy (cf. Guasoni et al.~\cite{guasoni2012fundamental} for a detailed discussion).
\begin{definition}
	A predictable stochastic process $\varphi$ of finite variation is called an {\em almost simple strategy} if there is a sequence of stopping times $(\tau_n)_{n\geq 0}$ with $\tau_n<\tau_{n+1}$ on $\{\tau_n<\infty\}$ and $\#\{n: \tau_n(\omega)<\infty\}<\infty$ for all $\omega\in\Omega$, s.t. 
	\begin{align*}
		\varphi=\sum_{n=0}^{\infty}(\varphi_{\tau_n}\mathbbm{1}_{\llbracket \tau_n\rrbracket}+\varphi_{\tau_n+}\mathbbm{1}_{\rrbracket \tau_n,\tau_{n+1}\llbracket}).
	\end{align*} 
\end{definition}

\begin{proposition}\label{prop:SimpleStrat}
	Let $\varphi$ be an almost simple strategy. We have
	\begin{align*}
		C_t(\varphi)&=\sum_{n=0}^{\infty}\mathbbm{1}_{\{\tau_n\leq t\}}\left((\overline{S}_{\tau_n-}-S_{\tau_n-})(\varphi_{\tau_n}-\varphi_{\tau_n-})^++(S_{\tau_n-}-\underline{S}_{\tau_n-})(\varphi_{\tau_n}-\varphi_{\tau_n-})^-\right)\\
		&+\sum_{n=0}^{\infty}\mathbbm{1}_{\{\tau_n< t\}}\left((\overline{S}_{\tau_n}-S_{\tau_n})(\varphi_{\tau_n+}-\varphi_{\tau_n})^++(S_{\tau_n}-\underline{S}_{\tau_n})(\varphi_{\tau_n+}-\varphi_{\tau_n})^-\right)
	\end{align*}
	for all $t\in[0,T]$.
\end{proposition}
\begin{proof}
	For $\omega\in\Omega$ fixed, there is some $n\in\bbn_0$ with  
	$\tau_0(\omega)<\ldots<\tau_{n-1}(\omega)\le T$ and $\tau_n(\omega)=\infty$. Now, it is sufficient to consider partitions containing $\tau_i(\omega)-\delta,\tau_i(\omega)$ if $(\ov{S}_{\tau_i-}(\omega)  -S_{\tau_i-}(\omega))\wedge (S_{\tau_i-}(\omega)-\un{S}_{\tau_i-}(\omega))>0$ and
	$\tau_i(\omega),\tau_i(\omega)+\delta$ if $(\ov{S}_{\tau_i}(\omega)  -S_{\tau_i}(\omega))\wedge (S_{\tau_i}(\omega)-\un{S}_{\tau_i}(\omega))>0$
	for $i=0,\ldots,n-1$ and $\delta>0$ small. We leave the details to the reader.
\end{proof}

At last, we show how a $\varphi\in\bP$, which incurs finite cost on a stochastic interval where the spread is bounded away from zero, can be approximated by almost simple strategies on this interval s.t. the cost terms converges as well. 
\begin{proposition}\label{1.1.2020.2}
	Let $\varphi\in\bP$ and $\sigma\le \tau$ two stopping times s.t. \[\inf_{\sigma(\omega)\leq t<\tau(\omega)}(\overline{S}_t(\omega)-\underline{S}_t(\omega)>0\] for all $\omega\in\Omega$ and $C(\varphi,[\sigma\wedge T,\tau\wedge T])<\infty$ a.s.
	Then, there exists a uniformly bounded sequence $(\varphi^n)_{n\in\mathbb{N}}$ s.t. $\varphi^n\mathbbm{1}_{\rrbracket \sigma, \tau\rrbracket}$ is almost simple with $\varphi^n_\sigma=\varphi_\sigma$ on $\{\sigma<\infty\}$ and $\vert\varphi-\varphi^n\vert\leq 1/n$ on $\llbracket \sigma,\tau\rrbracket$ (up to evanescence) for all $n\in\mathbb{N}$, and s.t.
	\begin{align}\label{eq:ApproximationBV}
		\sup_{t\in[0,T]}\vert C(\varphi^n,[\sigma\wedge t, \tau\wedge t])-C(\varphi,[\sigma\wedge t, \tau\wedge t])\vert\to0,\quad\PM\mbox{-a.s.} 
	\end{align}
\end{proposition}
The proof is postponed to Appendix~\ref{app:TechnicalCosts}.
\subsection{Definition and characterization}
For the remainder of the paper, we make the following assumption on the bid-ask spread.

\begin{assumption}\label{8.12.2019.3}
	For every $(\omega,t)\in \Omega\times [0,T)$ with $\ov{S}_t(\omega)=\un{S}_t(\omega)$
	there exists an $\eps>0$ s.t.
	$\ov{S}_s(\omega)=\un{S}_s(\omega)$ for all $s\in(t,(t+\eps)\wedge T)$ or 
	$\ov{S}_s(\omega)>\un{S}_s(\omega)$ for all $s\in(t,(t+\eps)\wedge T)$.
\end{assumption}
This means that each zero of the path $t\mapsto \overline{S}_t(\omega)-\underline{S}_t(\omega)$ is either an inner point from the right 
of the zero set or a starting point of an excursion away from zero. This excludes, e.g., Brownian behavior of the spread, which is exploited in 
Example~\ref{16.12.2019.1}, where we show what can go wrong without this assumption.\\

For the rest of the paper, we work with the predictable versions of the cost processes (cf. Proposition~\ref{prop:costprocessproperties}), and identify processes that coincide up to evanescence.
Given a semimartingale~$S$, we define the operator $\Pi$ that maps a bounded, predictable strategy~$\vp$ starting at zero, i.e., $\vp\in\bP$, to the associated 
$[-\infty,\infty)$-valued risk-less position (also starting at zero) by 
\beam\label{6.1.2020.01}
\Pi_t(\vp) := \vp\mal S_t - \vp_t S_t - C_t(\vp),\quad t\in[0,T],
\eeam
which coincides with $\vp\mal S_{t-} - \vp_t S_{t-} - C_t(\vp)$.
Throughout the paper, $\vp\mal S$ denotes the \emph{standard} stochastic integral as defined by 
Definition~III.6.17 in \cite{jacod.shiryaev}. If stock positions are evaluated by the semimartingale~$S$,
the wealth process is given by $V_t(\vp):= \vp\mal S_t -C_t(\vp) =  \Pi_t(\vp) + \vp_t S_t$. If there is ambiguity about the semimartingale $S$ used in the construction, we write $C^S(\vp),\Pi^S(\vp),V^S(\vp)$ instead of $C(\vp),\Pi(\vp),V(\vp)$.

We still have to introduce a measure that gives some information about the convergence of integrals w.r.t. $S$. There exists a probability measure~$Q\sim \PM$ s.t. the semimartingale~$S$ possesses a decomposition $S=M+A$, where $M$ is a $Q$-square integrable martingale and $A$ is a process of $Q$-integrable variation
(Theorem~58 in Chapter~VII of Dellacherie and Meyer~\cite{dellacherie.meyer.1982}). We introduce the finite measure
\begin{align}\label{eq:Q.muS}
	\mu^S(B):= \E_{Q}\left(\mathbbm{1}_B\mal \langle M,M\rangle_T \right)
	+ \E_Q\left( \mathbbm{1}_B\mal {\rm Var}_T(A) \right),\quad B\in\mathcal{P},
\end{align}
where $\langle M,M\rangle$ denotes the predictable quadratic variation of $M$ (see, e.g., \cite[Chapter 1, Theorem 4.2]{jacod.shiryaev}).

The following theorem characterizes the process~$V(\vp)$ as the limit of
wealth processes associated with suitable almost simple strategies. 
Note that for almost simple strategies, $V$ coincides with the intuitive wealth process that can be written down without any limiting procedure. 
\begin{theorem}\label{8.12.2019.1}
	Let  $\vp\in \bP$ and let $\mu$ be a $\sigma$-finite measure on the predictable $\sigma$-algebra with $\mu^S\ll \mu$. 
	\begin{enumerate}[(i)]
		\item For all \{0,1\}-valued decreasing predictable processes $A$ and all uniformly bounded sequences of predictable processes~$(\vp^n)_{n\in\bbn}$, the following implication holds:
		\beao
		& & \vp^n\to \vp\ \mbox{pointwise on\ }\{\ov{S}_->\un{S}_-\}\cap \{A=1\},\ \mu^S\mbox{-a.e. on\ }\{\ov{S}_-=\un{S}_-\}\cap\{A=1\}\\
		& & \implies \liminf_{n\to\infty} V(\vp^n)\le V(\vp)\quad
		\mbox{on}\ \{A=1\}\ \mbox{up to evanescence.}
		\eeao 
		\item There exists a uniformly bounded sequence of almost simple 
		strategies~$(\vp^n)_{n\in\bbn}$ s.t.
		\beao
		& & \vp^n\to \vp\ \mbox{pointwise on\ }\{\ov{S}_->\un{S}_-\}\cap \{C(\vp)<\infty\},\  \mu\mbox{-a.e. on\ }\{\ov{S}_-=\un{S}_-\}{\cap \{C(\vp)<\infty\}},\\
		& & \mbox{and}\  \sup_{t\in[0,T]}\vert V_t(\vp^n)-V_t(\vp)\vert\mathbbm{1}_{\{C_t(\vp)\leq K\}}\to 0\ \mbox{in probability for}\ n\to\infty\ \mbox{and\ all}\ K\in\bbn.
		\eeao 
	\end{enumerate}
\end{theorem}
\begin{remark}
	In the special case $C(\vp)<\infty$, which is equivalent to $V(\vp)>-\infty$,  setting $A=1$ yields the following 
	characterization of the wealth process of a bounded strategy:
	(i) The wealth of the strategy exceeds the limiting wealth of (almost) pointwise converging simple strategies and (ii) there exists a special approximating  sequence s.t. the wealth processes converge. 
	
	On the set $\{V(\vp)=-\infty\}=\{C(\vp)=\infty\}$, one cannot expect the existence of a sequence of simple strategies that converge pointwise to $\vp$ on $\{\ov{S}_->\un{S}_-\}$.
	Nevertheless, Theorem~\ref{8.12.2019.1}(i) provides a motivation for $V(\vp)=-\infty$.
	
	For the proof of Corollary~\ref{cor:Independence}, we need the theorem in this general form, covering the case of infinite costs, since 
	{\em a priori} it is not  clear that the latter property does not depend on the choice of $S$.
\end{remark}
\begin{remark}
	In Theorem~\ref{8.12.2019.1}(i), one cannot expect convergence ``uniformly in probability'' as in the frictionless case. Indeed, consider $S=1$, $\ov{S}=2$, and $\vp^n = \mathbbm{1}_{\zu 1/n,1\zu}$ which  converges pointwise to $\vp = \mathbbm{1}_{\zu 0,1\zu}$ but $V(\vp^n)-V(\vp)=\mathbbm{1}_{\zu 0,1/n\zu}$.
\end{remark}
\begin{corollary} \label{cor:Independence} Let $\vp\in \bP$. The self-financing condition, i.e., the risk-less position~$\Pi(\vp)$, does not depend on the choice 
	of the semimartingale price system up to evanescence.  
\end{corollary}
\begin{proof}
	Let $\vp\in \bP$ and $S,\wt{S}$ be semimartingale price systems. Of course, the measure~$Q$ in \eqref{eq:Q.muS} can be chosen jointly for $S$ and $\widetilde{S}$ and w.l.o.g. $Q=\mathbb{P}$. We set $\mu := \mu^S + \mu^{\wt{S}}$. Let us fix $K\in\bbn$ and show that
	\beam\label{22.8.2020.1}
	\Pi^{\wt{S}}(\vp)\ge \Pi^S(\vp)\quad \mbox{on}\ \{C^S(\vp)\le K\}\ \mbox{up to evanescence.}
	\eeam
	Observe that (\ref{22.8.2020.1}) for all $K\in\bbn$ implies that $\Pi^{\wt{S}}(\vp)\ge \Pi^S(\vp)$ up to evanescence since $\Pi^S(\vp)=-\infty$ on 
	$\{C^S(\vp)=\infty\}=(\Omega\times[0,T])\setminus \cup_{K\in\bbn}\{C^S(\vp)\le K\}$.
	Then, the assertion of the corollary follows by symmetry. Thus, it is sufficient to show (\ref{22.8.2020.1}).
	
	For this, let $(\vp^n)_{n\in\bbn}$ be a sequence of almost simple strategies satisfying the properties in 
	Theorem~\ref{8.12.2019.1}(ii) for the semimartingale $S$ and $\mu$ given above. 
	According to Theorem~\ref{8.12.2019.1}(ii), we may suppose that  
	\begin{align}\label{eq:notinftyV}
		\sup_{t\in[0,T]}\vert V^S_t(\vp^n)-V^S_t(\vp)\vert\mathbbm{1}_{\{C^S_t(\vp)\leq K\}}\to 0 \quad \PM\text{-a.s.}
	\end{align}
	by passing to a subsequence. On the other hand, by applying 
	Theorem~\ref{8.12.2019.1}(i) with regard to the semimartingale~$\widetilde{S}$ and $A:=\mathbbm{1}_{\{C^S(\vp)\leq K\}}$, we get 
	\begin{align}\label{eq:-inftywtV}
		\liminf_{n\to\infty}V^{\wt{S}}(\vp^n)\le V^{\wt{S}}(\vp)\quad \mbox{on}\ \{C^S(\vp)\leq K\}\ \mbox{up to evanescence.}
	\end{align}
	In addition, Proposition~\ref{prop:SimpleStrat}
	and elementary calculations yield the assertion of the corollary for almost simple strategies, i.e.,
	\beam\label{22.8.2020.2}
	V^{\wt{S}}(\vp^n)=V^S(\vp^n)+\vp^n(\wt{S}-S),\quad n\in\bbn.
	\eeam
	It remains to analyze $(\vp^n-\vp)(\wt{S}-S)$, especially on $\{\ov{S}_-=\un{S}_-\}\cap\{\ov{S}>\un{S}\}$.
	If a sequence of c\`adl\`ag processes converges to zero uniformly in probability,
	the same holds for the associated jump processes. Thus,
	the choice of $\mu$ and the same arguments as in the proof of Theorem~\ref
	{8.12.2019.1}(i) yield 
	\begin{align}\label{eq:JumpConvergence}\begin{aligned}
			&\sup_{t\in[0,T]}\vert\vp^n_t\Delta S_t-\vp_t\Delta S_t\vert \mathbbm{1}_{\{\overline{S}_{t-}-\underline{S}_{t-}=0,\ C^S_t(\vp)<\infty\}}\to 0 \ \text{in probability for}\ n\to\infty\\ &\sup_{t\in[0,T]}\vert\vp^n_t\Delta \widetilde{S}_t-\vp_t\Delta \widetilde{S}_t\vert
			\mathbbm{1}_{\{\overline{S}_{t-}-\underline{S}_{t-}=0,\ C^S_t(\vp)<\infty\}}\to 0\ \text{in probability for}\ n\to\infty.
		\end{aligned}
	\end{align} 
	
	By passing to a further subsequence (again denoted by $(\vp^n)_{n\in\bbn}$), we can and do assume that the convergence in \eqref{eq:JumpConvergence} holds for $\mathbb{P}$-a.e. $\omega\in\Omega$. Thus, on $\{\ov{S}_-=\un{S}_-, C^S(\vp)<\infty\}$ we have $\vp^n(\widetilde{S}-S)=\vp^n(\wt{S}_{-}-S_{-})+\vp^n(\Delta\wt{S}-\Delta S)=\vp^n(\Delta\wt{S}-\Delta S)\to \vp(\Delta\wt{S}-\Delta S)=\vp(\widetilde{S}-S)$ up to evanescence. In addition, Theorem~\ref{8.12.2019.1}(ii) yields $\vp^n(\wt{S}-S)\to\vp(\wt{S}-S)$ on $\{\ov{S}_->\un{S}_-, C^S(\vp)<\infty\}$, i.e., we have $\vp^n(\wt{S}-S)\to\vp(\wt{S}-S)$ on $\{C^S(\vp)<\infty\}$ up to evanescence. Combining this 
	with \eqref{eq:notinftyV}, (\ref{eq:-inftywtV}), and (\ref{22.8.2020.2}) yields 
	\begin{align*}
		\Pi^{\wt{S}}(\vp)-\Pi^S(\vp)&= V^{\wt{S}}(\vp)-\vp \wt{S}-\left(V^{{S}}(\vp)-\vp{S}\right)
		\\&\geq\liminf_{n\to\infty}\left(V^{\wt{S}}(\vp^n)-V^{{S}}(\vp^n)-\vp^n\left( \wt{S}-{S}\right)\right)\\
		&=0\quad\mbox{on}\quad \{C^S(\vp)\leq K\}\ \mbox{up to an evanescence},
	\end{align*}
	and we are done. We note that the differences above are well-defined since
	$\Pi^S(\vp)$ and $V^{{S}}(\vp)$ are finite on $\{C^S(\vp)\leq K\}$. 
\end{proof}

The following example shows that our approach does not work without Assumption~\ref{8.12.2019.3}.
\begin{example}\label{16.12.2019.1}
	Let $\un{S}=-|B|+L^B$ and $\ov{S}=|B|+L^B$, where $B$ is a standard Brownian motion and $L^B$ its local time at zero in the
	sense of \cite[page 212]{protter2005stochastic}. Consider the strategy~$\vp:=\mathbbm{1}_{\{\un{S}=\ov{S}\}\cap(\Omega\times(0,T])}
	=\mathbbm{1}_{\{B=0\}\cap(\Omega\times(0,T])}$ and different semimartingale price systems~$S=\alpha|B|+L^B$ for $\alpha\in [-1,1]$. 
	By Definition~\ref{def:PathwiseCost}, we get $C(\vp)=0$. By \cite[Theorem~IV.69 and Corollary~3 of Theorem~IV.70]{protter2005stochastic},
	we have $\vp\mal S = (\alpha+1)L^B$. 
	Together this implies $\Pi(\vp) = (\alpha+1)L^B - \mathbbm{1}_{\{B=0\}} L^B$. 
	Since $L^B$ does not vanish, the self-financing condition would depend on the choice of $\alpha$.
\end{example}
\begin{corollary}\label{cor:ExtensionPrep}
	Let $\varphi\in \bP$ and $(\vp^n)_{n\in\bbn}$ be uniformly bounded. If $\vp^n\to \vp$ pointwise on $\{\overline{S}_->\underline{S}_-\}$ 
	and $\mu^S\text{-a.s.}$ on $\{\overline{S}_-=\underline{S}_-\}$, then there exists a deterministic subsequence~$(n_k)_{k\in\bbn}$ s.t. 
	\begin{align*}
		\lim\limits_{k\to\infty}(V(\vp^{n_k})-V(\vp))^+=0 \ \text{up to evanescence.}
	\end{align*} 
\end{corollary}
\begin{proof}
	The proof of Theorem~3.19~(i) shows that we have $\vp^n\bigcdot S\to \vp\bigcdot S$ uniformly in probability. 
	Hence, we can choose a subsequence $(n_k)_{k\in\mathbb{N}}$ s.t. $\vp^{n_k}\bigcdot S\to \vp\bigcdot S$ up to evanescence. 
	Finally, together with $\liminf_{k\to\infty} C(\vp^{n_k})\geq \liminf_{n\to\infty} C(\vp^n)\geq C(\vp)$ the assertion follows. 
\end{proof}
\section{Extension to unbounded strategies}\label{15.7.2021.2}

Let $(\bP)^\Pi:=\{\vp\in\bP\ :\ \Pi(\varphi)>-\infty\ \mbox{up to evanescence}\}$. Note that by the preceding corollary this set does not depend on the semimartingale price system. In this section, we want to extend the self-financing condition, i.e., the operator $\Pi$ from $(\bP)^\Pi$ to an as large as possible set of predictable strategies. Therefore, recall that the space of adapted l\`{a}dl\`{a}g processes $\mathcal{L}$ endowed with the topology of uniform convergence in probability, which is defined by the 
quasinorm~$\Vert X\Vert_{up}=\mathbb{E}\left[\sup_{t\in[0,T]}\vert X_t\vert \wedge 1\right],\ X\in\mathcal{L}$, is a complete metric space with metric $d_{up}(X,Y):=\Vert X-Y\Vert_{up}$ for $X,Y\in\mathcal{L}$. Indeed, this is a consequence of the completeness of the space of l\`{a}dl\`{a}g functions (also called regulated functions) equipped with the supremum norm (see, e.g., Fraňková~\cite[Point 1.8]{Frankova}). In addition, if $(X^n)_{n\in\bbn}\subseteq \mathcal{L}$ converges to $X\in\mathcal{L}$ with regard to $d_{up}$, we write $\uplim_{n\to\infty}X^n=X$. At this step, the restriction from $\bP$ to $(\bP)^\Pi$
is not critical since the latter is sufficiently large to approximate finite portfolio
processes, in which we are finally interested, in a reasonable way.

\begin{definition}\label{def:DefinitionL}
	Let $L$ denote the subset of real-valued, predictable strategies $\varphi$ s.t. there exists a sequence $(\varphi^n)_{n\in\bbn}\subset(\bP)^\Pi$ with 
	\begin{enumerate}[(i)]
		\item $\vp^n\to\vp$ pointwise on $\Omega\times [0,T]$ and $(\vp^n)^+\leq \vp^+$, $(\vp^n)^-\leq \vp^-$ for all $n\in\bbn$,
		\item there exists a semimartingale $S$ with $\underline{S}\leq S\leq \overline{S}$ s.t.
		\begin{align*}
			(V^S(\vp^n))_{n\in\bbn}=(\vp^n\bigcdot S- C^S(\vp^n))_{n\in\bbn}
		\end{align*}
		is Cauchy in $(\mathcal{L}, d_{up})$ and s.t. for all sequences $(\widetilde{\vp}^n)_{n\in\bbn}\subseteq (\bP)^\Pi$ satisfying (i), 
		there exists a deterministic subsequence $(n_k)_{k\in\mathbb{N}}$ s.t.
		\begin{align}\label{23.8.2020.1}
			\left(V^S(\widetilde{\vp}^{n_k})-V^S(\vp^{n_k})\right)^+ \to 0,\quad k\to\infty,\ \mbox{up to evanescence.} 
		\end{align}
	\end{enumerate}
\end{definition}

The requirement (ii) means that in the limit, the approximation with $(\vp^n)_{n\in\bbn}$ is better than all other pointwise 
approximations $(\widetilde{\vp}^n)_{n\in\bbn}$ if the stock position is evaluated by the same semimartingale. In (\ref{23.8.2020.1}), we cannot expect
uniform convergence in time, but exceptional $\PM$-null sets can be chosen independently of time.
By Corollary~\ref{cor:ExtensionPrep}, we have $(\bP)^\Pi\subseteq L$.

\begin{proposition}\label{prop:Welldefined}
	Let $\vp\in L$. If $(\vp^n)_{n\in\bbn}\subseteq (\bP)^\Pi$ is a sequence satisfying the assertions of Definition \ref{def:DefinitionL} for $\vp$ with regard to a semimartingales $S$ and $(\widetilde{\vp}^n)_{n\in\bbn}\subseteq (\bP)^\Pi$ is another sequence satisfying the same assertions for $\vp$ with regard to a semimartingale $\widetilde{S}$, then we have
	\begin{align*}
		\uplim_{n\to\infty} V^S(\vp^n)-\vp S=\uplim_{n\to\infty} V^{\wt{S}}(\widetilde{\vp}^n)-\vp \widetilde{S}
	\end{align*}
	up to evanescence.
\end{proposition}
We now can extend the operator $\Pi$ to $L$ by setting
\begin{align*}
	\Pi(\vp):=\uplim_{n\to\infty}V^S(\vp^n)-\vp S,\quad \vp\in L,
\end{align*}
where $(\vp^n)_{n\in\mathbb{N}}$ is a sequence satisfying the assertions of Definition~\ref{def:DefinitionL} with regard to the semimartingale $S$. By Proposition~\ref{prop:Welldefined}, $\Pi$ is well-defined on $L$, i.e., it does not depend on the choice of the approximating sequence and the semimartingale. 
\begin{proof}[Proof of Proposition \ref{prop:Welldefined}]
	Let $(\vp^n)_{n\in\bbn}$ and $(\widetilde{\vp}^n)_{n\in\bbn}$ be sequences that satisfy the assumptions of the proposition.
	Corollary \ref{cor:Independence} states that the process $\Pi(\wt{\vp}^n)$ does not depend on the semimartingale, i.e., we have
	\begin{align}\label{19.7.2021.1}
		V^S(\wt{\vp}^n)-\wt{\vp}^nS=V^{\wt{S}}(\wt{\vp}^n)-\wt{\vp}^n\wt{S} \quad \text{up to evanescence for all}\ n\in\bbn,
	\end{align}
	and thus
	\begin{align} \nonumber
		&\left(V^{\wt{S}}(\wt{\vp}^n)-\wt{\vp}^n\wt{S}-\left(V^S({\vp}^n)-{\vp}^nS\right)\right)^+ =\left(V^S(\wt{\vp}^n)-V^S(\vp^n)+(\vp^n-\wt{\vp}^n)S\right)^+\\ \label{eq:Welldefined1} &\leq \left(V^S(\wt{\vp}^n)-V^S(\vp^n)\right)^++\left((\vp^n-\wt{\vp}^n)S\right)^+
	\end{align}
	up to evanescence for all $n\in\bbn$. We have that $\vp^n\to\vp$ and $\wt{\vp}^n\to
	\vp$ pointwise as $n\to \infty$. We may pass to a subsequence s.t. $((V^S(\wt{\vp}^n)-V^S(\vp^n))^+)_{n\in\bbn}$ converges to zero pointwise up to evanescence by (\ref{23.8.2020.1}). In addition, we may further pass to subsequences, s.t. $(V^{\wt{S}}(\wt{\vp}^n))_{n\in\bbn}$, $(V^S(\vp^n))_{n\in\bbn}$ converge pointwise up to evanescence. Thus, by symmetry, (\ref{eq:Welldefined1}) yields the assertion.
\end{proof}

\subsection{Frictionless markets}

We now turn towards the frictionless case, i.e., $\overline{S}=\underline{S}=S$,
and show that $L$ equals the set $L(S)$ of $S$-integrable processes:
\begin{proposition} \label{prop:FrictionlessL}
	Let $\underline{S}=\overline{S}=S$ be a semimartingale. Then, we have $L=L(S)$ and $\Pi(\vp)=\vp\bigcdot S -\vp S$ for all $\vp\in L$.  
\end{proposition}
The set $L(S)$ was introduced as given in Definition~III.6.17 of \cite{jacod.shiryaev} by Jacob~\cite{jacob}, but there are  equivalent definitions that may look a bit smarter and that are based on $\bP\subseteq L(S)$.
For this, recall that the space of semimartingales~$\mathbb{S}$ endowed with the semimartingale topology defined by the metric
\begin{align}\label{eq:dSemimartingaleTop}
	d_{\mathbb{S}}(X,Y):=\sup_{H\in \bP,\ \Vert H\Vert_{\infty}\leq 1}\Vert H\mal (X-Y)\Vert_{up},\quad X,Y\in \mathbb{S}
\end{align}
is a complete metric space by Émery~\cite[Theorem 1]{Emery}. The following characterization of $S$-integrability is effectively due to Chou et al.~\cite{Chou}.
\begin{note}\label{note:ChouCharacterization}
	Let $S$ be a semimartingale and $\vp$ be a predictable process. The following 
	assertions are equivalent
	\begin{enumerate}[(i)]
		\item $\vp\in L(S)$.
		\item There exists a sequence $(\vp^n)_{n\in\bbn}\subseteq \bP$ s.t. $\vp^n\to\vp$ pointwise, 
		$(\vp^n)^+ \leq \vp^+$, $(\vp^n)^- \le \vp^-$ for all $n\in\bbn$, and $(\vp^n\mal S)_{n\in\bbn}$ is Cauchy in $(\mathbb{S}, d_{\mathbb{S}})$.
		\item For all sequences $(\vp^n)_{n\in\bbn}\subseteq \bP$ with $\vp^n\to\vp$ pointwise and $\vert \vp^n\vert \leq \vert \vp\vert$ for all $n\in\bbn$, the sequence $(\vp^n\mal S)_{n\in\bbn}$ is Cauchy in $(\mathbb{S}, d_{\mathbb{S}})$.
	\end{enumerate}
	In this case, the integral~$\vp\mal S$ is given by the $d_{\mathbb{S}}$-limit of 
	any such sequence~$(\vp^n\mal S)_{n\in\bbn}$.
\end{note}
\begin{proof}[Proof of Note~\ref{note:ChouCharacterization}]
	In the definition on page 130, Chou et al.~\cite{Chou} (see also \cite[Chapter VIII, 75]{dellacherie.meyer.1982}) introduce the special approximating
	sequence $\vp^n:= \vp\mathbbm{1}_{\{|\vp|\le n\}}$ for some predictable process $\vp$. Later on, the only
	properties of $(\vp^n)_{n\in\bbn}$ they use is that $\vp^n\in\bP$ for $n\in\bbn$, $|\vp^n|\le |\vp|$ for $n\in\bbn$, and $\vp^n\to \vp$ pointwise. Thus, the note is just a reformulation of their results \cite[Properties b), c), d) on page 130 and Theoreme 1]{Chou} (see also \cite[Chapter VIII, 74-77]{dellacherie.meyer.1982})
\end{proof}

A similar characterization is provided in Eberlein and Kallsen~\cite{eberlein.kallsen}, page~193 by 
\beao
L(S)=\{ \vp\ {\rm predictable}: \exists\ {\rm semimartingale\ } Z\ {\rm s.t.\ }
(\vp \mathbbm{1}_{\{|\vp|\le n\}})\mal S = \mathbbm{1}_{\{|\vp|\le n\}}\mal Z,\ n\in\bbn\}.
\eeao
It emphasizes the maximality of $L(S)$ if one {\em requires} that the integral~$\vp\mal S:=Z$ itself is a semimartingale.
By contrast, in our characterization from Definition~\ref{def:DefinitionL}, the semimartingale property can be seen more as a result since it is stated with the up-metric and not with the semimartingale metric. 
\begin{proof}[Proof of Proposition~\ref{prop:FrictionlessL}] 
	\emph{Ad $L(S)\subseteq L$:} This follows from (i)$\Rightarrow$ (ii)$\Rightarrow$(iii) in Note~\ref{note:ChouCharacterization}.
	
	\emph{Ad $L\subseteq L(S)$:} Let $\vp\in L$.  Thus, there exists $(\vp^n)_{n\in\bbn}\subseteq\bP$ satisfying Definition~\ref{def:DefinitionL}(i) and (ii). In particular, the sequence $(V^S(\vp^n))_{n\in\bbn}=(\vp^n\bigcdot S)_{n\in\bbn}$ is Cauchy with regard to $d_{up}$. Let us demonstrate that the sequence is also Cauchy in $(\mathbb{S},d_{\mathbb{S}})$ by contradiction, i.e., we assume that there exists $\varepsilon>0$, a sequence $(H^n)_{n\in\bbn}$ of predictable processes with $0\leq H^n\leq 1$ for all $n\in\bbn$ and a subsequence $(m_n)_{n\in\bbn}$ with $m_n\geq n$ s.t.
	\begin{align}\label{eq:frictionlessL1}
		\PM\left(\left(\left(H^n\left(\vp^n-\vp^{m_n}\right)\bigcdot S\right)\right)^*_T>\varepsilon\right)>\varepsilon, \quad \forall n\in\bbn
	\end{align} 
	We note that in \eqref{eq:frictionlessL1}, it can be assumed that $H^n$ is $[0,1]$-valued and not only $[-1,1]$-valued, since otherwise it can be decomposed into its positive and its negative part.
	Next, we define the strategies $\psi^n:=H^n\vp^n+(1-H^n)\vp^{m_n}\in\bP$  and $\theta^n:=(1-H^n)\vp^n+H^n\vp^{m_n}\in\bP$ for $n\in\bbn$. The strategies satisfy $\psi^n\to \vp$, $\theta^n\to \vp$ pointwise and $(\psi^n)^+ \vee (\theta^n)^+ \leq \vp^+,(\psi^n)^- \vee (\theta^n)^- \leq \vp^- $, i.e., they satisfy Definition \ref{def:DefinitionL} (i). 
	
	Let $\sigma^n:=\inf\{t\geq 0: \psi^n\bigcdot S_t-\vp^n\bigcdot S_t>\varepsilon/2 \}$ and $\tau^n:=\inf\{t\geq 0: \theta^n\bigcdot S_t-\vp^n\bigcdot S_t>\varepsilon/2 \}$.  As $(\vp^n-\vp^{m_n})\bigcdot S\to 0$ uniformly in probability by Definition \ref{def:DefinitionL} (ii), there is $N\in\mathbb{N}$ s.t.  $\PM(((\vp^n-\vp^{m_n})\bigcdot S)^*_T>\varepsilon/2)<\varepsilon/2$ for all $n\geq N$. Thus, we have
	\begin{align*}
		\PM\left(\sigma^n\wedge \tau ^n\leq T\right)&\geq \PM\left(\left(\left(H^n\left(\vp^n-\vp^{m_n}\right)\bigcdot S\right)\right)^*_T>\varepsilon\right)-\PM\left(\left(\left(\vp^n-\vp^{m_n}\right)\bigcdot S\right)^*_T>\varepsilon/2\right)\\&>\varepsilon/2\quad \forall n\geq N.
	\end{align*}
	Next, we define the strategies $\widetilde{\psi}^n:=\psi^n\mathbbm{1}_{\llbracket 0, \sigma^n\rrbracket}+\vp^n\mathbbm{1}_{\rrbracket \sigma^n, T\rrbracket}$ and $\widetilde{\theta}^n:=\theta^n\mathbbm{1}_{\llbracket 0, \tau^n\rrbracket}+\vp^n\mathbbm{1}_{\rrbracket \tau^n, T\rrbracket}$ that still satisfy Definition \ref{def:DefinitionL} (i).  Thus, together with
	\begin{align*}
		\PM\left(\left\{ \wt{\psi}^n\bigcdot S_T-\vp^n\bigcdot S_T>\varepsilon/2\right\}\cup \left\{ \wt{\theta}^n\bigcdot S_T-\vp^n\bigcdot S_T>\varepsilon/2\right\}\right)\geq \PM(\sigma^n\wedge \tau^n\leq T)>\varepsilon/2
	\end{align*}
	for all $n\geq N$, we have arrived at a contradiction to \eqref{23.8.2020.1}.  Thus $(\vp^n\bigcdot S)_{n\in\bbn}$ is Cauchy in $(\mathbb{S}, d_\mathbb{S})$ and the assertion follows by $(ii)\Rightarrow (i)$ in Note~\ref{note:ChouCharacterization}.
\end{proof}

One of the referees raised the following interesting question that can be considered as a generalization of Proposition~\ref{prop:FrictionlessL} to markets with friction.
Does $\vp\in L$ imply that there exists a semimartingale price system~$S$ s.t. $\vp\in L(S)$? This would mean, if stock positions are evaluated by $S$, the trading gains and the 
cost term of the approximating bounded strategies converge separately (and not only the sum).  

Under additional assumptions, the following theorem gives a positive answer to this question. Especially, the considered model is deterministic, see Remark~\ref{22.7.2021.1} below for 
a discussion. 
\begin{theorem}\label{15.7.2021.1}
	Let $\Omega=\{\omega\}$ and $\ov{S},\un{S}$ be continuous. If $\vp\in L$, $\vp>0$ on $(0,T]$, and $\vp$ is lower semi-continuous at all $t\in[0,T]$ with $\ov{S}_t>\un{S}_t$,
	then there exists a semimartingale price system~$S$ with $\vp\in L(S)$.
\end{theorem}
\begin{proof}
	We fix a semimartingale price system~$\wt{S}$ (whose existence is assumed in this section).
	
	{\em Step 1:} Let us show that
	\beam\label{22.7.2021.2}
	\sup_{\psi\ \mbox{\small bounded},\ 0\le \psi \le \vp} V^{\wt{S}}_T(\psi)<\infty.
	\eeam
	Assume by contradiction that there exist bounded strategies $\psi^n$, $n\in\bbn$ s.t.
	$0\le \psi^n\le \vp$ and $V^{\wt{S}}_T(\psi^n)\to\infty$. On the other hand, since $\vp\in L$ and by \eqref{19.7.2021.1},
	there exist bounded $\vp^n$,\ $n\in\bbn$ with $0\le \vp^n\le \vp$, $\vp^n\to \vp$, and
	$V^{\wt{S}}_T(\vp^n)\to V^{\wt{S}}_T(\vp)\in\bbr$. Thus, there is a null
	sequence~$(\eps_n)_{n\in\bbn}\subset (0,1)$ s.t.
	\beao
	V^{\wt{S}}_T(\eps_n \psi^n +(1-\eps_n)\vp^n) \ge  \eps_n V^{\wt{S}}_T(\psi^n)
	+ (1-\eps_n) V^{\wt{S}}_T(\vp^n)\to \infty,
	\eeao
	which is a contradiction to $\vp\in L$.
	
	{\em Step 2:} Next, we show that for each nonnegative bounded function $\wt{\psi}$, 
	\beam\label{16.5.2021.2}
	\sup_{0\le \psi \le \wt{\psi}} V^{\wt{S}}_T(\psi)
	\eeam 
	is attained by a maximizer~$\psi^*$. To see this, let $({\psi}^n)_{n\in\bbn}$ be a maximizing sequence, i.e.,  
	$V^{\wt{S}}_T({\psi}^n)\to \sup_{0\le \psi \le \wt{\psi}} V^{\wt{S}}_T({\psi})$. Since ${\psi}^n\mal \wt{S}_T\le \sup_{t\in[0,T]}\wt{\psi}_t \cdot{\rm Var}_T(\wt{S})$ 
	for all $n\in\bbn$, the sequence of cost terms~$(C^{\wt{S}}_T({\psi}^n))_{n\in\bbn}$ is bounded. In addition, the set $\{\ov{S}>\un{S}\}$ can be written as a countable union of closed intervals
	on which either
	$\wt{S}\ge \un{S} + 1/3(\ov{S}-\un{S})$ or  $\wt{S}\le \un{S} + 2/3(\ov{S}-\un{S})$. In the first case,
	sells lead to essential costs on such an interval~$[a,b]$. Consequently, one must have
	$\sup_{n\in\bbn}{\rm Var}_a^b({\psi}^n)<\infty$. Then, by the same arguments as in Campi and
	Schachermayer~\cite{campi.schachermayer}, proof of Proposition~3.4, after passing to convex combinations,
	we obtain a pointwise limit~$\lim_{n\to\infty}{\psi}^n=:{\psi}^*$
	everywhere on $\{\ov{S}>\un{S}\}$ and ${\rm
		Var}(\wt{S})$-a.e. on $\{\ov{S}=\un{S}\}$, which has to be a maximizer by Theorem~\ref{8.12.2019.1}(i).

	{\em Step 3:}
	We now construct a sequence $(\wh{\vp}^n)_{n\in\bbn}$ s.t. $\wh{\vp}^n$ is a solution of \eqref{16.5.2021.2} with $\wt{\psi}=\vp\wedge n$ for all $n\in\bbn$ and for $n<m$ 
	the strategy $\wh{\vp}^{m}$ has to ``buy/sell'' if $\wh{\vp}^{n}$ ``buys/sells''. 
	
	Starting with solutions $\wh{\eta}^k$ of (\ref{16.5.2021.2}) with $\wt{\psi}=(\vp-(k-1))^+\wedge 1$ for each $k\in\bbn$, we define the strategies 
	$\eta^{n,k}:=\left(\sum_{l=1}^{n}\wh{\eta}^l-(k-1)\right)^+\wedge 1$ for $n\in\bbn$ and $k\leq n$. We have 
	\begin{align*}
		\sum_{k=1}^nV^{\wt{S}}_T(\eta^{n,k})=V^{\wt{S}}_T\left(\sum_{k=1}^n \eta^{n,k}\right)=V^{\wt{S}}_T\left(\sum_{k=1}^n \wh{\eta}^k\right)\geq \sum_{k=1}^nV^{\wt{S}}_T(\wh{\eta}^k).
	\end{align*}
	Indeed, $V_T^{\wt{S}}(\cdot)$ is superadditiv and additive for $\eta^{n,k}$, $k=1,\dots,n$. The later can be seen by the additivity of the cost term for approximating simple 
	strategies. Together with $V_T^{\wt{S}}(\wh{\eta}^{k})\geq V^{\wt{S}}_T(\eta^{n,k})$ for all $k\leq n$, this implies $V^{\wt{S}}_T(\wh{\eta}^{k})=V^{\wt{S}}_T(\eta^{n,k})$ for 
	all $n\in\bbn$ and $k\leq n$. Defining $\eta^k:=\lim_{n\to\infty} \eta^{n,k}=(\sum_{l=1}^{\infty}\wh{\eta}^l-(k-1))^+\wedge 1,\ k\in\bbn,$
	we observe $\eta^k=0$ on $\{\eta^{k-1}<1\}$ and $\eta^k\leq(\vp - (k-1))\wedge 1$. In addition, we have 
	$V^{\wt{S}}_T(\eta^k)\geq\lim_{n\to\infty}V^{\wt{S}}_T(\eta^{n,k})= V^{\wt{S}}_T(\wh{\eta}^k)$ by Theorem~\ref{8.12.2019.1}~(i)
	and, thus, $\eta^k$ solves (\ref{16.5.2021.2}) with $\wt{\psi}=(\vp - (k-1))\wedge 1$. 
	Finally, we set	$\wh{\vp}^n:=\sum_{k=1}^n\eta^k$, $n\in\bbn$.
	Then, for an arbitrary strategy~${\psi}$ with $\psi\le \vp\wedge n$, the optimality of 
	$\eta^k$ yields $V^{\wt{S}}_T({\psi}) = \sum_{k=1}^n V^{\wt{S}}_T(({\psi}-(k-1))^+\wedge 1)\leq \sum_{k=1}^nV^{\wt{S}}_T(\eta^k)=V^{\wt{S}}_T(\wh{\vp}^n)$, 
	i.e., $\wh{\vp}^n$ solves \eqref{16.5.2021.2} with $\wt{\psi}=\vp\wedge n$.
	
	{\em Step 4:} Let $(\wh{\vp}^n)_{n\in\bbn}$ be the sequence of maximizers from the previous step. Since short positions are forbidden, we can replace $\un{S}_T$
	by $\wt{S}_T$ and assume that positions are sold at $T$.
	The aim is to construct a 
	finite variation process~$S$ s.t. $V_T^{\wt{S}}(\wh{\vp}^n)=\wh{\vp}^n\mal S_T$ and ${\psi}\mal S_T\leq \wh{\vp}^n\mal S_T$ for all strategies $0\leq {\psi}\leq \vp\wedge n$, i.e., 
	$S$ is a shadow price simultaneously for all problems \eqref{16.5.2021.2} with $\wt{\psi}=\vp\wedge n$, $n\in\bbn$. 
	Under Assumption~\ref{8.12.2019.3} and by an exhaustion argument, it
	is possible to construct $S$ in the following way. On the frictionless intervals, cf. Lemma~\ref{25.11.2019.2},
	$S$ is defined as 
	$S=\ov{S}=\un{S}$.
	Now, let $a$ be a ``buying time'' with $\ov{S}_a>\un{S}_a$, i.e., there exists $n\in\bbn$ s.t.
	in any neighborhood of $a$ there are $t_1<t_2$ with $\wh{\vp}^n_{t_2}>\wh{\vp}^n_{t_1}$. 
	Let $b$ be the next selling time (defined as infimum over $n\in\bbn$), and $d$ the next 
	buying time after $b$. In addition, $c$ is the last selling time before $d$. 
	We have that $a<b\le c\le d$. The strict inequality is crucial for the exhaustion argument. It holds since, by $\ov{S}_a>\un{S}_a$ and the continuity of the bid-ask prices, 
	any investment needs some time to amortize, and by Step~3, for any pair of buying and selling time, there is a joint strategy~$\wh{\vp}^n$ that realizes this investment.
	Summing up, all $\wh{\vp}^n$,\ $n\in\bbn$, are nondecreasing on $(a,b)$, nonincreasing on $(b,c)$, and constant on $(c,d)$.
	
	For $t\in[a,b)$, we define
	\beam\label{16.5.2021.5}
	\tau_t:=\inf\{s\in [a,t] : \exists\eps>0\ \inf_{u\in(s,t+\eps)}\vp_u > \inf_{u\in(t,b)}\vp_u\}\wedge t
	\eeam
	and
	\beao
	S_t:=\inf_{u\in[\tau_t,b)}\ov{S}_u \wedge \un{S}_b.
	\eeao
	Roughly speaking, $S$ can only increase at a ``bottleneck'' on the way to $b$, where the constraint is binding. 
	For $t\in[b,c)$, we define
	\beao
	\sigma_t:=\sup\{s\in [t,c) : \forall \eps>0\ \inf_{u\in(t+\eps,s)}\vp_u > \inf_{u\in[b,t]}\vp_u\}\vee t
	\eeao
	and
	\beam\label{25.7.2021.1}
	S_t:=\sup_{u\in[b,\sigma_t]}\un{S}_u.
	\eeam
	For $[c,d)$, $c<d$, we make a case differentiation. For $\wh{\vp}^1=0$ on $(c,d)$, 
	we define $S$ on $[c,d)$ as the Snell envelope of the process~$L_t:=\un{S}_t \mathbbm{1}_{\{t<d\}} + \ov{S}_{d}\mathbbm{1}_{\{t=d\}}$,\ $t\in[c,d]$,
	i.e., $S_t:=\sup_{u\in[t,d]}L_u$,\ $t\in[c,d)]$.
	Otherwise, we define 
	$S_t:= \un{S}_c \mathbbm{1}_{\{t<\wt{\tau}_d\}} + \ov{S}_d \mathbbm{1}_{\{t\ge \wt{\tau}_d\}}$, where 
	$\wt{\tau}_d:=\inf\{s\in [c,d] : \inf_{u\in(s,d)}\vp_u > \inf_{u\in(d,\wt{b})}\vp_u\}\wedge d$ with $\wt{b}$ being the next selling time after $d$. 
	By using the maximality and the monotonicity of all $\wh{\vp}^n$,\ $n\in\bbn$, it is easy to check that $S$ has to lie in the bid-ask spread.
	
	Now, any excursion of the spread away from zero, cf. Lemma~\ref{25.11.2019.1}, can be exhausted by intervals of the form $[a,b)$, $[b,c)$, and $[c,d)$. 
	In the special case that there is no further buying time, (\ref{25.7.2021.1}) is applied to the closed interval from $b$ to the end of 
	the excursion of the spread away from zero or to $T$.
	The resulting process~$S$ is c\`adl\`ag 
	and does not depend on the choice of the intervals. Note that $\ov{S}_a>\un{S}_a$ is only needed to guarantee that $b>a$.   
	
	\emph{Step 5:} Let us show that $S$ is of finite variation and 
	\begin{align}\label{eq:Shadowprice_prop1}
		\wh{\vp}^n\mal S_T = V^S_T(\wh{\vp}^n) = V^{\wt{S}}_T(\wh{\vp}^n),\quad n\in\bbn.
	\end{align}
	Let $a$ be a buying time and $\wt{a}$ be the time~$\inf\{t>a\ :\ \wh{\vp}^1_t=0\}$ truncated at the end of the excursion. 
	We have that $S_a=\ov{S}_a\ge \wt{S}_a$ and $S_{\wt{a}}=\un{S}_{\wt{a}}\le \wt{S}_{\wt{a}}$, and $S$ is nondecreasing on $[a,\wt{a}]$.
	From $\wt{a}$ up to (and including) the next buying time, $S$ is nonincreasing. This yields ${\rm Var}_T(S)\le {\rm Var}_T(
	\wt{S})<\infty$. Finally, by construction of $S$, the cost terms~$C^S(\wh{\vp}^n)$ vanish for all $n\in\bbn$ and thus \eqref{eq:Shadowprice_prop1} holds. 
	E.g., on $[a,b)$, the process~$\wh{\vp}^n$ is nondecreasing and has to be constant on $\{S<\ov{S}\}$ by optimality.
	
	\emph{Step 6:} Next, we show that 
	\beam\label{31.7.2021.1}
	{\psi}\mal S_T\le \wh{\vp}^n\mal S_T\ \mbox{for all $n\in\bbn$ and all strategies~${\psi}$ with $0\le \psi\le \vp\wedge n$}.
	\eeam
	Of course, it is sufficient to show this assertion for excursions of the spread away from zero (cf., again, Lemma~\ref{25.11.2019.1}).
	
	From now on, we need the assumed lower semi-continuity, i.e.,
	\beam\label{26.7.2021.1}
	\vp_t=\lim_{\varepsilon\to 0}\inf_{u\in[t-\varepsilon,t+\varepsilon]}\vp_u\quad\mbox{for all $t\in(0,T)$ with $\ov{S}_t>\un{S}_t$.}
	\eeam
	We start with the buying period, i.e., the interval $[a,b)$ (cf. Step~4). Setting $\xi_t:=\inf_{u\in[t,b)} \varphi_u$, we claim that
	\begin{align}\label{eq:Step6claim}
		\int_{[a,b)}{\psi}_t\,dS_t\leq \int_{[a,b)}(\vp_t\wedge n)\,dS_t\leq \int_{[a,b)}(\xi_t\wedge n)\,dS_t\leq\int_{[a,b)}\wh{\vp}^n_t\,dS_t
	\end{align}
	for every strategy ${\psi}$ with ${\psi}\leq \vp\wedge n$.
	
	The first inequality is obvious as $S$ is nondecreasing on $[a,b)$. We start by showing the second inequality in \eqref{eq:Step6claim}. 
	It follows from (\ref{26.7.2021.1}) that $(\xi_t)_{t\in[a,b)}$ is left-continuous and the set $\{t\in [a,b): \xi_t<\vp_t\}$ is open.
	Hence, we find a sequence of open intervals $(u^k_1,u^k_2)$, $u^k_1\leq u^k_2$, $k\in\bbn$ s.t.
	\begin{align}\label{eq:opencover}
		\{t\in [a,b): \xi_t<\vp_t\}=\bigcup_{k\in\bbn}(u^1_k,u^2_k).
	\end{align}
	For all $t_1,t_2$ with $u^k_1<t_1<t_2<u^k_2$, we have that $\inf_{t\in[t_1,t_2]}(\vp_t-\xi_t)>0$ and, thus, $S_{t_2}=S_{t_1}$. This yields
	$S_{u^k_2-}=S_{u_1^k}$ if $u_1^k<u_2^k$ and, hence, $\int_{[a,b)}(\vp_t\wedge n)\,dS_t=\int_{[a,b)}(\xi_t\wedge n)\,dS_t$ due to \eqref{eq:opencover}.
	
	Moving towards the last inequality in \eqref{eq:Step6claim}, we exclude the trivial case that $\ov{S}_a=\un{S}_b$. For a given $\varepsilon>0$, there is a partition
	$a=t_0< t_1<\dots< t_m=b$ s.t.
	\begin{align}\label{eq:buyingphase_step2_1}
		\int_{[a,b)}(\xi_t\wedge n)\,dS_t\leq \sum_{i=1}^{m-1}(\xi_{t_{i-1}}\wedge n)(S_{t_i}-S_{t_{i-1}})+ (\xi_{t_{m-1}}\wedge n)(S_{b-}-S_{t_{m-1}}) +\varepsilon
	\end{align}
	by \cite[Theorem~II.21]{protter2005stochastic} and the left-continuity of $\xi$. Let $s:=\sup\{u>a: \overline{S}_u<\underline{S}_b\}\le b$.
	Next, we define a perturbation $\wh{\vp}^{n,p}$ of the optimal strategy $\wh{\vp}^n$ in the bid-ask model, which approximately realizes the gains on the RHS of 
	\eqref{eq:buyingphase_step2_1} on $[a,b)$. 
	We set $\wh{\vp}^{n,p}=\wh{\vp}^n$ on $[0,a)\cup [s,T]$ and construct $\wh{\vp}^{n,p}$ on $[a,s)$ by iteratively specifying possible purchases. 
	At time~$t_0=a$, we buy until we reach $\wh{\vp}^{n,p}_a:=\xi_{t_0}\wedge n \ge \wh{\vp}^n_a$, paying price~$\ov{S}_a=S_a$ per share (time $t_0$ has the special property that 
	it is a ``buying time'' in the sense of Step~4). We proceed as follows:  if $S_{t_1}<S_{t_2}$ (which is equivalent to $\inf_{u\in[\tau_{t_{1}},\tau_{t_{2}})}\overline{S}_u<S_{t_2}$ 
	and, in this case, $S_{t_1}=\inf_{u\in[\tau_{t_1},\tau_{t_2})}\overline{S}_u$), we buy until we reach $\xi_{t_{1}}\wedge n$ shares at 
	time~$t_1^*:=\arg\min_{u\in [\tau_{t_{1}}, \tau_{t_2})}\overline{S}_u$. Hereby, 
	we have $\overline{S}_{t_1^*}<S_{t_2}\leq \underline{S}_b$, i.e., $t_1^*< s$, and, since $t_1^*\ge \tau_{t_1}$, the constraint $\vp\wedge n$ is also satisfied. 
	This is repeated for 
	the intervals $[\tau_{t_{i-1}},\tau_{t_i})$ for $i=3,\dots, m$. Since purchasing prices are strictly below $\un{S}_b$, in the bid-ask market, purchases take place on $[a,s)$. 
	For $s<b$, we have $\widehat{\vp}^{n,p}_{s-}\le \xi_s\wedge n = \widehat{\vp}^n_s$, where the equality follows from the optimality of $\wh{\vp}^n$ and (\ref{26.7.2021.1}). 
	Finally, the missing position $\widehat{\vp}^n_s - \widehat{\vp}^{n,p}_{s-}\ge 0$ is purchased at price $\ov{S}_s=\un{S}_b$ if $s<b$.
	In the case $s=b$, we must have $\un{S}_b=\ov{S}_b$ and need not care about the sign of the missing position.  
	Hence, the optimality of $\wh{\vp}^n$, together with 
	$V_T^{\wt{S}}(\wh{\vp}^n)- V_T^{{\wt{S}}}(\wh{\vp}^{n,p})=V_T^{S}(\wh{\vp}^n)- V_T^{S}(\wh{\vp}^{n,p})$, yields
	\begin{align}
		\label{eq:buyingphase_step2_2}
		0\leq V_T^{S}(\wh{\vp}^n)- V_T^{S}(\wh{\vp}^{n,p})\leq \int_{[a,b)}\wh{\vp}^n_tdS_t-\int_{[a,b)}(\xi_t\wedge n)dS_t+\varepsilon,
	\end{align}
	where for the second inequality we use \eqref{eq:buyingphase_step2_1} and the fact that $\wh{\vp}^{n,p}$ does not produce any costs w.r.t. $S$. 
	\eqref{eq:buyingphase_step2_2} implies the last inequality in \eqref{eq:Step6claim} as the $\varepsilon>0$ is arbitrary. 
	
	It  remains to show $\psi_t\,dS_t \le \wh{\vp}^n_t\,dS_t$ on sets other than $[a,b)$.
	After a time reversal, the proof for a selling interval~$[b,c)$ is the same as for a buying interval~$[a,b)$. Namely, w.l.o.g. we assume that $\un{S}_c>\un{S}_b$
	and consider an approximation similar to (\ref{eq:buyingphase_step2_1}) ``backward in time'' (the last point is $b-$ with $S_{b-}=\un{S}_b$). 
	Time~$s$ from above is replaced by $\wt{s}:=\inf\{u>b\ :\ \un{S}_u>\un{S}_b\}\le c$. 
	From the optimality of $\wh{\vp}^n$, the assumption that $b$ is a selling time in the sense of Step~4, and (\ref{26.7.2021.1}), it follows that 
	$\wh{\vp}^n_{b-}\ge \inf_{u\in[b,\wt{s}]}\vp_u\wedge n$.
	We leave the details as an exercise for the reader. On intervals with $\wh{\vp}^1=0$, we use that the Snell envelope is nonincreasing.
	
	{\em Step 7:} By $\vp\in L$ and \eqref{19.7.2021.1}, we can find a sequence of strategies~$(\vp^n)_{n\in\bbn}$ with $\vp^n\to \vp$ and $0\le \vp^n \le \vp\wedge n$ s.t. for all 
	other strategies~$(\wt{\vp}^n)_{n\in\bbn}$ with $\wt{\vp}^n\to \vp$ and
	$0\le \wt{\vp}^n\le \vp\wedge n$, one has $(V^S_T(\wt{\vp}^n) - V^S_T(\vp^n))^+ \to 0$. Let us show that 
	$(\vp^n\mal S)_{n\in\bbn}$ has to be Cauchy in $(\mathbb{S}, d_{\mathbb{S}})$. We first show that
	\beam\label{16.5.2021.3}
	\forall\eps>0\ \exists K\in\bbr_+\ \forall n\in\bbn, B\in\mathcal{B}([0,T])\qquad 
	(\mathbbm{1}_{\{\vp>K\}\cap B} \vp^n)\mal S_T\le \eps.
	\eeam
	Indeed, since $S$ is a shadow price, see (\ref{31.7.2021.1}), and by (\ref{22.7.2021.2}), we have
	\beam\label{20.7.2021.1}
	(\mathbbm{1}_{\{\vp>K\}\cap B} \vp^n)\mal S_T & \le & (\mathbbm{1}_{\{\vp>K\}} \wh{\vp}^n)\mal S_T
	\le \sum_{k=1}^\infty ((\mathbbm{1}_{\{\vp>K\}}\eta^k)\mal S_T)<\infty
	\eeam
	for all $K\in\bbr_+$ and $B\in\mathcal{B}([0,T])$.
	By (\ref{20.7.2021.1}), $(\mathbbm{1}_{\{\vp>K\}} \eta^k)\mal S_T \le \eta^k\mal S_T$ (which follows from (\ref{31.7.2021.1})), 
	and dominated convergence, we obtain (\ref{16.5.2021.3}). Let us show that
	\beam\label{16.5.2021.4}
	\forall\eps>0\ \exists K\in\bbr_+, N\in\bbn\ \forall n\geq N, B\in\mathcal{B}([0,T])\  
	(\mathbbm{1}_{\{\vp>K\}\cap B} \vp^n)\mal S_T\ge -\eps.
	\eeam
	Assume by contradiction that there exists $\eps>0$, a subsequence $(n_k)_{k\in\bbn}$, and a 
	sequence~$(B_k)_{k\in\bbn}\subset\mathcal{B}([0,T])$ s.t. 
	$(\mathbbm{1}_{\{\vp>k\}\cap B_k} \vp^{n_k})\mal S_T< -\eps$ for all $k\in\bbn$. On the other hand, since $d_{\mathbb{S}}(\mathbbm{1}_{\{\vp>k\}}\mal S, 0)\to 0$ for $k\to\infty$,  
	there must 
	exist a sequence $(\lambda_k)_{k\in\bbn} \subset \bbr_+$ with $\lambda_k\to\infty$ slowly enough s.t. $\mathbbm{1}_{\{\vp>k\}\cap B_k}(\vp^{n_k}\wedge \lambda_k)\mal S_T\to 0$
	for $k\to\infty$. Thus, we have $(\mathbbm{1}_{\{\vp>k\}\cap B_k} (\vp^{n_k}-\lambda_k)^+)\mal S_T< -\eps/2$ for $k$ large enough. 
	As in (\ref{20.7.2021.1}), we can estimate $(\mathbbm{1}_{[0,T]\setminus(\{\vp>k\}\cap B_k)}(\vp^{n_k}-\lambda_k)^+)\mal S_T = 
	(\mathbbm{1}_{\{\vp>\lambda_k\}\setminus(\{\vp>k\}\cap B_k)}(\vp^{n_k}-\lambda_k)^+)\mal S_T \le \sum_{l=1}^\infty ((\mathbbm{1}_{\{\vp>\lambda_k\}}\eta^l)\mal S_T)$,
	which converge to $0$ as $\lambda_k\to \infty$ for $k\to\infty$. This yields that
	$((\vp^{n_k}-\lambda_k)^+)\mal S_T< -\eps/4$ for $k$ large enough. Since the cost term of 
	$\vp^{n_k}$ 
	exceeds those of $\vp^{n_k}\wedge \lambda_k$, we arrive at 
	$V^S_T(\vp^{n_k})< V^S_T(\vp^{n_k}\wedge\lambda_k)-\eps/4$ for $k$ 
	large enough. This is a contradiction to the maximality of $(\vp^n)_{n\in\bbn}$ stated at the beginning of this step. Thus, (\ref{16.5.2021.4}) holds.
	
	Putting (\ref{16.5.2021.3}), (\ref{16.5.2021.4}), and $\vp^n\to\vp$ with $\vp^n\le \vp$ for all $n\in\bbn$ 
	together, we obtain that $(\vp^n\mal S)_{n\in\bbn}$ is Cauchy in $(\mathbb{S}, d_{\mathbb{S}})$. This implies that $\vp\in L(S)$ (cf. Note~\ref{note:ChouCharacterization}).
\end{proof}

\begin{remark}\label{22.7.2021.1}
	The proof demonstrates how the maximality condition in the definition of $L$ works. For $\vp\in L$, problem~(\ref{22.7.2021.2}) has to be finite, but its 
	maximizer~$\wh{\vp}:=\lim_{n\to \infty}\wh{\vp}^n$ can be different from $\vp=\lim_{n\to\infty}\vp^n$. Also, in the frictionless shadow price market, $\wh{\vp}^n$ dominates 
	all other strategies that are bounded by $\vp\wedge n$. This upper bound is key to show that $\vp^n\mal S$ is Cauchy w.r.t. the semimartingale topology.
	
	It is an open (but possibly insolvable) problem whether the theorem also holds in the general stochastic case. The construction of the shadow price~$S$ is essentially 
	based on the assumptions that the 
	model is deterministic and $\vp$ is lower semi-continuous. The latter is needed since on the intervals with friction,
	$S$ has its upward movements at the ``bottlenecks'' of the constraint~$\vp\wedge n$.
	
	Nevertheless, we think that the proof already provides the basic intuition for the relation between $L$ and $L(S)$ in the general stochastic case.
	In addition, the sequence of strategies constructed in Step~3 and the ideas from Step~7 should also be of general use to solve related problems in the stochastic model.
	By contrast, the other assumptions are less essential. They are made to focus on the main ideas and to avoid further case differentiations and technicalities.  
\end{remark}
\section{Proof of Theorem~\ref{8.12.2019.1}}\label{4.1.2020.1} 

We start with two lemmas that prepare the proof of Theorem~\ref{8.12.2019.1}.
In the following, we set $X:=\ov{S}-\un{S}$ with the convention that $X_{0-}:=0$. Let $M$ be the set of starting points of excursions of the spread away from zero, i.e.,
\begin{align*}
	M:=(\{X=0\}\cup\{X_-=0\})\cap\{(\omega,t)\in \Omega\times [0,T) : \exists \eps>0\ \forall s\in(t,(t+\eps)\wedge T)\ X_s(\omega)>0 \}.
\end{align*}
Here, we follow the convention that an excursion also ends (and thus a new excursion can start) if only the left limit of the spread process is zero.
Under the usual conditions and Assumption~\ref{8.12.2019.3}, the 
process~$Y:=\mathbbm{1}_{\{(\omega,t)\in\Omega\times[0,T) : \exists\eps>0\ \forall s\in(t,(t+\eps)\wedge T)\ X_s(\omega)>0\}}$ is right-continuous on $\Omega\times [0,T)$ and adapted (for the latter one uses that for all $t\in[0,T)$ and $\wt{\eps}\in(0,T-t)$, one has
$\{\omega\in\Omega : \exists\eps>0\  \forall s\in(t,(t+\eps)\wedge T)\ X_s(\omega)>0\} = 
\Omega\setminus \{\omega\in\Omega : \exists\eps\in(0,\wt{\eps})\cap\bbq\ \forall s\in(t,t+\eps)\cap 
\bbq\ X_s(\omega)=0\}$). Thus, $Y$ is a progressive process (see, e.g., Theorem~3.11 in \cite{he.wang.yan.1992}), which implies that $M$ is a progressive set.  Consequently, $\{\omega\in\Omega\ :\ \tau(\omega)<\infty,\ (\omega,\tau(\omega))\not\in M\}\in\mathcal{F}$ if $\tau$ is a stopping time. 

For a stopping time $\tau$, we define the associated stopping time $\Gamma_2(\tau)$ by 
\begin{align*}
	\Gamma_2(\tau):=\inf\{t>\tau\ :\ X_t =0\ \mbox{or}\ X_{t-}=0\}.
\end{align*}

\begin{lemma}\label{25.11.2019.1}
	There exists a sequence of stopping times $(\tau_1^n)_{n\in\bbn}$ with 
	$\PM(\{\omega\in\Omega : \tau_1^n(\omega)<\infty,\ (\omega,\tau_1^n(\omega))\not\in M\})=0$  for all $n\in\bbn$, $\PM(\tau^{n_1}_1=\tau^{n_2}_1<\infty)=0$ for all $n_1\not=n_2$, and 
	\beam\label{14.11.2019.1}
	\{X_->0\}\subset \cup_{n\in\bbn} \zu \tau_1^n,\Gamma_2(\tau_1^n)\zu\quad\mbox{up to evanescence}.
	\eeam
\end{lemma}
\begin{proof}
	We define a finite measure~$\mu$ on the predictable $\sigma$-algebra by $$\mu(A):= \sum_{k=1}^\infty 2^{-k}\PM(\{\omega\in\Omega : (\omega,q_k)\in  A\}),\quad A\in\mathcal{P},$$ where $(q_k)_{k\in\bbn}$ is a counting of the rational numbers.
	Let $\mathcal{M}$ be the set of predictable processes of the form $\mathbbm{1}_{\zu \tau,\Gamma_2(\tau)\zu}$, where $\tau$ runs through all stopping times satisfying $\PM(\{\omega\in\Omega : \tau(\omega)<\infty,\ (\omega,\tau(\omega))\not\in M\})=0$.
	The essential supremum of $\mathcal{M}$ w.r.t. $\mu$ can be written as
	\beao
	{\rm esssup}\ \mathcal{M} =  \sup_{n\in\bbn} \mathbbm{1}_{\zu \tau_1^n,\tau_2^n\zu}
	= \mathbbm{1}_{\cup_{n\in\bbn} \zu \tau_1^n,\tau_2^n\zu}\quad \mu\mbox{-a.e.},
	\eeao
	where $\tau^n_2=\Gamma_2(\tau^n_1)$. Obviously, the sequence $(\tau_1^n)_{n\in\bbn}$ can be chosen s.t. $\PM(\tau^{n_1}_1=\tau^{n_2}_1<\infty)=0$ holds for all $n_1\not=n_2$. Then, by the definition of $M$ and $\Gamma_2$, one has that
	$\zu \tau^{n_1}_1,\tau^{n_1}_2\zu\cap\zu\tau^{n_2}_1,\tau^{n_2}_2\zu=\emptyset$ 
	up to evanescence for all $n_1\not=n_2$.
	
	Now consider the random time~$\sigma:=\inf\{t\in( 0,T] : X_{t-}>0\ \mbox{and}\ t\not\in\cup_{n\in\bbn}(\tau_1^n,\tau_2^n]\}$. 
	Since $\sigma$ can be written as the debut $\inf\{t\in(0,T] : Z_t>0\}$, where 
	$Z:=X_-(1-\sum_{n=1}^\infty\mathbbm{1}_{\zu \tau_1^n,\tau_2^n\zu})$ is a finite predictable process, it is a stopping time (see Theorem~7.3.4 in \cite{cohen.elliott}). By the definition of the infimum and $\Gamma_2$, we must have that 
	$X_\sigma=0$ or $X_{\sigma-}=0$ on the set~$\{\sigma<\infty\}$. 
	Together with Assumption~\ref{8.12.2019.3}, 
	this means that in $\sigma$ there starts an excursion, and it is not yet overlapped. By the definition of the essential supremum, one has $\mu(\zu \sigma,\Gamma_2(\sigma)\zu)=0$. Since $\Gamma_2(\sigma)>\sigma$ on $\{\sigma<\infty\}$, this is only possible if $\PM(\sigma<\infty)=0$
	and thus $\PM(\{\omega\in\Omega : \exists t\in(0,T]\ X_{t-}(\omega)>0\ \mbox{and}\  t\not\in\cup_{n\in\bbn}(\tau_1^n(\omega),\tau_2^n(\omega)]\})=0$.
\end{proof}
Next, we analyze the time the spread spends at zero. Define
\begin{align*}
	&M_1:=\{(\omega,t)\in\Omega\times[0,T] : t=0\ \mbox{or}\ \forall\eps>0\ \exists s\in((t-\eps)\vee 0,t)\ X_s(\omega)>0\}\cap \{X_-=0\}\ \\ 
	&\text{and}\ M_2:=\{X_->0\}\cap\{X=0\}.
\end{align*}
The optional set $M_1\cup M_2$ consists of the ending points of an excursion and of
their accumulation points. For a stopping time~$\tau$, we define the starting point of the next excursion after $\tau$ by $(\Gamma_1(\tau))(\omega):=\inf\{t\ge \tau(\omega): (\omega,t)\in M\}$ for $\omega\in\Omega$,
which is the debut of a progressive set and thus a  stopping time by \cite[Theorem~7.3.4]{cohen.elliott}.
\begin{lemma}\label{25.11.2019.2}
	There exists a sequence of stopping times $(\sigma_1^n)_{n\in\bbn}$ with 
	$\PM(\{\omega\in\Omega: \sigma_1^n(\omega)<\infty,\ (\omega,\sigma_1^n(\omega))\not\in M_1\cup M_2\})=0$ s.t. $(\sigma^n_1)_{\{X_{\sigma^n_1-}=0\}}$ are predictable 
	stopping times for all $n\in\bbn$, $\PM(\sigma^{n_1}_1=\sigma^{n_2}_1<\infty)=0$ for all $n_1\not=n_2$, and 
	\beam\label{25.11.2019.3}
	\{X_-=0\} \subset \cup_{n\in\bbn} \left( 
	\auf(\sigma^n_1)_{\{X_{\sigma^n_1-}=0\}}\zu
	\cup
	\zu \sigma_1^n,\Gamma_1(\sigma_1^n)\zu\right)\quad\mbox{up to evanescence}.
	\eeam
	for $\Gamma_1$ from above.
\end{lemma}
(\ref{25.11.2019.3}) can be interpreted as follows. If the spread approaches zero continuously at some time~$t$, the investment between $t-$ and $t$ already falls into the ``frictionless regime''. On the other hand, if the spread jumps to zero at time $t$, the frictionless regime only starts immediately after $t$ (if at all). 
\begin{proof}[Proof of Lemma~\ref{25.11.2019.2}]
	We take the starting points~$\tau^n_1$ of the excursions from 
	Lemma~\ref{25.11.2019.1} and define the measure~$\mu(A):=\sum_{n=1}^\infty 2^{-n}\PM(\{\omega\in\Omega : (\omega,\tau^n_1(\omega))\in A\}) + \PM(\{\omega\in\Omega : (\omega,T)\in A\})$ for all $A\in\mathcal{P}$. Consider the essential supremum w.r.t. $\mu$ of the set of predictable processes~$\mathbbm{1}_{\auf\sigma_{\{X_{\sigma-}=0\}}\zu
		\cup
		\zu \sigma,\Gamma_1(\sigma)\zu}$, where $\sigma$ runs through the set of stopping times satisfying $\PM(\{\omega\in\Omega : \sigma(\omega)<\infty,\ (\omega,\sigma(\omega))\not\in M_1\cup M_2\})=0$ with the further constraint that $\sigma_{\{X_{\sigma-}=0\}}$ is a predictable stopping time.
	Again, the supremum can be written as
	\beao
	\mathbbm{1}_{\cup_{n\in\bbn} \left( 
		\auf(\sigma^n_1)_{\{X_{\sigma^n_1-}=0\}}\zu
		\cup
		\zu \sigma_1^n,\Gamma_1(\sigma_1^n)\zu\right)}\quad \mu\mbox{-a.e.}
	\eeao
	Consider the random time
	\beam\label{25.1.2020.1}
	\sigma:=\inf\{t\ge 0 : X_{t-}=0\ \mbox{and}\ t\not\in\cup_{n\in\bbn}\left(
	[(\sigma^n_1)_{\{X_{\sigma_1^n-}=0\}}]
	\cup (\sigma_1^n,\sigma_2^n]\right)\},
	\eeam
	where $\sigma_2^n:=\Gamma_1(\sigma_1^n)$. Since $\sigma=\inf\{t\ge 0 : Z_t=0\}$, where $$Z:=X_- + \sum_{n=1}^\infty \mathbbm{1}_{
		\auf(\sigma^n_1)_{\{X_{\sigma^n_1-}=0\}}\zu\cup
		\zu \sigma_1^n,\sigma_2^n\zu}$$ is predictable, $\sigma$ is a stopping time (see Theorem~7.3.4 in \cite{cohen.elliott}). In addition, one has
	\begin{align*}
		\auf\sigma_{\{X_{\sigma-}=0\}}\zu &= \auf\sigma\zu\cap\{X_-=0\}\\
		&= \left(\auf 0, \sigma\zu \setminus \cup_{n\in\bbn} \auf(\sigma^n_1)_{\{X_{\sigma^n_1-}=0\}}\zu\cup
		\zu \sigma_1^n,\sigma_2^n\zu\right)\cap\{X_-=0\}\in\mathcal{P},
	\end{align*}
	where we use that the infimum in (\ref{25.1.2020.1}) must be attained if $X_{\sigma-}=0$. Thus, $\sigma_{\{X_{\sigma-}=0\}}$ is a predictable stopping time. 
	Finally, we have that $\PM(\{\omega\in\Omega : \sigma(\omega)<\infty,\ (\omega,\sigma(\omega))\not\in M_1\cup M_2)=0$. By the maximality of the supremum, one has $$\mu(\auf\sigma_{\{X_{\sigma-}=0\}}\zu\cup\zu \sigma,\Gamma_1(\sigma)\zu)=0.$$ Since the intervals overlap $T$ or some $\tau^n_1(\omega)$ if they are nonempty, we arrive at $\PM(\sigma<\infty)=0$, and thus (\ref{25.11.2019.3}) holds.
\end{proof}

\begin{note}\label{3.2.2020.1}
	For any $\vp\in\bP$ and any $\sigma$-finite measure $\mu$ on $\mathcal{P}$  with $\mu^S\ll \mu$, there exists a uniformly bounded sequence of simple strategies~$(\vp^n)_{n\in\bbn}$ with $\vp^n\to\vp$,\ $\mu$-a.e., and for any such sequence $(\vp^n)_{n\in\bbn}$ one has $\vp^n\mal S\to\vp\mal S$ uniformly in probability.
\end{note}
\begin{proof}
	The existence of such a sequence with $\vp^n\to\vp$,\ $\mu$-a.e. follows from the approximation theorem for measures (see, e.g., Theorem~1.65(ii) in \cite{klenke}). 
	Then, the convergence of the integrals follows for the martingale parts by (3) on page~49 of \cite{jacod.shiryaev} and  for the finite variation parts by dominated convergence.
\end{proof}

\begin{proof}[Proof of Theorem~\ref{8.12.2019.1}]
	Obviously, it is sufficient to show the theorem under an equivalent measure $Q\sim \mathbb{P}$. Hence, we assume w.l.o.g. that $\mathbb{P}=Q$, where $Q$ is the measure introduced above \eqref{eq:Q.muS}.
	
	Ad (i): Let $(\vp^n)_{n\in\bbn}\subset \bP$ satisfy $\vp^n\to \vp$ pointwise on $\{\ov{S}_->\un{S}_-,\ A=1\}$.
	For any $J\in \mathcal{I}$ from (\ref{21.2.2020.1}), Proposition~\ref{prop:liminf} yields that
	$\liminf_{n\to\infty} C(\vp^n,J\cap[0,t])(\omega)\ge  C(\vp,J\cap[0,t])(\omega)$ for all $(\omega,t)\in \{A=1\}$. 
	It follows that $\liminf_{n\to\infty}C_t(\vp^n)(\omega) 
	\ge \sup_{J\in\mathcal{I}} C(\vp,J\cap[0,t])(\omega)= C_t(\vp)(\omega)$ for all $(\omega,t)\in \{A=1\}$.
	If in addition $(\vp^n)_{n\in\bbn}$ is uniformly bounded and $\vp^n\to \vp\ \mu^S\mbox{-a.e. on\ }\{\ov{S}_-=\un{S}_-,\ A=1\}$,  we have that
	\beam\label{25.8.2020.1}				
	(\vp^n\mathbbm{1}_{\{A=1\}})\mal S\to (\vp\mathbbm{1}_{\{A=1\}})\mal  S\quad\mbox{uniformly in probability}
	\eeam
	(see Note~\ref{3.2.2020.1}). Since $\{A=1\}$ is a predictable set of interval type, there is an increasing sequence of stopping times $(T^m)_{m\in\bbn}$ s.t. $\{A=1\}\cup(\Omega\times\{0\})=\cup_{m\in\bbn}\auf0,T^m\zu$ (see, e.g., \cite[Theorem 8.18]{he.wang.yan.1992}). For each $m\in\bbn$, we obviously have \begin{align*}
		\left(	\left(\mathbbm{1}_{\auf 0, T^m\zu}\vp\right)\bigcdot S\right)\mathbbm{1}_{\auf 0, T^m\zu}=\left(\vp\bigcdot S\right)^{T^m}\mathbbm{1}_{\auf 0, T^m\zu}=\left(\vp\bigcdot S\right)\mathbbm{1}_{\auf 0, T^m\zu}.
	\end{align*}
	Letting $m\to\infty$ this yields
	\begin{align}\label{eq:A=1}
		(\vp\mathbbm{1}_{\{A=1\}}\bigcdot S)\mathbbm{1}_{\{A=1\}}=(\vp\bigcdot S)\mathbbm{1}_{\{A=1\}} 
	\end{align} up to evanescence by Note~\ref{3.2.2020.1} and, analogously, $(\vp^n\mathbbm{1}_{\{A=1\}}\bigcdot S)\mathbbm{1}_{\{A=1\}}=(\vp^n\bigcdot S)\mathbbm{1}_{\{A=1\}}$ up to evanescence for $n\in\bbn$. Thus, 
	together with (\ref{25.8.2020.1}), we have 
	\begin{align*}
		\liminf_{n\to\infty}\left(\vp^n\mal S
		-\vp\mal S\right)^+\mathbbm{1}_{\{A=1\}} = 0\quad\mbox{up to evanescence.}
	\end{align*} 
	Putting the cost terms and the trading gains w.r.t. $S$ together, we arrive at (i).\\
	
	Ad (ii): The following analysis is based on the stopping times $(\tau^n_1)_{n\in\bbn}$ and $(\sigma^n_1)_{n\in\bbn}$ from Lemma~\ref{25.11.2019.1} and 
	Lemma~\ref{25.11.2019.2}, respectively. We can and do choose $(\sigma^n_1)_{n\in\bbn}$ s.t.
	\beam\label{19.2.2020.2}
	\PM(\sigma^n_1=\tau^m_1<\infty,\ X_{\sigma^1_n-}>0)=0,\quad\forall n,m\in\bbn.
	\eeam
	This means that if the spread~$X$ only touches zero at a single point and its left limit is non-zero, there directly starts the next excursion without a one point frictionless regime in between.
	
	For the rest of the proof, we write $\{X_{\tau-}\in B\}$ for the set 
	$\{\omega\in\Omega : \exists t\in[0,T]\ \tau(\omega)=t,\ X_{t-}(\omega)\in B\}$, where $\tau$ is a $[0,T]\cup\{\infty\}$-valued stopping time and $B\subset\bbr$.
	Let \begin{align}\nonumber
		A^n&:= \zu (\tau^n_1)_{\{X_{\tau^n_1-}>0\}},\Gamma_2(\tau_1^n)\auf\cup\auf(\Gamma_2(\tau_1^n))_{ 
			\{X_{\tau^n_1-}>0\}\cap\{X_{\Gamma_2(\tau_1^n)-}>0\}}
		\zu\in\mathcal{P},\ n\in\bbn, \\\nonumber
		B^n&:= \auf (\sigma^n_1)_{\{X_{\sigma^n_1-}=0\}}\zu \cup \zu \sigma_1^n,\Gamma_1(\sigma_1^n)\zu\in\mathcal{P},\ n\in\bbn, \\ \nonumber
		\wt{B}^n&:= \zu\Gamma_1(\sigma_1^n),\Gamma_2(\Gamma_1(\sigma_1^n))\auf\cup\auf (\Gamma_2(\Gamma_1(\sigma_1^n)))_{X_{\{\Gamma_2(\Gamma_1(\sigma_1^n)))-}>0\}}\zu\in\mathcal{P},\ n\in\bbn,
	\end{align}
	and 
	\begin{align}
		\label{3.12.2019.3}
		\vp^{N}:=\vp \mathbbm{1}_{\cup_{n=1,\ldots,N}(A^n\cup B^n\cup \wt{B}^n)},\quad N\in\bbn.
	\end{align}
	
	Excursions away from zero are either included by $A^n$ or by
	$\wt{B}^n$ with the frictionless forerunner~$B^n$. 
	In the first case, the spread cannot jump away from zero since $X_{\tau^n_1}=0$ on $\{X_{\tau^n_1-}>0\}$. 
	In the latter case, the frictionless forerunner avoids that $\vp^N$ produces costs when the spread jumps away from zero, which do not occur with the strategy~$\vp$. Namely, at a time the spread jumps away from zero, $\vp^N$ either remains zero or it already coincides with $\vp$. Note that the frictionless forerunner may consist of a single point only. For example, this is the case if the jump time is an accumulation point of starting/ending points of excursions shortly before.   
	
	First, we approximate $\vp$ by the strategies~$\vp^N$.
	
	{\em Step 1:} Let $E\in\mathcal{F}_T$ be a set with $\PM(E)=1$ s.t. the properties from Lemma~\ref{25.11.2019.1} and Lemma~\ref{25.11.2019.2} hold for all $\omega\in E$. Let us show that $\vp^{N}_t(\omega)\to \vp_t(\omega)$ for all $t\in[0,T]$ and $\omega\in E$. By construction of $\vp^N$, we only have to show that for each $n\in\bbn$, the excursion starting in $\tau^n_1(\omega)$ is overlapped by $A^n_\omega:=\{t\in[0,T] : (\omega,t)\in A^n\}$, the $\omega$-intersection of $A^n$, or by some $\wt{B}^m_\omega$, $m\in\bbn$. 
	In the case that $X_{\tau^n_1(\omega)-}(\omega)>0$, the excursion is overlapped by $A^n_\omega$. In the case that $X_{\tau^n_1(\omega)-}(\omega)=0$, we have by Lemma~\ref{25.11.2019.2} that
	$\tau^n_1(\omega)\in[\sigma^m_1(\omega),\Gamma_1(\sigma^m_1(\omega))]$ for some $m\in\bbn$ and thus the excursion starting in $\tau^1_n(\omega)$ is overlapped by $\wt{B}^m_\omega$. By Note~\ref{3.2.2020.1}, it follows that $\vp^{N}\mal S$ to $\vp\mal S$ uniformly in probability for $N\to\infty$.
	
	{\em Step 2:} W.l.o.g we assume that the bounded process~$\vp$ takes values in $[-1/2,1/2]$ to get rid of a further constant. Let us show that
	\beam\label{3.12.2019.1}
	\sup_{t\in[0,T]}|C_t(\vp^{N}) - C_t(\vp)|\mathbbm{1}_{\{C_t(\vp)\leq K\}}\to 0,\ N\to\infty,\ \mbox{pointwise on\ } E\ \forall K\in\bbn.
	\eeam 
	From $X_{\tau^n_1}=0$ on $\{X_{\tau^n_1-}>0\}$ and $X_{\sigma^n_1}=0$ on $\{X_{\sigma^n_1-}>0\}$, we conclude: for fixed $\omega\in E$ and $0\le a\le b\le T$ with $\inf_{u\in[a,b)}X_u(\omega)>0$, we either have that $\vp^N_u(\omega)=\vp_u(\omega)$ for all $u\in[a,b]$ or  $\vp^N_u(\omega)=0$ for all $u\in[a,b]$.  
	By the definition of the cost term in (\ref{31.12.2019.1}), this yields
	$C(\vp^{N},I\cap[0,t])\le C(\vp,I\cap[0,t])$ for all $I\in\mathcal{I}$, $(\omega,t)\in E\times[0,T]$ and thus $C_t(\vp^{N})\le C_t(\vp)$ for all $(\omega,t)\in E\times[0,T]$. We define \begin{align}\label{eq:DefTheta}
		\theta^m:=\inf\{t\geq 0\ :\ C_t(\vp)>m\}\wedge T\quad \text{for}\ m\in\bbn.
	\end{align} By $\Delta^-C_{\theta^m}(\vp)\le \sup_{u\in[0,T]}X_u$, the paths of the stopped process~$C^{\theta^m}(\vp)$ are bounded.
	Fix $\omega\in E$ and $\eps>0$. For $K\in\bbn$ we set $u:=\theta^K$. Proposition~\ref{prop:SubIntervallProperty} yields that $C(\vp,I\cap[0,u]) =  C(\vp,I\cap [0,t]) +  C(\vp,I\cap [t,u])$ for all $I\in\mathcal{I}$ and $t\leq u$.
	Therefore, together with Proposition~\ref{prop:PathwiseProperties}(i), there exists $I\in\mathcal{I}$ s.t. $$\sup_{t\in[0,T]}\left(C_t(\vp)- C(\vp,I\cap[0,t])\right)\mathbbm{1}_{\{C_t(\vp)\leq K\}}\le \eps.$$
	The set $I$ is  overlapped by finitely many $\omega$-intersections of $A^n$ and $B^n\cup\wt{B}^n$, i.e., for $N$ large enough, one has $I\subset 
	\cup_{n\le N}(A^n\cup B^n\cup \wt{B}^n)_\omega$, i.e., $C(\vp^N, I\cap [0,t])=C(\vp, I\cap [0,t])$ and, consequently, 
	$(C_t(\vp)-C_t(\vp^N))\mathbbm{1}_{\{C_t(\vp)\leq K\}}\leq (C(\vp, I\cap[0,t])-C(\vp^N, I\cap [0,t]))\mathbbm{1}_{\{C_t(\vp)\leq K\}}+\varepsilon= \varepsilon$ for all $t\in[0,T]$. This implies (\ref{3.12.2019.1}). Together with Step~1, we have that
	\begin{align}\label{2.2.2020.1}
		&\vp^N\to \vp\ \mbox{pointwise up to evanescence}\\ \label{2.2.2020.1.2}
		&\text{and}\ \sup_{t\in[0,T]}\vert V_t(\vp^{N})- V_t(\vp)\vert\mathbbm{1}_{\{C_t(\vp)\leq K\}}\to 0\ \mbox{in probability}
	\end{align}
	for $N\to\infty$ and each $K\in\bbn$.

	{\em Step 3:} It remains to approximate the strategies~$\vp^{N}$, $N\in\bbn$, by almost simple strategies. Since the pointwise convergence that we need on $\{X_->0\}\cap \{C(\vp)<\infty\}$ is not metrizable, it is not sufficient to approximate each $\vp^N$ separately by 
	a sequence of almost simple strategies. Let $\mu$ be a $\sigma$-finite measure on $\mathcal{P}$ with $\mu^S\ll \mu$. We fix some $N\in\bbn$ and let $\eps:=2^{-N}$. 
	In the following, we construct an almost simple strategy step by step on disjoint stochastic intervals. The main idea is to approximate the cost term on subintervals of excursions where the spread is bounded away from zero while controlling the error at the beginning and the end of the excursions. 
	We start with the construction of an almost simple strategy on $A^n$ with $n\le N$. We recall that $\tau^n_2:=\Gamma_2(\tau^n_1)$. There exists a stopping time~$\tau^{n,N}_1$ with $\theta^N\wedge\tau^n_2\geq\tau^{n,N}_1> \tau^n_1$ on 
	$\{\tau^n_1<\theta^{N}\}\cap\{X_{\tau^n_1-}>0\}$, $\tau^{n,N}_1=\theta^N$ on $\{\theta^N\le\tau^n_1\}$ and, for notational convenience, $\tau^{n,N}_1 = \tau^n_1$ elsewhere s.t. 
	\beam\label{5.2.2020.01}
	\PM(\tau^n_1\wedge \theta^N \le \tau^{n,N}_1\le  \tau^n_1+\eps)= 1,
	\eeam
	$\PM((\vp^N \mathbbm{1}_{\zu\tau^n_1,\tau^{n,N}_1\zu}\mal S)^\star>\eps)\le \eps$, 
	$\PM(\tau^n_1<\infty,\ |X_{\tau^{n,N}_1} -X_{\tau^n_1\wedge \theta^N}| >\eps)\le \eps$,
	and $\PM(\tau^n_1<\infty,\ C_{\tau^{n,N}_1}(\vp^N)- C_{\tau^n_1\wedge \theta^N}(\vp^N)> \eps)\le \eps$, where we
	use the notation $Y^\star:=\sup_{t\in[0,T]}|Y_t|$ and $\theta^N$ was defined in \eqref{eq:DefTheta}.
	This follows from the right-continuity of the processes $\vp^N \mathbbm{1}_{\zu\tau^n_1,T\zu}\mal S$, $X$ and from the definition of the cost process together with $X_{\tau^n_1}=0$ on $\{X_{\tau^n_1-}>0\}$.
	In addition, since $\auf (\tau^n_2)_{\{X_{\tau^n_2-=0}\}}\zu = \zu \tau^n_1,\tau^n_2\zu\cap\{X_-=0\}\in\mathcal{P}$, the stopping time $(\tau^n_2)_{\{X_{\tau^n_2-=0}\}}$ is predictable.
	Thus, by the existence of an announcing sequence (see, e.g., \cite[Theorem~4.34]{he.wang.yan.1992}),				
	there is a stopping time~$\tau^{n,N}_2$ with $\tau^{n,N}_1\leq\tau^{n,N}_2\le \tau^n_2\wedge \theta^N$ and $\tau^{n,N}_2<\tau^n_2$ on $\{X_{\tau^n_2-}=0, \tau^{n,N}_1<\tau^n_2\}$ s.t.  
	\begin{align}\label{20.02.2020.01}
		&\PM(\tau^{n,N}_2 < \tau^n_2\wedge \theta^N -\eps)\le\eps, \qquad \PM(X_{\tau^n_2-}>0,\ \tau^{n,N}_2 < \tau^n_2\wedge\theta^N)\le\eps,\\
		\nonumber &\PM((\vp^N \mathbbm{1}_{\zu\tau^{n,N}_2,\tau^n_2\wedge\theta^N\auf \cup \auf(\tau^n_2)_{\{X_{\tau^n_2-}>0,\ \tau^{n,N}_2<\tau^n_2\wedge\theta^N\}}\zu}\mal S)^\star>\eps)\le \eps,\\
		\nonumber &\PM(X_{\tau^{n,N}_2}> \eps, \tau^{n,N}_2<\tau^n_2\wedge \theta^N)\le \eps,
		\ \text{and} \ \PM(\tau^n_2<\infty,\ C_{\tau^n_2 \wedge\theta^{N}}(\vp^N)-C_{\tau^{n,N}_2}(\vp^N) > \eps)\le \eps.
	\end{align}
	By Proposition~\ref{1.1.2020.2} applied to the stopping times $\tau^{n,N}_1\leq \tau^{n,N}_2$, there exists an almost simple strategy~$\wt{\psi}^N$ with $\wt{\psi}^N_{\tau^{n,N}_1}=\vp^N_{\tau^{n,N}_1}$,  
	\beam\label{5.2.2020.02}
	\sup_{t\in[\tau^{n,N}_1,\tau^{n,N}_2]}|\wt{\psi}^N_t-\vp^N_t|\le \eps,
	\eeam
	\beao
	\PM\left(\sup_{t\in[\tau^{n,N}_1,\tau^{n,N}_2]}|C_t(\wt{\psi}^N)-C_{\tau^{n,N}_1}(\wt{\psi}^N)- (C_t(\vp^N)-C_{\tau^{n,N}_1}(\vp^N))| >\eps\right)\le \eps,
	\eeao 
	and
	$\PM(((\wt{\psi}^N-\vp^N) \mathbbm{1}_{\zu \tau^{n,N}_1,\tau^{n,N}_2\zu}\mal S)^\star>\eps)\le \eps$ (the later also uses Note~\ref{3.2.2020.1}).   
	We define the almost simple strategy by 
	\beam\label{4.12.2019.1}
	\psi^N_t := \wt{\psi}^N_t \mathbbm{1}_{(\tau^{n,N}_1< t\le \tau^{n,N}_2)}\quad\mbox{on}\quad A^n.
	\eeam
	Since $\psi^N$ can be updated for free at the left endpoint of $A^n$, for the increments of the process $V(\psi^N)-V(\vp^N) = (\psi^N-\vp^N)\mal S - (C(\psi^N) - C(\vp^N))$ we get the estimate
	\beam\label{19.2.2020.1}
	& & \PM\left( \sup_{t\in(\tau^n_1,\tau^n_2)\cup[(\tau^n_2)_{\{X_{\tau^n_2-}>0\}}]}
	|V_t(\psi^N)-V_{\tau^n_1}(\psi^N) - (  V_t(\vp^N)-V_{\tau^n_1}(\ph^N))|\mathbbm{1}_{\{C_t(\vp)\leq K\}}>8\eps
	\right.\nonumber\\
	& & \left. \qquad \tau^n_1<\infty,\ X_{\tau^n_1-}>0, 
	\right)\le 8\eps\qquad\ \mbox{for all}\ n=1,\ldots,N,\ K\le N,
	\eeam
	regardless of how $\psi^N$ is defined outside $A^n$, especially at time $\tau^n_1$. Indeed, in the worst case, there are $2$ error terms on $(\tau^n_1,\tau^{n,N}_1]$, $3$ error terms on 
	$(\tau^{n,N}_1,\tau^{n,N}_2]$, and $3$ error terms between $(\tau^{n,N}_2,\tau^n_2)\cup[(\tau^n_2)_{\{X_{\tau^n_2-}>0\}}]$. 
	
	We proceed with the construction of the almost simple strategy on $B^n\cup \wt{B}^n$ with $n\leq N$. A strategy with support~$B^n$ has zero costs, and by Note~\ref{3.2.2020.1}, we find an (almost) simple strategy~$\wh{\psi}^N$ with 
	\beam\label{3.2.2020.3}
	\mu(|\wh{\psi}^N-\vp^N|\mathbbm{1}_{B^n}>\eps)\le \eps,
	\eeam
	\beam\label{18.2.2020.01}
	\PM(\Gamma_1(\sigma_1^n)<\infty,\ |\wh{\psi}^N_{\Gamma_1(\sigma_1^n)} - \vp^N_{\Gamma_1(\sigma_1^n)}|X_{\Gamma_1(\sigma_1^n)}>\eps)\le \eps,
	\eeam
	and $\PM(((\wh{\psi}^N-\vp^N) \mathbbm{1}_{B^n}\mal S)^\star>\eps)\le \eps$. 
	After $\Gamma_1(\sigma_1^n)$, we proceed similar to (\ref{4.12.2019.1}). Setting $\wt{\tau}^n_2:=\Gamma_2(\Gamma_1(\sigma_1^n))$, there exists a stopping time $\wt{\tau}^{n,N}_1$ with $\wt{\tau}^{n,N}_1=\theta^N$ on $\{\theta^N\le\Gamma_1(\sigma^n_1)\}$, $\wt{\tau}^{n,N}_1 = \Gamma_1(\sigma_1^n)$ on $\{\Gamma_1(\sigma_1^n)<\theta^N,\ X_{\Gamma_1(\sigma_1^n)}>0\}$ and
	$\theta^N\wedge \wt{\tau}^n_2\geq \wt{\tau}^{n,N}_1> \Gamma_1(\sigma_1^n)$ on 
	$\{\Gamma_1(\sigma_1^n)<\theta^{N},\ X_{\Gamma_1(\sigma_1^n)}=0\}$ s.t. $\PM(\Gamma_1(\sigma_1^n)\wedge \theta^N\le \wt{\tau}^{n,N}_1\le \Gamma_1(\sigma_1^n)+\eps)= 1$,
	$\PM(((\vp^N-\vp^N_{\Gamma_1(\sigma_1^n)}) \mathbbm{1}_{\zu\Gamma_1(\sigma_1^n),\wt{\tau}^{n,N}_1\zu}\mal S)^\star>\eps)\le \eps$, 
	$\PM(\Gamma_1(\sigma_1^n)<\infty,\ |X_{\wt{\tau}^{n,N}_1} -X_{\Gamma_1(\sigma_1^n)\wedge\theta^N}|>\eps)\le \eps$,
	and $\PM(\Gamma_1(\sigma_1^n)<\infty,\ C_{\wt{\tau}^{n,N}_1}(\vp^N)-C_{\Gamma_1(\sigma_1^n)\wedge\theta^N}(\vp^N)> \eps)\le \eps$. $\wt{\tau}^{n,N}_2$ is defined completely analogous to $\tau^{n,N}_2$ from above. We set
	\beam\label{4.12.2019.2}
	\psi^N_t := \wh{\psi}^N_t \mathbbm{1}_{(t\le \Gamma_1(\sigma_1^n)\wedge\theta^N)}
	+ \ov{\psi}^N_t \mathbbm{1}_{(\wt{\tau}^{n,N}_1< t\le  
		\wt{\tau}^{n,N}_2)}\ \mbox{on}\ B^n\cup \wt{B}^n
	\eeam
	for some almost simple strategy~$\ov{\psi}^N$ with $\ov{\psi}^N_{\wt{\tau}^{n,N}_1}=\vp^N_{\wt{\tau}^{n,N}_1}$ and $\sup_{t\in[\wt{\tau}^{n,N}_1,\wt{\tau}^{n,N}_2]}|\ov{\psi}^N_t-\vp^N_t|\le \eps$. 
	As in (\ref{19.2.2020.1}), but with the additional error terms on $B^n$ and (\ref{18.2.2020.01}) for the case that the spread jumps away from zero, 
	we get that
	\begin{align}
		& \nonumber \PM\left(\sup_{\substack{t\in[(\sigma^n_1)_{\{X_{\sigma^n_1-}=0\}}\cup(\sigma^n_1,\Gamma_2(\Gamma_1(\sigma^n_1)))\\ \cup[(\Gamma_2(\Gamma_1(\sigma^n_1)))_{\{X_{\Gamma_2(\Gamma_1(\sigma^n_1))-}>0\}}]}}
		|V_t(\psi^N)-V^1 - (  V_t(\vp^N)-V^2)|\mathbbm{1}_{\{C_t(\vp)\leq K\}}\ >10\eps 
		\right)\\ \label{19.2.2020.3} & \qquad\le 10\eps\qquad \mbox{for all}\ n=1,\dots,N,\ K\le N,
	\end{align}
	where $V^1:=V_{\sigma^n_1-}(\psi^N)$, $V^2:=V_{\sigma^n_1-}(\vp^N)$
	on $\{X_{\sigma^n_1-}=0\}$ and $V^1:=V_{\sigma^n_1}(\psi^N)$, $V^2:=V_{\sigma^n_1}(\vp^N)$ on $\{X_{\sigma^n_1-}>0\}$.
	By (\ref{19.2.2020.2}), $A^n$ and $B^m\cup \wt{B}^m$ are disjoint. Thus, (\ref{4.12.2019.1})
	and (\ref{4.12.2019.2}) can be used to define an almost simple
	strategy on $\Omega\times[0,T]$: for $n\le N$, define $\psi^N$ on $\cup_{n\le N}(A^n\cup B^n\cup \wt{B}^n)$ as above and set $\psi^N:=0$ on $(\Omega\times[0,T])\setminus\cup_{n\le N}(A^n\cup B^n\cup \wt{B}^n)$.
	By $V_0(\psi^N)=V_0(\vp^N)=0$ and the construction of $A^n$ and $B^n\cup\wt{B}^n$,  for each $(\omega,t)$, $(V_t(\psi^N_t)(\omega)-V_t(\vp^N)(\omega))\mathbbm{1}_{\{C_t(\vp)\leq K\}}(\omega)$ can be written as a finite sum of increments from (\ref{19.2.2020.1}) and (\ref{19.2.2020.3}).
	For this, we again use that at the right endpoint of $A^n$ and $\wt{B}^n$, the position can be liquidated without any costs.
	Summing up the error terms and recalling that $\varepsilon=2^{-N}$, this yields $\PM(\sup_{t\in[0,T]}|V_t(\psi^N)-V_t(\vp^N)|\mathbbm{1}_{\{C_t(\vp)\leq K\}} > 18 N 2^{-N})\le 18 N 2^{-N}$ for all $N\geq K$. Together with \eqref{2.2.2020.1.2}, we obtain $\sup_{t\in[0,T]}\vert V_t(\psi^N)-V_t(\vp)\vert\mathbbm{1}_{\{C_t(\vp)\leq K\}}\to 0$ in probability for $N\to\infty$ and all $K\in\bbn$.\\

	By \eqref{3.2.2020.3} and \eqref{4.12.2019.2}, we have that $(\psi^N)_{N\in\bbn}$ converges to $\vp$ $\mu$-a.e. on $\{X_-=0\}{\cap\{C(\vp)<\infty\}}$. It remains to show that $(\psi^N)_{N\in\bbn}$ converges pointwise up to evanescence to $\vp$ on the set~$\{X_->0\}\cap\{C(\vp)<\infty\}$.
	Let $(\omega,t)\in \Omega\times[0,T]$ with $X_{t-}(\omega)>0$ and $C_t(\vp)(\omega)<\infty$. By the arguments in Step~1, there exists an $n\in\bbn$ with $(\omega,t)\in A^n\cup \wt{B}^n$. W.l.o.g. $(\omega,t)\in A^n$.
	By (\ref{5.2.2020.01}), one has $\tau^{n,N}_1(\omega)\le \tau^n_1(\omega)+2^{-N}<t$  and, as the costs at $t$ are finite, $\theta^{N}(\omega)\geq t$ for $N$ large enough. 
	
	{\em Case 1:}   $t<\tau^n_2(\omega)$. By (\ref{20.02.2020.01}) and the lemma of Borel-Cantelli, we have that 
	$\PM(E^n)=0$, where $E^n:=\cap_{\wt{N}\in\bbn}\cup_{N\ge \wt{N}}
	\{\tau^{n,N}_2 < \tau^n_2 - 2^{-N}\}$. If $\omega\not\in E^n$, this implies that $t<   \tau^n_2(\omega) - 2^{-N}\le \tau^{n,N}_2(\omega)$ for $N$ large enough and thus by (\ref{5.2.2020.02}), $|\psi^N_t(\omega)-\vp^N_t(\omega)|\le 2^{-N}$ for $N$ large enough.

	{\em Case 2:} $t=\tau^n_2(\omega)$ and thus $X_{\tau^n_2(\omega)-}>0$. By (\ref{20.02.2020.01}) and the lemma of Borel-Cantelli, we have that 
	$\PM(\wt{E}^n)=0$, where $\wt{E}^n:=\cap_{\wt{N}\in\bbn}\cup_{N\ge \wt{N}}
	\{X_{\tau^n_2-}>0,\ \tau^{n,N}_2 < \tau^n_2\}$. If $\omega\not\in \wt{E}^n$, this implies that $t=\tau^{n,N}_2(\omega)$ for $N$ large enough and thus by (\ref{5.2.2020.02}), $|\psi^N_t(\omega)-\vp^N_t(\omega)|\le 2^{-N}$ for $N$ large enough.

	Since $\vp^N_t(\omega)=\vp_t(\omega)$ for all $N\ge n$, we conclude that 
	the sequence~$(\psi^N)_{N\in\bbn}$ converges pointwise up to evanescence to $\vp$ on the set~$\{X_->0\}\cap\{C(\vp)<\infty\}$.
\end{proof}
\appendix\normalsize
\section{Technical results: Construction of the cost term} \label{app:TechnicalCosts}

\begin{proof}[Proof of Proposition~\ref{prop:ExistenceCostTerm} and Proposition~\ref{prop:approximationcost}] As the two propositions are interrelated, we give their proofs together. 
	Recall that the arguments below are path-by-path, i.e., $\omega\in\Omega$ is fixed.
	
	\emph{Step 1:} We begin by establishing the uniqueness of the cost term. Therefore, assume that there are exist $C_1,C_2\in [0,\infty]$ satisfying the condition in Definition~\ref{def:CostTermIntervall}.  
	This means that for each $i\in\{1,2\}$, $\varepsilon>0$, we find  a partition $P^i_\varepsilon$ of $I=[a,b]$ s.t. for every refinement $P$ of $P^i_\varepsilon$ and every modified intermediate subdivision $\lambda$ of $P$, we have $d(C_i,R(\varphi,P,\lambda))<\varepsilon$,
	where $d(x,y):=\vert \arctan(x)-\arctan(y)\vert$ with
	$\arctan(\infty):=\pi/2$, which defines a metric on $[0,\infty]$.	
	But, letting $\lambda$ denote an arbitrary modified intermediate subdivision of $P^1_\eps\cup P^2_\eps$, this means
	\begin{align*}
		d(C_1,C_2)\leq d(C_1,R(\varphi,P^1_\eps\cup P^2_\eps,\lambda))+d(C_2,R(\varphi,P^1_\eps\cup P^2_\eps,\lambda))<2\varepsilon,
	\end{align*}
	which means $C_1=C_2$ as the above holds for all $\varepsilon>0$.
	
	\emph{Step 2:} We now turn towards existence. Let $(\delta_n)_{n\in\bbn},(\eta_n)_{n\in\bbn}\subseteq(0,\infty) $ be sequences with $\delta_n\downarrow 0$ and $\eta_n\downarrow0$. It follows from a minor adjustment of \cite[Lemma 2.1]{mikosch2000stochastic} that for each $n\in\bbn$
	there is a partition $P_n=\{t^n_0,\dots,t^n_{k_n}\}$ of $I$ s.t. \begin{align}\label{eq:ExistenceCostTerm1}
		\begin{aligned}
			\mathrm{osc}(\overline{S}-S, [t^n_{i-1},t^n_i))<\delta_n\quad \text{and}\quad \mathrm{osc}(S-\underline{S}, [t^n_{i-1},t^n_i))<\delta_n\
	\end{aligned}\end{align}
	for $i=1,\dots,k_n$. By the definition of the oscillation of a function, \eqref{eq:ExistenceCostTerm1} also holds for every refinement of $P_n$. Hence, $P_n$ can be chosen s.t. we also have  
	\begin{align}\label{eq:ExistenceCostTerm2New}
		\begin{cases}
			\sum_{i=1}^{k_n}\vert \varphi_{t^n_i}-\varphi_{t^n_{i-1}}\vert +\eta_n \geq \mathrm{Var}_{a}^b(\varphi),& \text{if}\ \mathrm{Var}_a^b(\varphi)<\infty\\
			\sum_{i=1}^{k_n}\vert \varphi_{t^n_i}-\varphi_{t^n_{i-1}}\vert>1/\eta_n, & \text{if}\ \mathrm{Var}_a^b(\varphi)=\infty
		\end{cases} \quad\text{for all}\ n\in\bbn.
	\end{align}
	In addition, we can obviously choose the sequence $(P_n)_{n\in\bbn}$ s.t. it is refining. This shows that there exists a refining sequence of partitions satisfying assertions~(i) and (ii) of Proposition~\ref{prop:approximationcost}.

	\emph{Step 3:} Next, let $(P_n)_{n\in\bbn}$ be a refining sequence of partitions from step 2, i.e., $P_n=\{t^n_0,\dots,t^n_{k_n}\}$ satisfies \eqref{eq:ExistenceCostTerm1} and \eqref{eq:ExistenceCostTerm2New}.
	
	\textit{Case 1:} Let us first assume $\mathrm{Var}_a^b(\varphi)<\infty$. Let $M:=\sup_{t\in I}(\overline{S}_t-\underline{S}_t)$. We claim that for all subdivisions 
	$\lambda=\{s_1,\dots,s_{k_n}\}$ of $P_n$, all refinements 
	$P'=\{t'_0,\dots,t'_m\}$ of $P_n$, and all subdivisions 
	$\lambda'=\{s'_1,\dots, s'_m\}$ of $P'$, we have
	\begin{align}\label{eq:ExistenceCostTerm2new}
		\vert R(\varphi,P_n,\lambda)-R(\varphi,P',\lambda')\vert \leq \eta_n M+ \delta_n \mathrm{Var}_a^b(\varphi).
	\end{align} 
	The key estimate to derive (\ref{eq:ExistenceCostTerm2new}) is
	\begin{align*}
		&\left\vert (\overline{S}_{s_i}-S_{s_i})\left(\varphi_{t^n_i}-\varphi_{t^n_{i-1}}\right)^+-\sum_{k=1}^{n_i} \left(\overline{S}_{s'_{i_k}}-S_{s'_{i_k}}\right)\left(\varphi_{t'_{i_k}}-\varphi_{t'_{i_{k-1}}}\right)^+\right\vert\\
		&\leq \left\vert (\overline{S}_{s_i}-S_{s_i})\left(\left(\varphi_{t^n_i}-\varphi_{t^n_{i-1}}\right)^+-\sum_{k=1}^{n_i}\left(\varphi_{t'_{i_k}}-\varphi_{t'_{i_{k-1}}}\right)^+\right)\right\vert\\ &\ +\left\vert\sum_{k=1}^{n_i} \left(\left(\overline{S}_{s'_{i_k}}-S_{s'_{i_k}}\right)-\left(\overline{S}_{s_i}-S_{s_i}\right)\right)\left(\varphi_{t'_{i_k}}-\varphi_{t'_{i_{k-1}}}\right)^+\right\vert\\
		&\leq M  \left(\sum_{k=1}^{n_i} \left(\varphi_{t'_{i_k}}-\varphi_{t'_{i_{k-1}}}\right)^+
		-\left(\varphi_{t^n_i}-\varphi_{t^n_{i-1}}\right)^+\right)+\delta_n\sum_{k=1}^{n_i} \left(\varphi_{t'_{i_k}}-\varphi_{t'_{i_{k-1}}}\right)^+,
	\end{align*}
	where $i\in\{1,\dots, k_n\}$ and $t'_{i_1},\dots, t'_{i_{n_i}}$ denote the elements of $P'$ with $t^n_{i-1}=t'_{i_1}<\dots< t'_{i_{n_i}}=t^n_i$.
	
	Now, let $(\lambda_n)_{n\in\bbn}$ be arbitrary modified intermediate subdivisions of $(P_n)_{n\in\bbn}$. Then, as the sequence $(P_n)_{n\in\bbn}$ is refining, \eqref{eq:ExistenceCostTerm2new} yields
	\begin{align*}
		\sup_{m\geq n}\vert R(\varphi, P_m,\lambda_m)-R(\varphi,P_n,\lambda_n)\vert\leq  \eta_n M+\delta_n \mathrm{Var}_a^b(\varphi).
	\end{align*} 
	Thus, the sequence $(R(\varphi, P_n,\lambda_n))_{n\in\bbn}$ is Cauchy in $\mathbb{R}_+$ and $C:=\lim\limits_{n\to\infty} R(\varphi, P_n,\lambda_n)\in\mathbb{R}_+$ exists.  It remains to show that $C$ satisfies Definition~\ref{def:CostTermIntervall}(i). Therefore, let $\varepsilon>0$ and choose $n\in\mathbb{N}$ s.t. $\eta_n M+\delta_n \mathrm{Var}_a^b(\varphi)<\varepsilon/2$ and  $\vert C-R(\varphi, P_n,\lambda_n)\vert <\varepsilon/2$. Together with \eqref{eq:ExistenceCostTerm2new}, this implies that for all refinements $P'$ of $P_n$ and subdivisions~$\lambda'$ of $P'$, we have 
	\begin{align*}
		\vert C-R(\varphi, P',\lambda')\vert &\leq \vert C-R(\varphi, P_n,\lambda_n)\vert+\vert R(\varphi, P_n,\lambda_n)- R(\varphi, P',\lambda')\vert<\varepsilon.
	\end{align*} 
	Thus, $C$ satisfies Definition~\ref{def:CostTermIntervall}(i).

	\emph{Case 2:} We now treat the case $\mathrm{Var}_a^b(\varphi)=\infty$. In this case, we will show that the cost term exists and $C(\varphi,I)=\infty$. Recall that we assumed $\delta:=\inf_{t\in[a,b)}(\overline{S}_t-\underline{S}_t)>0$.  We define a sequence $(\sigma_k)_{k\geq 0}$ by $\sigma_0=a$ and 
	\begin{align*}
		\sigma_k:=
		\begin{cases}
			\inf\{t\geq \sigma_{k-1}: S_t\leq \underline{S}_t+\delta/3\}\wedge b, & k \ \text{odd}\\
			\inf\{t\geq \sigma_{k-1}: S_t\leq \overline{S}_t-\delta/3\}\wedge b, & k\ \text{even}.
		\end{cases}
	\end{align*}
	As $\underline{S},S$, and $\overline{S}$ are c\`adl\`ag, we have $\sigma_k=b$ for $k$ large enough. Hence, let $K\in\mathbb{N}$ denote the smallest number s.t. $\sigma_K=b$. In addition, note that we also have $\sigma_0\leq \sigma_1<\sigma_2<\dots<\sigma_K=b$ and, per construction,
	\begin{align}\label{eq:ExistenceCostTerm5}
		\inf_{t\in[\sigma_{2k},\sigma_{2k+1})} S_t-\underline{S}_t>\delta/3,\quad \text{and}\quad \inf_{t\in[\sigma_{2k+1},\sigma_{2(k+1)})}\overline{S}_t-S_t>\delta/3.
	\end{align}  
	Recall that $\mathrm{Var}_a^b(\varphi)=\infty$ implies that $\sum_{i=1}^{k_n}\vert \varphi_{t_i^n}-\varphi_{t_{i-1}^n}\vert\to\infty$ as $n\to\infty$ by \eqref{eq:ExistenceCostTerm2New}. 
	Since $K<\infty$ and $\vp$ is bounded, this implies that for at least one $k\in\{0,1,\ldots,K-1\}$, we have \[\sum_{\substack{t^n_i,t^n_{i-1}\in P_n\\ t_i^n,t^n_{i-1}\in[\sigma_k,\sigma_{k+1}]}}|\varphi_{t_i^n}-\varphi_{t^n_{i-1}}| \to\infty,\quad n\to\infty,\] which, again by the boundedness of $\varphi$, implies that 
	\begin{align}
		\label{1.1.2020.1}
		\begin{aligned}
			&\sum_{\substack{t^n_i,t^n_{i-1}\in P_n\\ t^n_i,t_{i-1}^n\in[\sigma_k,\sigma_{k+1}]}}(\varphi_{t_i^n}-\varphi_{t^n_{i-1}})^+ \to\infty,\ n\to\infty\\ \text{and}\quad & \sum_{\substack{t^n_i,t^n_{i-1}\in P_n\\ t^n_i,t_{i-1}^n\in[\sigma_k,\sigma_{k+1}]}}(\varphi_{t_i^n}-\varphi_{t^n_{i-1}})^- \to\infty,\ n\to\infty.	
		\end{aligned}
	\end{align}
	By \eqref{eq:ExistenceCostTerm5}, this implies that $R(\vp,P_n,\lambda_n)\to \infty$ as $n\to\infty$ for arbitrary subdivisions $\lambda_n$ of $P_n$. Since the sums in (\ref{1.1.2020.1}) get even bigger if $P_n$ are replaced by refining partitions~$G_n$,
	the cost term~$C(\varphi,I)$ exists and is $\infty$. 	
	
	This finishes the proof of Propositions~\ref{prop:ExistenceCostTerm} and \ref{prop:approximationcost}. Indeed, in step~2 above, we showed that there exists a sequence of partitions satisfying the assumptions (i) and (ii) of Proposition~\ref{prop:approximationcost}. Subsequently, in step~3 we showed that for every refining sequence of partitions with these properties the corresponding Riemann-Stieltjes sums converge and their limits satisfy Definition~\ref{def:CostTermIntervall}. Thus, by the uniqueness shown in step~1, their limits coincide and we are done. 
\end{proof} 
We now turn to the proof of Lemma~\ref{lemma:PredictableIntervall}. This will rely on the following concept and result of Doob~\cite{doob1975stochastic}. 
\begin{definition}
	Let $\varphi$ be a stochastic process. A sequence $(T_n)_{n\in\mathbb{N}}$ of predictable stopping times is called a \emph{predictable separability set for $\varphi$} if for each $\omega\in\Omega$ the set $\{T_n(\omega): n\in\mathbb{N}\}$ contains $0$ and is dense in $[0,T]$ and 
	\beam\label{3.1.2020.1}
	\{(t,\varphi_t(\omega)):t\in[0,T]\}=\overline{\{(T_n(\omega),\varphi_{T_n(\omega)}(\omega)): n\in\mathbb{N}\}},
	\eeam
	i.e., the graph of the sample function $t\mapsto\varphi_t(\omega)$ is the closure of the graph restricted to the set $\{T_n(\omega): n\in\mathbb{N}\}$. A stochastic process $\varphi$ having a predictable separability set is called \emph{predictably separable}.
\end{definition} 
\begin{theorem}[Doob \cite{doob1975stochastic}, Theorem 5.2]\label{theo:DoobSeparable}
	A predictable process coincides with some predictably separable predictable process up to evanescence. 
\end{theorem}
\begin{proof}[Proof of Lemma~\ref{lemma:PredictableIntervall}]
	By Theorem~\ref{theo:DoobSeparable}, we have to show that for a predictably separable predictable process~$\vp$, the process $C(\vp,[\sigma\wedge\cdot,\tau\wedge\cdot])$
	is predictable. 
	
	Let $\{T_n: n\in\mathbb{N}\}$ denote the predictable separability set for $\vp$. 
	By (\ref{3.1.2020.1}), we can find a sequence of finite sequences of 
	(not necessarily predictable) stopping times $\sigma=T^n_0\leq T^n_1\leq\dots \leq T^n_{m_n}=\tau$ s.t. 
	\begin{align*}
		\mathrm{Var}_{\sigma\wedge t}^{\tau\wedge t}(\varphi)=\lim\limits_{n\to\infty}\sum_{i=1}^{m_n}\vert \varphi_{T^n_i\wedge t}-\varphi_{T^n_{i-1}\wedge t}\vert,\quad\mbox{pointwise},\quad t\in[0,T].
	\end{align*}
	Next, we define for each $n\in\mathbb{N}$ and $i\in\{1,\dots, m_n\}$ a sequence $(V^{n,i}_l)_{l\in\bbn}$ of stopping times by $V_0^{n,i}=T_{i-1}^n$ and recursively
	\begin{align*}
		V_l^{n,i}:=\inf\{t> V_{l-1}^{n,i}: &\vert \overline{S}_t-S_t-(\overline{S}_{V^{n,i}_{l-1}}-S_{V^{n,i}_{l-1}})\vert >\frac{1}{2n}\\ & \text{or}\ \vert S_t-\underline{S}_t-(S_{V^{n,i}_{l-1}}-\underline{S}_{V^{n,i}_{l-1}})\vert >\frac{1}{2n}\}\wedge T^{n}_{i}.
	\end{align*}
	This leads to the sequence of random partitions $\ov{P}_n:=\bigcup_{k\le n}\bigcup_{i=1,\dots, m_k}\bigcup_{l\in\bbn_0} \{V^{i,k}_l\}$, $n\in\bbn$, which is for each $\omega$ refining. Note that for $\omega$ and $n$ fixed, $P_n$ is finite.
	Rearranging the resulting stopping times in increasing order yields 
	a refining sequence of increasing sequences of stopping times  $(\nu^n_k)_{k\in\mathbb{N}}$, $n\in\bbn$,
	s.t. 
	$\#\{k: \nu^n_k(\omega)<\infty\}<\infty$ for all $n\in\mathbb{N}$, $\mathrm{Var}_{\sigma\wedge t}^{\tau\wedge t}(\varphi)=\lim\limits_{n\to\infty}\sum_{k=0}^\infty\vert \varphi_{\nu^n_k\wedge t}-\varphi_{\nu^n_{k-1}\wedge t}\vert$ for all $t\in[0,T]$, and $\max(\mathrm{osc}(\overline{S}-S, [\nu^n_k,\nu^n_{k+1})),\mathrm{osc}(S-\underline{S}, [\nu^n_k,\nu^n_{k+1})))\leq 1/n$ for all $k\in\mathbb{N}_0$ and $n\in\mathbb{N}$.
	In particular, this means that for each $\omega\in\{\sigma<\tau\}$ and $t\in[0,T]$ the sequence of partitions $(P_n(\omega))_{n\in\mathbb{N}}$ defined by $P_n(\omega):=\{\nu^n_k(\omega)\wedge t:k\in\mathbb{N}\}$ satisfies the assumptions of Proposition~\ref{prop:approximationcost}. Hence, Proposition~\ref{prop:approximationcost} together with $C(\varphi,[\sigma\wedge\cdot,\tau\wedge \cdot])=0$ on $\{\sigma=\tau\}$ implies that the sequence of predictable processes
	\begin{align*}
		\sum_{k=1}^{\infty}(\overline{S}_{\nu^n_{k-1}}-S_{\nu^n_{k-1}})(\varphi_{\nu^n_{k}\wedge \cdot}-\varphi_{\nu^n_{k-1}\wedge \cdot})^++\sum_{k=1}^{\infty}(S_{\nu^n_{k-1}}-\underline{S}_{\nu^n_{k-1}})(\varphi_{\nu^n_{k}\wedge \cdot}-\varphi_{\nu^n_{k-1}\wedge \cdot})^-,\quad n\in\mathbb{N}
	\end{align*}
	converges pointwise to
	$C(\vp,[\sigma\wedge \cdot,\tau\wedge \cdot])$, which yields the assertion.
\end{proof} 

\begin{proof}[Proof of Proposition~\ref{1.1.2020.2}]
	In the following, we can and do assume with no loss of generality that $\sigma$ and $\tau$ are $[0,T]$-valued stopping times. In addition, by Proposition~\ref{prop:ExistenceCostTerm}, we have $\mathrm{Var}_\sigma^\tau(\varphi)<\infty$ a.s. and thus w.l.o.g. also for all paths.
	This implies that the paths of $\varphi$ are l\`agl\`ad on $\auf\sigma,\tau\zu$.
	
	\emph{Step 1.} We start by constructing the sequence $(\varphi^n)_{n\in\mathbb{N}}$.
	Therefore, we define
	\beam\label{003.1.2020.1}
	T_0^n:=\sigma,\quad T^n_k:=\inf\{t\in(T^n_{k-1},\tau]:  \vert\varphi_t-\varphi_{T^n_{k-1}+}\vert\geq 1/n\},\quad k\in\mathbb{N},
	\eeam
	which are obviously stopping times.
	In addition, we have $T^n_{k-1}<T^n_k$ on $\{T^n_{k-1}<\infty\}$ and $\#\{k:T^n_k(\omega)\leq\tau\}<\infty$ for all 
	$\omega\in\Omega$ as $\mathrm{Var}_\sigma^\tau(\varphi)<\infty$. We have to distinguish between a portfolio adjustment at $T^n_k$ and at $T^n_k+$. For this, we define further 
	stopping times:
	\begin{align*}\pi^n_0:=\sigma, \quad 
		\pi^n_k:=(T^n_k)_{\{\vert\varphi_{T^n_k}-\varphi_{T^n_{k-1}+}\vert\geq 1/n\}},\quad k\in\mathbb{N}
	\end{align*}
	and note that $\pi^n_k$ is a predictable stopping time for all $k\in\bbn$. Indeed, for $k\geq 1$ we have \begin{align*}
		\llbracket \pi^n_k\rrbracket =\llbracket 0, T^n_k\rrbracket \cap\{(\omega,t): Y_t(\omega)\geq 1/n\}\in\mathcal{P}
	\end{align*}
	since the process $Y_t:=\vert\varphi_t-\varphi_{T^n_{k-1}+}\vert\mathbbm{1}_{\rrbracket T^n_{k-1},\tau\rrbracket}$ is a predictable.
	Hence, we may define $(\varphi^n)_{n\in\mathbb{N}}$ by
	\begin{align*}
		\varphi^n:=\sum_{k=0}^{\infty}\left(\varphi_{\pi^n_k}\mathbbm{1}_{\llbracket \pi^n_k\rrbracket}+\varphi_{T^n_k+}\mathbbm{1}_{\rrbracket T^n_k, T^n_{k+1}\rrbracket\setminus \llbracket \pi^n_{k+1}\rrbracket}\right)
	\end{align*}
	which satisfies $\varphi^n_\sigma=\varphi_\sigma$ and $\varphi^n\mathbbm{1}_{\rrbracket \sigma,\tau\rrbracket}$ is predictable and, consequently, almost simple. In addition, the definition ensures $\vert\varphi-\varphi^n\vert\leq 1/n$ on $\llbracket \sigma,\tau\rrbracket$.

	\emph{Step 2:} Let us show that $\sup_{t\in[\sigma,\tau]}\vert\mathrm{Var}_\sigma^t(\varphi)-\mathrm{Var}_\sigma^t(\varphi^n)\vert\to 0$ pointwise. Let $\omega\in\Omega$ and $\varepsilon>0$ be fixed. We take a partition $P=\{t_0,\dots,t_m\}$ s.t. $\mathrm{Var}_\sigma^\tau(\varphi(\omega))\leq \sum_{i=1}^{m}\vert \varphi_{t_i}(\omega)-\varphi_{t_{i-1}}(\omega)\vert+\varepsilon.$ This yields 
	\begin{align}\label{eq:ApproximationProof1}
		\mathrm{Var}_\sigma^{t}(\varphi(\omega))\leq\sum_{i=1}^{m}\vert \varphi_{t_i\wedge t}(\omega)-\varphi_{t_{i-1}\wedge t}(\omega)\vert+\varepsilon, \quad \forall t\in[\sigma(\omega),\tau(\omega)].
	\end{align} 
	
	Now, recall from Step 1 that $\varphi^n(\omega)\to\varphi(\omega)$ uniformly on $[\sigma(\omega),\tau(\omega)]$. Thus, we may choose $N\in\bbn$ large enough s.t. for all $n\geq N$ we have $\vert \varphi_t(\omega)-\varphi^n_t(\omega)\vert\leq \varepsilon/(2m)$ for all $t\in [\sigma(\omega),\tau(\omega)]$. Therefore, we get 
	\begin{align*}
		\mathrm{Var}_\sigma^t(\varphi(\omega))-\mathrm{Var}_\sigma^t(\varphi^n(\omega))&\leq \sum_{i=1}^{m}\vert \varphi_{t_i\wedge t}(\omega)-\varphi_{t_{i-1}\wedge t}(\omega)\vert+\varepsilon-\mathrm{Var}_\sigma^t(\varphi^n(\omega))\\
		&\leq \sum_{i=1}^{m}\vert \varphi^n_{t_i\wedge t}(\omega)-\varphi^n_{t_{i-1}\wedge t}(\omega)\vert+2\varepsilon-\mathrm{Var}_\sigma^t(\varphi^n(\omega))
		\leq 2\varepsilon
	\end{align*}
	for all $t\in[\sigma(\omega),\tau(\omega)]$.
	Hence, we have proven the claim as we have $\mathrm{Var}_\sigma^t(\varphi(\omega))\geq \mathrm{Var}_\sigma^t(\varphi^n(\omega))$ by construction.
	
	\emph{Step 3:} We now show that \eqref{eq:ApproximationBV} holds. We again argue path-by-path, i.e., $\omega\in\Omega$ is fixed without explicitly mentioning it.  Therefore, note that the jumps of the cost term on $[\sigma,\tau]$ are given by
	\begin{align*}
		\Delta C_{t}(\varphi)&=\lim\limits_{s\uparrow t}C(\varphi,[s,t])=(\overline{S}_{t-}-S_{t-})(\Delta\varphi_t)^++(S_{t-}-\underline{S}_{t-})(\Delta\varphi_t)^-,\quad t\in (\sigma,\tau],\\
		\Delta^+ C_t(\varphi)&=\lim\limits_{s\downarrow t}C(\varphi,[t,s])=(\overline{S}_{t}-S_{t})(\Delta^+\varphi_t)^++(S_{t}-\underline{S}_{t})(\Delta^+\varphi_t)^-,\quad t\in [\sigma,\tau). 
	\end{align*}
	In the following, given $k\in\bbn$, we use the notation $C(\varphi,(T^n_{k-1}, T^n_k]):= C(\varphi, [T^n_{k-1},T^n_k])- \Delta^+C_{T^n_{k-1}}(\varphi)$ and $ C(\varphi,(T^n_{k-1}, T^n_k)):= C(\varphi,(T^n_{k-1}, T^n_k])- \Delta C_{T^n_{k}}(\varphi)$, where it is tacitly assumed that $T^n_k\leq \tau$.  In particular, this means that for $\varphi^n$, we have
	$C(\varphi^n,(T^n_{k-1}, T^n_{k}])=(\overline{S}_{T^n_{k}-}-S_{T^n_{k}-})(\varphi^n_{T^n_{k}}-\varphi^n_{T^n_{k}-})^++(S_{T^n_{k}-}-\underline{S}_{T^n_{k}-})(\varphi^n_{T^n_{k}}-\varphi^n_{T^n_{k}-})^-$
	as $C(\varphi^n, (T^n_{k-1}, T^n_k))=0$ according to Proposition~\ref{prop:SimpleStrat}. We now want to get an estimate on \begin{align}\label{eq:estimate}
		\vert C(\varphi,(T^n_{k-1}, T^n_{k}])+\Delta^+ C_{T^n_k}(\varphi)-(C(\varphi^n,(T^n_{k-1}, T^n_{k}])+\Delta^+ C_{T^n_k}(\varphi^n))\vert
	\end{align}
	(this means that we move forward from $T_{k-1}^n+$ to $T_k^n+$ and tacitly assume $T_k^n<\tau$).
	
	\emph{Step 3.1:}
	We start by establishing a strong bound on the difference~\eqref{eq:estimate}, which only holds if the prices do not vary too much between $T^n_{k-1}$ and $T^n_k$. To formalize this, we take $\delta>0$, which will be specified later, and define $(\rho_m)_{m\geq0}$ by $\rho_0: =\sigma$ and 
	\[\rho_{m}:=\inf\{t\in(\rho_{m-1}, \tau]: \vert \overline{S}_t-S_t-(\overline{S}_{\rho_{m-1}}-S_{\rho_{m-1}})\vert>\delta \ \text{or}\ \vert S_t-\underline{S}_t-(S_{\rho_{m-1}}-\underline{S}_{\rho_{m-1}})\vert>\delta \}.\]
	We now claim that on $\{\rho_{m-1}\leq T^n_{k-1}<T^n_k< \rho_m\}$ for some $m\geq 1$, we have
	\begin{align}\label{eq:estimatestrong}
		\begin{aligned}
			&\vert C(\varphi,(T^n_{k-1}, T^n_{k}])+\Delta^+ C_{T^n_k}(\varphi)-(C(\varphi^n,(T^n_{k-1}, T^n_{k}])+\Delta^+ C_{T^n_k}(\varphi^n))\vert\\&\leq \delta \mathrm{Var}^{T^n_k+}_{T^n_{k-1}+}(\varphi) +\sup_{t\in [0,T]}(\overline{S}_t-\underline{S}_t)\left(\mathrm{Var}^{T^n_k+}_{T^n_{k-1}+}(\varphi)-\mathrm{Var}^{T^n_k+}_{T^n_{k-1}+}(\varphi^n)\right), \quad k\geq 1.
		\end{aligned}
	\end{align} 
	In order to prove this, we distinguish between two cases.
	
	\emph{Case 1:}  We start by considering the event $\{T^n_k=\pi^n_k\}$, i.e., 
	the infimum in (\ref{003.1.2020.1}) is attained and
	$\Delta^+C_{T^n_k}(\varphi)=\Delta^+C_{T^n_k}(\varphi^n)$. First, we assume that $\varphi_{T^n_k}-\varphi_{T^n_{k-1}+}\geq 0$, i.e., the strategy $\varphi$ buys (after netting buying and selling) $$a:=\varphi_{T^n_k}-\varphi_{T^n_{k-1}+}=\mathrm{Var}^{T^n_k}_{T^n_{k-1}+}(\varphi^n)\geq 0$$ stocks on $(T^n_{k-1},T^n_k]$. Now observe that $\varphi^n$ buys $a$ stocks at a cost of $\overline{S}_{T^n_k-}-S_{T^n_k-}$ and $\varphi$ buys at least $a$ stocks at different cost, which differs from $\overline{S}_{T^n_k-}-S_{T^n_k-}$ by at most $\delta$. In addition, the continuous strategy purchases $\varphi^\uparrow_{T^n_k}-\varphi^\uparrow_{T^n_{k-1}+}-a$ additional stocks and sells  $\varphi^\downarrow_{T^n_k}-\varphi^\downarrow_{T^n_{k-1}+}$ stocks on the same interval. But the cost of those trades can be estimated above by $\sup_{t\in [0,T]}(\overline{S}_t-\underline{S}_t)$. Putting these arguments together, we get  
	\begin{align*}
		&\vert C(\varphi,(T^n_{k-1}, T^n_{k}])+\Delta^+ C_{T^n_k}(\varphi)-(C(\varphi^n,(T^n_{k-1}, T^n_{k}])+\Delta^+ C_{T^n_k}(\varphi^n))\vert\\
		&\leq \delta a+ \sup_{t\in [0,T]}(\overline{S}_t-\underline{S}_t)(\varphi^\uparrow_{T^n_k}-\varphi^\uparrow_{T^n_{k-1}+}-a+\varphi^\downarrow_{T^n_k}-\varphi^\downarrow_{T^n_{k-1}+})\\
		&=\delta\mathrm{Var}^{T^n_k}_{T^n_{k-1}+}(\varphi^n)+\sup_{t\in [0,T]}(\overline{S}_t-\underline{S}_t)\left(\mathrm{Var}^{T^n_k}_{T^n_{k-1}+}(\varphi)-\mathrm{Var}^{T^n_k}_{T^n_{k-1}+}(\varphi^n)\right)\\
		&\leq \delta \mathrm{Var}^{T^n_k+}_{T^n_{k-1}+}(\varphi)+\sup_{t\in [0,T]}(\overline{S}_t-\underline{S}_t)\left(\mathrm{Var}^{T^n_k+}_{T^n_{k-1}+}(\varphi)-\mathrm{Var}^{T^n_k+}_{T^n_{k-1}+}(\varphi^n)\right), 
	\end{align*}
	where we used $\mathrm{Var}^{T^n_k}_{T^n_{k-1}+}(\varphi^n)\leq \mathrm{Var}^{T^n_k}_{T^n_{k-1}+}(\varphi)\leq \mathrm{Var}^{T^n_k+}_{T^n_{k-1}+}(\varphi)$ and $\Delta^+C_{T^n_k}(\varphi)=\Delta^+C_{T^n_k}(\varphi^n)$ on $\{T^n_k=\pi^n_k\}$. For $\varphi_{T^n_k}-\varphi_{T^n_{k-1}+}< 0$, the argument is analogue. 
	
	\emph{Case 2:} We still have to prove the claim on $\{T^n_k\neq \pi^n_k\}$. Here, we have $\Delta C_{T^n_k}(\varphi^n)=0$ and, therefore, the argument is similar to the previous case but this time with $a:=\varphi_{T^n_k+}-\varphi_{T^n_{k-1}+}$. Thus, we skip the details. 
	
	\emph{Step 3.2:} We still need a bound on \eqref{eq:estimate} if the costs vary by more than $\delta$ between $T^n_{k-1}$ and $T^n_k$. Fortunately, a weaker bound will be sufficient here.
	We now claim that, in general, we have 
	\begin{align}\nonumber
		& \vert C(\varphi,(T^n_{k-1}, T^n_{k}])+\Delta^+ C_{T^n_k}(\varphi)-(C(\varphi^n,(T^n_{k-1}, T^n_{k}])+\Delta^+ C_{T^n_k}(\varphi^n))\vert\\ \label{eq:estimateweak}
		&\leq \sup_{t\in[0,T]}(\overline{S}_t-\underline{S}_t)\left[\frac{2}{n}+(\mathrm{Var}_{T^n_{k-1+}}^{T^n_k+}(\varphi)-\mathrm{Var}_{T^n_{k-1}+}^{T^n_k+}(\varphi^n))\right].
	\end{align}
	We distinguish between the same cases as above.

	\emph{Case 1:} We first consider the event $\{T^n_k=\pi^n_k\}$. Recall that in this case we have $\Delta^+C_{T^n_k}(\varphi)=\Delta^+C_{T^n_k}(\varphi^n)$. In addition, let us assume that $\varphi_{T^n_k}-\varphi_{T^n_{k-1}+}\geq 0$. In this case, we have $\varphi_{T^n_k}-\varphi_{T^n_{k-1}+}\geq 1/n$ and  $\varphi_{T^n_k-}-\varphi_{T^n_{k-1}+}\leq 1/n$ by the Definition of $T^n_k$. This implies
	\begin{align*}
		\varphi_{T^n_{k}}-\varphi_{T^n_{k}-}=\varphi_{T^n_{k}}-\varphi_{T^n_{k-1}+}-\left( \varphi_{T^n_k-}-\varphi_{T^n_{k-1}+}\right)\geq 0,
	\end{align*} 
	i.e., both strategies buy at $T^n_k$, but possibly different amounts. Thus, we have $\Delta C_{T^n_k}(\varphi)=(\overline{S}_{T^n_k-}-S_{T^n_k-})(\varphi_{T^n_k}-\varphi_{T^n_{k}-})$ and can write
	\begin{align}\nonumber
		& \vert C(\varphi,(T^n_{k-1}, T^n_{k}])+\Delta^+ C_{T^n_k}(\varphi)-(C(\varphi^n,(T^n_{k-1}, T^n_{k}])+\Delta^+ C_{T^n_k}(\varphi^n))\vert\\ \nonumber
		&=\vert  C(\varphi, (T^n_{k-1},T^n_k)) + (\overline{S}_{T^n_k-}-S_{T^n_k-})(\varphi_{T^n_k}-\varphi_{T^n_{k}-})-(\overline{S}_{T^n_k-}-S_{T^n_k-})(\varphi_{T^n_k}-\varphi_{T^n_{k-1+}})\vert
		\\ \nonumber
		&= \vert C(\varphi, (T^n_{k-1},T^n_k))-(\overline{S}_{T^n_k-}-S_{T^n_k-})(\varphi_{T^n_k-}-\varphi_{T^n_{k-1}+})\vert.
	\end{align} 
	
	Since the costs per share are bounded by $\sup_{t\in [0,T]}(\overline{S}_t-\underline{S}_t)$, this yields
	\begin{align}\nonumber
		&\vert C(\varphi, (T^n_{k-1},T^n_k))-(\overline{S}_{T^n_k-}-S_{T^n_k-})(\varphi_{T^n_k-}-\varphi_{T^n_{k-1}+})\vert\\ \nonumber &\leq \sup_{t\in[0,T]}(\overline{S}_t-\underline{S}_t)\left[\mathrm{Var}_{T^n_{k-1+}}^{T^n_k-}(\varphi)
		+|\varphi_{T^n_k-}-\varphi_{T^n_{k-1}+}|\right]\\ \nonumber
		&\leq\sup_{t\in[0,T]}(\overline{S}_t-\underline{S}_t)\left[\frac{2}{n}+\mathrm{Var}_{T^n_{k-1+}}^{T^n_k+}(\varphi)-\mathrm{Var}_{T^n_{k-1}+}^{T^n_k+}(\varphi^n)\right],
	\end{align}
	where we use $|\varphi_{T^n_k-}-\varphi_{T^n_{k-1}+}|\leq 1/n$ per construction of $T^n_k$, 
	$\mathrm{Var}_{T^n_{k-1+}}^{T^n_k-}(\varphi)-
	|\varphi_{T^n_k-}-\varphi_{T^n_{k-1}+}|\le \mathrm{Var}_{T^n_{k-1+}}^{T^n_k}(\varphi)-\mathrm{Var}_{T^n_{k-1+}}^{T^n_k}(\varphi^n)$, and $\Delta^+C_{T^n_k}(\varphi)=\Delta^+C_{T^n_k}(\varphi)$ on $\{T^n_k=\pi^n_k\}$. 
	
	The case $\varphi_{T^n_k}-\varphi_{T^n_{k-1}+}\leq 0$ is 
	analogous.\\
	
	\emph{Case 2:} We still need to consider the event $\{T^n_k\neq \pi^n_k\}$, i.e., $\Delta C_{T^n_k}(\varphi^n)=0$.  However, as this is analogous to Case 1, we leave it to the reader.
	
	In addition, note that on $\{T^n_{k-1}\leq t<T^n_k\}$, we have $\mathrm{Var}_{T^n_{k-1}+}^t(\varphi^n)=0$ and, thus, the trivial estimate 
	\begin{align}
		\label{eq:estimateweak2}
		\begin{aligned}
			&\ \ \vert C(\varphi,(T^n_{k-1}, t])-C(\varphi^n,(T^n_{k-1}, t])\vert\\ \leq &\sup_{t\in [0,T]}(\overline{S}_t-\underline{S}_t)(\mathrm{Var}_{T^n_{k-1}+}^t(\varphi)-\mathrm{Var}_{T^n_{k-1}+}^t(\varphi^n)).
		\end{aligned}
	\end{align}
	
	\emph{Step 4:} We can now finish the proof by putting the different estimates together. Therefore, let $a(\delta):=\#\{m:\rho_m\leq \tau\}$
	and note that $a(\delta)<\infty$  (recall that $\omega\in\Omega$ is fixed). Next, note that we have $\Delta^+C_\sigma(\varphi)=\Delta^+C_\sigma(\varphi^n)$ by construction of $\varphi^n$. For $t\in[\sigma,\tau]$ let $K_n:=\#\{k:T^n_k\leq t\}$.  We get 
	\begin{align}
		\nonumber
		&\vert C(\varphi,[\sigma,t])-C(\varphi^n,[\sigma,t])\vert \\ \nonumber\leq & \sum_{k=1}^{K_n}\vert C(\varphi,(T^n_{k-1}, T^n_{k}])+\Delta^+ C_{T^n_k}(\varphi)\mathbbm{1}_{\{T^n_k<t\}}-(C(\varphi^n,(T^n_{k-1}, T^n_{k}])+\Delta^+ C_{T^n_k}(\varphi^n)\mathbbm{1}_{\{T^n_k<t\}})\vert\\ \label{eq:estimatealmost}
		&+\vert C(\varphi,(T^n_{K_n}, t])-C(\varphi^n,(T^n_{K_n}, t])\vert
	\end{align} 
	On $\{T^n_{K_n}<t\}$ we  apply the estimate~\eqref{eq:estimateweak} to all pairs $T^n_{k-1},T^n_k$ with $k=1,\dots,K_n$ s.t. there is at least one $m=1,\dots,a(\delta)$ with $T^n_{k-1}<p_m\leq T^n_{k}$, the estimate~\eqref{eq:estimateweak2} to the last interval $(T^n_{K_n}, t]$ and for all other pairs we use the stronger estimate~\eqref{eq:estimatestrong}. On $\{T^n_{K_n}=t\}$ we apply the same estimates to all pairs $T^n_{k-1},T^n_k$ with $k=1,\dots,K_n-1$. In addition,  on $\{T^n_{K_n}=\pi^n_{K_n}=t\}$ the arguments in Step~3.1, Case~1 resp. Step~3.2, Case~1 show that  $\vert C(\varphi,(T^n_{K_n-1}, T^n_{K_n}])-C(\varphi^n,(T^n_{K_n-1}, T^n_{K_n}])\vert$ is bounded from above by the RHS of \eqref{eq:estimatestrong} if there is no $m\in\{1,\dots,a(\delta)\}$ with $T^n_{K_n-1}< p_m\leq T^n_{K_n}$ or by the RHS of \eqref{eq:estimateweak} if there is. Finally, on $\{T^n_{K_n}=t, \pi^n_{K_n}=\infty\}$, we have $\mathrm{Var}_{T^n_{K_n-1}+}^t(\varphi^n)=0$ and thus $\vert C(\varphi,(T^n_{K_n-1}, T^n_{K_n}])-C(\varphi^n,(T^n_{K_n-1}, T^n_{K_n}])\vert$ is bounded from above by the RHS of \eqref{eq:estimateweak2}. Plugging all this into \eqref{eq:estimatealmost}, we get
	\begin{align}\nonumber
		&\vert C(\varphi,[\sigma,t])-C(\varphi^n,[\sigma,t])\vert\\\nonumber
		&\leq \delta \mathrm{Var}_\sigma^t(\varphi)+\sup_{t\in [0,T]}(\overline{S}_t-\underline{S}_t)\left(\mathrm{Var}_\sigma^t(\varphi)-\mathrm{Var}_\sigma^t(\varphi^n)+\frac{2a(\delta)}{n}\right)\\ \label{eq:estimatefinal}
		&\leq \delta \mathrm{Var}_\sigma^\tau(\varphi)+\sup_{t\in [0,T]}(\overline{S}_t-\underline{S}_t)\left(\sup_{t\in [\sigma,\tau]}(\mathrm{Var}_\sigma^t(\varphi)-\mathrm{Var}_\sigma^t(\varphi^n))+\frac{2a(\delta)}{n}\right)
	\end{align}
	for all  $t\in[\sigma,\tau]$. Given an $\varepsilon>0$, we first choose $\delta<\varepsilon/(2\mathrm{Var}_\sigma^\tau(\varphi))$ and, subsequently, applying Step 2 together with the fact that $a(\delta)<\infty$ and $\sup_{t\in [0,T]}(\overline{S}_t-\underline{S}_t)<\infty$ (for fixed $\omega\in\Omega$). We can choose $N\in\bbn$ s.t. 
	for $n\geq N$ the second term in \eqref{eq:estimatefinal} is smaller than $\varepsilon/2$. At last this yields $\sup_{t\in [\sigma,\tau]}\vert C(\varphi,[\sigma,t])-C(\varphi^n,[\sigma,t])\vert<\varepsilon$ for  $n\geq N$. Thus, we have established the assertion. 
\end{proof}


\begin{thebibliography}{10}
	
	
	\bibitem{balint.schweizer} B\'alint, D.\'{A}. and Schweizer, M.: 
	Properly discounted asset prices are semimartingales. 
	Math. Fin. Econ. \textbf{14}, 661--674 (2020)		
	
	\bibitem{bayraktar2018market} Bayraktar, E., Yu, X.: On the market viability under proportional transaction costs.  Math. Finance \textbf{28}, 800--838 (2018) 		
	
	\bibitem{beiglbock2014riemann} Beiglb{\"o}ck, M., Siorpaes, P.: Riemann-integration and a new proof of the Bichteler-Dellacherie theorem.  Stochastic Process. Appl. \textbf{124}, 1226--1235 (2014)
	
	\bibitem{beiglbock2011direct} Beiglb{\"o}ck, M., Schachermayer, W., Veliyev, B.: A direct proof of the Bichteler-Dellacherie theorem and connections to arbitrage.  Ann. Probab. \textbf{39}, 2424--2440 (2011)
	
	\bibitem{campi.schachermayer} Campi, L., Schachermayer, W.:
	A super-replication theorem in Kabanov’s model of transaction costs.
	Finance Stoch. \textbf{10}, 579--596 (2006) 
	
	\bibitem{Chou} Chou C.S., Meyer P.A., Stricker C.:  Sur les intégrales stochastiques de processus prévisibles non bornés. In: Azéma J., Yor M. (eds.) Séminaire de Probabilités XIV. Lect. Notes Math, vol. 784, pp. 128--139. Springer, Berlin (1980)
	
	\bibitem{cohen.elliott}
	Cohen, S.N.,Elliott, R.J.: Stochastic Calculus and Applications, 2nd edn. Birkh\"auser, New York (2015)
	
	\bibitem{CvitanicKaratzas1996} Cvitani{\'c}, J., Karatzas, I.:  Hedging and portfolio optimization under transactioncosts:  a martingale approach. Math. Finance \textbf{6}, 113--165 (1996)
	
	\bibitem{CzichowskyMuhleKarbeSchachermayer2014} Czichowsky, C., Muhle-Karbe, J., Schachermayer W.: Transaction Costs, Shadow Prices and Duality in Discrete Time. SIAM J. Finan. Math. \textbf{5}, 258--277 (2014)
	
	\bibitem{CzichowskyShachermayer2016}  Czichowsky, C., Schachermayer W.: Duality theory for portfolio optimization under proportional transaction costs. Ann. Appl. Probab. \textbf{26},  1888--1941 (2016)
	
	\bibitem{CzichowskyShachermayer2017} Czichowsky, C., Schachermayer W.: Portfolio optimization beyond semimartingales: Shadow prices and fractional Brownian motion. Ann. Appl. Probab. \textbf{27}, 1414--1451 (2017)
	
	\bibitem{CzichowskySchachermayerYang2017} Czichowsky, C., Schachermayer W., Yang, J.: Shadow prices for continuous processes. Math. Finance \textbf{27}, 623--658 (2017)
	
	\bibitem{CzichowskyPeyreSchachermayerYang2018} Czichowsky, C., Peyre, R., Schachermayer, W., Yang J.: Shadow prices, fractional Brownian motion, and portfolio optimisation under transaction costs. Finance Stoch. \textbf{22}, 161--180 (2018)		
	
	\bibitem{delbaen1994general} Delbaen, F., Schachermayer, W.: A general version of the fundamental theorem of asset pricing. Math. Ann. \textbf{300}, 463--520 (1994)  
	
	\bibitem{dellacherie.meyer.1982} Dellacherie, C., Meyer, P.: Probabilities and potential B. North-Holland, Amsterdam (1982)		
	
	\bibitem{doob1975stochastic} Doob, J. L.: Stochastic process measurability conditions.  Ann. Inst. Fourier \textbf{25}, 163--176 (1975) 
	
	\bibitem{eberlein.kallsen} Eberlein, E., Kallsen, J.: Mathematical Finance. Springer, Berlin (2019)		
	
	\bibitem{Emery} Émery, M.: Une topologie sur l'espace des semimartingales.  In: Azéma J., Yor M. (eds.) Séminaire de Probabilités XIII. Lect. Notes Math, vol. 721, pp. 260--280. Springer, Berlin (1979)		
	\bibitem{Frankova} Fraňková, D.: Regulated functions. Math. Bohem. \textbf{116}, 20--59 (1991)
	
	\bibitem{guasoni.2006}
	Guasoni, P.: No arbitrage under transaction costs, with fractional Brownian motion and beyond. 
	Math. Finance, \textbf{16}, 569--582 (2006)
	
	\bibitem{guasoni2012fundamental} Guasoni, P., L{\'e}pinette, E., R{\'a}sonyi, M.: The fundamental theorem of asset pricing under transaction costs. Finance Stoch. \textbf{16}, 741--777 (2012)
	
	\bibitem{guasoni.r.s.2008}  Guasoni, P., R\'asonyi, M., Schachermayer, W.:
	Consistent price systems and face-lifting pricing under transaction costs.
	Ann. Appl. Probab.
	\textbf{18}, 491--520 (2008)
	
	\bibitem{guasoni.r.s.2010}  Guasoni, P., R\'asonyi, M., Schachermayer, W.:
	The fundamental theorem of asset pricing for continuous processes under small transaction costs.
	Ann. Finance \textbf{6}, 157--191 (2010)
	
	\bibitem{he.wang.yan.1992} He, S., Wang, J., Yan, J.:
	Semimartingale theory and stochastic calculus. CRC Press, Boca Raton (1992)
	
	\bibitem{hildebrandt1938definitions} Hildebrandt, T.H.: Definitions of Stieltjes Integrals of the Riemann Type. Amer. Math. Monthly \textbf{45},  265--278 (1938)
	
	\bibitem{jacob} Jacod J.: Integrales stochastiques par rapport a une semimartingale vectorielle et changements de filtration. In: Azéma J., Yor M. (eds.) Séminaire de Probabilités XIV. Lect. Notes Math, vol. 784, pp. 161--172. Springer, Berlin (1980)	
	
	
	\bibitem{jacod.shiryaev} Jacod, J., Shiryaev, A.N.:	Limit Theorems of Stochastic Processes, 2nd edn. Springer, Berlin (2003) 	
	
	\bibitem{kabanov.safarian.2009} Kabanov, Y., Safarian, M.: Markets with Transaction Costs: Mathematical Theory. Springer, Berlin (2009).
	
	\bibitem{kabanov.stricker.2002} Kabanov Y., Stricker C.: Hedging of contingent claims under transaction costs. In: Sandmann K., Schönbucher P. (eds.) Advances in Finance and Stochastics. Essays in Honour of Dieter Sondermann, pp. 125--136. Springer, Berlin  (2002)		
	
	\bibitem{KallsenMuhleKarbe2011} Kallsen J., Muhle-Karbe, J.: Existence of shadow prices in finite probability spaces. Math. Methods Oper. Res., \textbf{73}, 251-–262 (2011)
	
	\bibitem{KallsenMuhleKarbe2010} Kallsen, J., Muhle-Karbe, J.: On using shadow prices in portfolio optimization with transaction costs. Ann. Appl. Probab. \textbf{20}, 1341--1358 (2010)
	
	\bibitem{kardaras.platen} Kardaras, C., Platen, E.: On the semimartingale property of discounted asset-price processes.
	Stochastic Process. Appl. \textbf{121}, 2678--2691 (2011)	
	
	\bibitem{kifer2000game} Kifer, Y.: Game options. Finance Stoch. \textbf{4}, 443--463 (2000)
	
	\bibitem{klenke} Klenke, A.: Probability Theory -- A Comprehensive Course, 2nd edn. Springer, Berlin (2014)
	
	\bibitem{lepeltier1984jeu} Lepeltier, J.-P., Maingueneau, M. A.: Le jeu de Dynkin en th{\'e}orie g{\'e}n{\'e}rale sans l'hypoth{\`e}se de Mokobodski.  Stochastics \textbf{13}, 25--44 (1984)  
	
	\bibitem{mikosch2000stochastic} Mikosch, T., Norvai{\v{s}}a, R.: Stochastic integral equations without probability.  Bernoulli \textbf{6}, 401--434 (2000) 
	
	\bibitem{neveu1975discrete} Neveu, J.: Discrete-parameter martingales. North-Holland, Amsterdam (1975)
	
	\bibitem{protter2005stochastic} Protter, P.E.: Stochastic Integration and Differential Equations, 2nd edn. Springer, Berlin (2004)
	
	\bibitem{rao1969quasi} Rao, K.M.: Quasi-martingales. J. Math. Scand. \textbf{24}, 79--92 (1969) 
	\bibitem{schachermayer2004fundamental}Schachermayer W.: The fundamental theorem of asset pricing under proportional transaction costs in finite discrete time. Math. Finance \textbf{14}, 19--48 (2004)
\end{thebibliography}
\end{document}